\documentclass[12pt]{article}
\usepackage{amsmath}
\usepackage{amssymb}
\usepackage{amsthm}
\usepackage{latexsym}
\usepackage{cite}
\usepackage{psfrag}
\usepackage{epsfig}
\usepackage{graphicx}
\usepackage[latin1]{inputenc}
\usepackage{url}
\usepackage{hyperref}
\usepackage{color}

\newcommand{\F}{\mathbb{F}}

\newtheorem{theorem}{Theorem}[section]

\newtheorem{proposition}[theorem]{Proposition}
\newtheorem{conjecture}[theorem]{Conjecture}
\newtheorem{observation}[theorem]{Observation}
\theoremstyle{definition}
\newtheorem{definition}[theorem]{Definition}
\newtheorem{remark}[theorem]{Remark}
\newtheorem{example}[theorem]{Example}

\makeatletter \@addtoreset{equation}{section} \makeatother

\begin{document}

\title{\textbf{Tables, bounds and graphics of short linear codes with covering
radius 3 and codimension 4 and 5}
\thanks{This work has been carried out using computing resources of the
federal collective usage center Complex for Simulation and Data
Processing for Mega-science Facilities at NRC Kurchatov Institute,
http://ckp.nrcki.ru/.}
\date{}
\author{Daniele Bartoli\footnote{The research of  D. Bartoli, S. Marcugini, and F.~Pambianco was
 supported in part by the Italian
National Group for Algebraic and Geometric Structures and their Applications (GNSAGA - INDAM) and by
University of Perugia (Project: Curve, codici e configurazioni di punti, Base Research
Fund 2018).} \\
{\footnotesize Dipartimento di Matematica e Informatica,
Universit\`{a}
degli Studi di Perugia, }\\
{\footnotesize Via Vanvitelli~1, Perugia, 06123, Italy. E-mail:
daniele.bartoli@unipg.it}
\and Alexander A. Davydov\footnote{The research of A.A.~Davydov was done at IITP RAS and supported by the Russian Government (Contract No 14.W03.31.0019).} \\
{\footnotesize Institute for Information Transmission Problems
(Kharkevich
institute), Russian Academy of Sciences}\\
{\footnotesize Bol'shoi Karetnyi per. 19, Moscow,
127051, Russian Federation. E-mail: adav@iitp.ru}
\and Stefano Marcugini$^\dag$ and Fernanda Pambianco$^\dag$  \\
{\footnotesize Dipartimento di Matematica e Informatica,
Universit\`{a}
degli Studi di Perugia, }\\
{\footnotesize Via Vanvitelli~1, Perugia, 06123, Italy. E-mail:
\{stefano.marcugini,fernanda.pambianco\}@unipg.it}}}
\maketitle
\textbf{Abstract.} The \textbf{\emph{length function}} $\ell_q(r,R)$ is the smallest
length of a $ q $-ary linear code of codimension (redundancy) $r$ and covering radius $R$.
The {\mathversion{bold}$d$}-\textbf{\emph{length function}} $\ell_q(r,R,d)$ is the smallest
length of a $ q $-ary linear code with codimension (redundancy) $r$, covering radius $R$, and minimum distance~$d$.

By computer search in wide regions of $q$, we obtained following short codes of covering radius $R=3$: $[n,n-4,5]_q3$ quasi-perfect MDS codes, $[n,n-5,5]_q3$ quasi-perfect Almost MDS codes, and  $[n,n-5,3]_q3$ codes. In computer search, we use the step-by-step leximatrix and inverse leximatrix algorithms to obtain parity check matrices of codes. These algorithms are versions of the recursive $g$-parity check matrix algorithm for greedy codes. The new codes imply the following new upper bounds (called \emph{\textbf{lexi-bounds}}) on the length function and the $d$-length function:
\begin{align*}
&\ell_q(4,3)\le\ell_q(4,3,5)<2.8\sqrt[3]{\ln q}\cdot q^{(4-3)/3}=2.8\sqrt[3]{\ln q}\cdot\sqrt[3]{q}=2.8\sqrt[3]{q\ln q}~\text{for}~11\le q\le7057;\\
&\ell_q(5,3)\le\ell_q(5,3,5)<3\sqrt[3]{\ln q}\cdot q^{(5-3)/3}=3\sqrt[3]{\ln q}\cdot\sqrt[3]{q^2}=3\sqrt[3]{q^2\ln q}~~\text{ for }~37\le q\le839.
\end{align*}
Moreover, we  improve the lexi-bounds, applying randomized greedy algorithms, and show that
 \begin{align*}
 &\ell_q(4,3)\le \ell_q(4,3,5)< 2.61\sqrt[3]{q\ln q}~\text{ if }~13\le q\le4373; \smallskip\\
 &\ell_q(4,3)\le \ell_q(4,3,5)< 2.65\sqrt[3]{q\ln q}~\text{ if }~4373<q\le7057;\\
 &\ell_q(5,3)<2.785\sqrt[3]{q^2\ln q}~\text{ if }~11\le q\le401; \smallskip\\
 &\ell_q(5,3)\le\ell_q(5,3,5)<2.884\sqrt[3]{q^2\ln q}~\text{ if }~401<q\le839.
\end{align*}
The general form of the new bounds is
\begin{align*}
\ell_q(r,3)< c_\ell\sqrt[3]{\ln q}\cdot q^{(r-3)/3},~~c_\ell\text{ is a constant independent of } q,~~r=4,5\neq 3t.
\end{align*}

The codes, obtained in this paper by  leximatrix and inverse leximatrix algorithms, provide the following new upper bounds (called \emph{\textbf{density lexi-bounds}}) on the \textbf{\emph{smallest covering density}} $\mu_q(r,R)$
of a $ q $-ary linear code of codimension $r$ and covering radius~$R$:
\begin{align*}
&\mu_q(4,3)<3.3\cdot\ln q~~\text{ for }~11\le q\le7057;\\
&\mu_q(5,3)<4.2\cdot\ln q~~\text{ for }~37\le q\le839.
\end{align*}
In the general form, we have
\begin{align*}
\mu_q(r,3)<c_\mu\cdot\ln q,~~c_\mu\text{ is a constant independent of } q,~~r=4,5.
\end{align*}
The new bounds on the length function, the $d$-length function and covering density hold for the field basis $q$ of an arbitrary structure, including $q\neq (q^{\prime})^3$ where $q'$ is a prime power.



\section{Introduction}
\subsection{Covering codes. The length function. The $d$-length function}
Let $\F_{q}$ be the Galois field with $q$ elements. Let $F_{q}^{\,n}$ be the $n$-dimensional vector space over$~\F_{q}.$ Denote by $[n,n-r]_{q}$ a $q$-ary
linear code of length $n$ and codimension (redundancy)~$r$, that is, a
subspace of $F_{q}^{\,n}$ of dimension $n-r.$  The sphere of radius $R$ with center $c$ in $F_{q}^{\,n}$
is the set $\{v:v\in F_{q}^{\,n},$ $d(v,c)\leq R\}$ where $d(v,c)$ is the Hamming distance between the vectors $v$ and $c$.

\begin{definition} \label{Def1_CoverRad}
 A linear $[n,n-r]_{q}$ code has \emph{covering radius} $R$ and is denoted as an $[n,n-r]_{q}R$ code if any of the following equivalent properties holds:
\begin{description}
  \item[(i)] The value $R$
 is the least integer such that the space $F_{q}^{\,n}$ is covered by
the spheres of radius $R$ centered at the codewords.
  \item[(ii)]
Every column of $F_{q}^{\,r}$ is equal to a linear combination of at most $R$ columns
of a parity check matrix of the code, and $R$ is the smallest integer with
this property.
\end{description}
\end{definition}

Let an $[n,n-r,d]_{q}R$ code be an $[n,n-r]_{q}R$ code of minimum distance $d$. For an introduction to coverings of vector Hamming
spaces over finite fields, see \cite{CKMS-1985,Handbook-coverings,CHLS-bookCovCod,DGMP-AMC,GrahSlo-1985}.

The covering density $\mu$  of an $[n, n-r]_qR$-code is
defined as the ratio of the total volume of all $q^{n-r}$ spheres of radius $R$ centered at the codewords to the volume $q^n$ of the space $F_{q}^{\,n}$. By
Definition \ref{Def1_CoverRad}(i), we have $\mu\ge1$. In the other words,
\begin{align}\label{eq1_density}
   \mu =\left(q^{n-r}\sum_{i=0}^{R}(q-1)^i\binom{n}{i}\right)\frac{1}{q^{n}}=\frac{1}{q^{r}}\sum_{i=0}^{R}(q-1)^i\binom{n}{i}\ge1.
\end{align}
The covering quality of a code is better if its covering density is smaller.
For fixed $q,r$, and $R$, the covering density of an $[n, n-r]_{q}R$ code decreases with decreasing~$n$.

Codes investigated from the point view of the covering quality are usually called \emph{covering codes} \cite{CHLS-bookCovCod}; see an online bibliography  \cite{LobstBibl}, works \cite{BDGMP-R2R3CC_2019,Handbook-coverings,Dav-1995,DGMP_ACCT2008,DGMP_Petersb2008,DGMP-AMC,DMP-DESI2019,DavOst-IEEE2001,DMP-Redund2019,DavOst-DESI2010,EtzStorm2016,%
Giul2013Survey,GrahSlo-1985,Klein-Stor,LandSt,Janwa,Struik,CKMS-1985}, and the references therein.

\emph{This paper is devoted to non-binary covering codes with radius $R=3$.}

Note that for relatively small $q>2$ many results are given in \cite{Dav-1995,DGMP-AMC,DavOst-IEEE2001,DavOst-DESI2010} and the references therein.

\begin{definition}
\begin{description}
  \item[(i)]
\cite{Handbook-coverings,CHLS-bookCovCod}
\emph{\textbf{The length function}} $\ell_q(r,R)$ is the smallest
length of a $ q $-ary linear code of codimension (redundancy) $r$ and covering radius $R$.
\item[(ii)]
\cite{BGP-2015,DMP-Redund2019} The {\mathversion{bold}$d$}-\emph{\textbf{length function}} $\ell_q(r,R,d)$ is the smallest
length of a $ q $-ary linear code with codimension (redundancy) $r$, covering radius $R$, and minimum distance~$d$.
\end{description}
\end{definition}

Obviously,
$$\ell_q(r,R)\le\ell_q(r,R,d).$$

Let $\mathcal{A}_{R,q}$ denote a family of covering codes, in
which the covering radius $R$ and the size $q$ of the
underlying Galois field are fixed, while the code length tends
to infinity. The construction of families with small asymptotic
covering densities is a classical problem in the area of covering
codes.

 In \cite{DGMP-AMC}, infinite sets of families $\mathcal{A}_{R,q}$,
where $R$ is fixed but $q$ ranges over an infinite set of prime
powers are considered. It is shown that for the upper limit $\mu
_{q}^{\ast}(R,\mathcal{A}_{R,q})$ of the covering density of
$\mathcal{A}_{R,q}$, the best possibility is
\begin{equation}
\mu_{q}^{\ast }(R,\mathcal{A}_{R,q})= O(q).\label{ab}
\end{equation}
Moreover, in \cite{DGMP-AMC}, for any
covering radius $R\ge 2$, it is proposed the construction
of {\em optimal} infinite sets of families $\mathcal{A}_{R,q}$
 such that (\ref{ab}) holds. For this,
in \cite{DGMP-AMC},  the following results are obtained:
In first, it is shown that for a given $R$, to obtain optimal
infinite sets of families it is enough to construct $R$
infinite families
$\mathcal{A}_{R,q}^{(0)},\mathcal{A}_{R,q}^{(1)},\ldots
,\mathcal{A}_{R,q}^{(R-1)}$ such that, for all $u\geq u_{0}$,
the family $\mathcal{A}_{R,q}^{(\gamma )}$ contains codes of
codimension $r_{u}=Ru+\gamma$ and length $n_{q}^{(\gamma
)}(r_{u})$ where $$n_{q}^{(\gamma )}(r)=O\left(q^{(r-R)/R}\right)$$
and $u_{0}$ is a constant. Then, the
needed families $\mathcal{A}_{R,q}^{(\gamma )}$ are constructed for any
covering radius $R\geq 2$, with $q$ ranging over the (infinite)
set of $R$-th powers. A result of independent interest is that
in each of these families $\mathcal{A}_{R,q}^{(\gamma )}$, the
\emph{lower limit of the covering density is bounded from above by a
constant independent of $q$}.

So, infinite families of $[n,n-r]_qR$ codes of length
\begin{align}\label{eq1_desired_length}
 n=O\left(q^{(r-R)/R}\right)
\end{align}
 play an important role in covering code theory.

Infinite families of covering $[n,n-r]_qR$ codes of length \eqref{eq1_desired_length} are known for the following cases (see \cite{Handbook-coverings,CHLS-bookCovCod,Dav-1995,DGMP_ACCT2008,DGMP-AMC,DMP-DESI2019,DGMP_Petersb2008} and the references therein):
\begin{align*}
    &r=tR, ~\text{the field basis }q \text{ has an arbitrary structure including $q\neq (q^{\prime})^R$}&&\text{\cite{Dav-1995,DGMP_ACCT2008,DGMP-AMC,DMP-DESI2019,DGMP_Petersb2008}} \\
    &&&\text{\cite{DavOst-IEEE2001,DavOst-DESI2010}};\\
   & r\ne tR, ~~q=(q')^R&&\text{\cite{DGMP_ACCT2008,DGMP-AMC,DGMP_Petersb2008}};\\
   & R=sR^*,~~ r=Rt+s, ~~q=(q')^{R^*}&&\text{\cite{DGMP_ACCT2008,DGMP-AMC,DMP-DESI2019}}.
\end{align*}
  Here $t$ and $s$ are integers, $q'$ is a prime power.

  In the general case,  \emph{for arbitrary $r,R,q$, the problem to construct infinite families of $[n,n-r]_qR$ codes of length \eqref{eq1_desired_length} is open}.

In the last decades, upper bounds on $\ell_q(r,R)$ and $\ell_q(r,R,d)$ have been intensively investigated, see \cite{CHLS-bookCovCod,Handbook-coverings,Dav-1995,DFMP-ConjCap,DGMP_ACCT2008,DGMP_Petersb2008,DGMP-AMC,DavOst-IEEE2001,DavOst-DESI2010,EtzStorm2016,Giul2013Survey,%
Klein-Stor,LandSt,LobstBibl,DMP-DESI2019,DMP-Redund2019,BDGMP-R2R3CC_2019,Janwa,GrahSlo-1985,Struik,CKMS-1985,BGP-2015} and the references therein.

The \emph{goal of this paper} is to obtain new \emph{upper bounds on the length functions $\ell_q(4,3)$, $\ell_q(5,3)$ and the $d$-length functions $\ell_q(4,3,5)$, $\ell_q(5,3,5)$  where codimension $r\ne tR$ and the field basis $q$ has an arbitrary structure, including $q\ne (q')^3$ with $q'$ is a prime power.} It is an open problem.

\subsection{Saturating sets in projective spaces. Complete arcs}

Let $\mathrm{PG}(N,q)$ be the $N$-dimensional projective space over the field $\F_q$; see \cite{Hirs,HirsSt-old,HirsStor-2001} for an introduction to the projective spaces over finite fields, see also \cite{DGMP-AMC,Dav-1995,EtzStorm2016,HirsSt-old,Klein-Stor,LandSt} for connections  between coding theory and Galois geometries.

Effective methods to obtain upper bounds on $\ell_q(r,R)$ are connected with saturating sets in $\mathrm{PG}(N,q)$.

\begin{definition}\label{def1_usual satur}
 A point set $\mathcal{S}\subseteq\mathrm{PG}(N,q)$ is
$\rho$-\emph{saturating} if any of the following equivalent properties holds:

\textbf{(i)} For any point $A$ of\/ $\mathrm{PG}(N,q)\setminus \mathcal{S}$
there exist $\rho+1$ points in $\mathcal{S}$ generating a subspace of $\mathrm{PG}(N,q)$ containing
$A$, and $\rho$ is the smallest value with this property.

\textbf{(ii)} Every
point $A\in\mathrm{PG}(N,q)$ (in homogeneous coordinates) can be written as a linear combination of at most $\rho+1$ points of $\mathcal{S}$, and $\rho$ is the smallest value with this property (cf. Definition \ref{Def1_CoverRad}(ii)).
\end{definition}

Saturating sets are considered, for instance, in \cite{BDFKMP-PIT2014,BDFKMP_ArcFOPArXiv2015,BDFKMP_ArcComputJG2016,DFMP-ConjCap,Handbook-coverings,%
DGMP_ACCT2008,DGMP_Petersb2008,DGMP-AMC,DMP-JCTA2003,Dav-1995,DMP-Redund2019,DavOst-IEEE2001,EtzStorm2016,Janwa,Giul2013Survey,Klein-Stor,LandSt,%
Ughi,DMP-DESI2019,BDGMP-R2R3CC_2019}. In the literature, saturating sets are also called ``saturated
sets'', ``spanning sets'', ``dense sets''.

Let $s_q(N,\rho)$ be \emph{the smallest size of a $\rho$-saturating set} in $\mathrm{PG}(N,q)$.

If $q$-ary positions of a column of an $r\times n$ parity check matrix of an $[n,n-r]_qR$ code are treated as homogeneous coordinates of a point in $\mathrm{PG}(r-1,q)$ then this parity check matrix defines an $(R-1)$-saturating set of size $n$ in $\mathrm{PG}(r-1,q)$  and vice versa \cite{Dav-1995,DGMP_ACCT2008,DGMP_Petersb2008,Giul2013Survey,DGMP-AMC,DMP-Redund2019,EtzStorm2016,Janwa,Klein-Stor,LandSt,DMP-DESI2019,BDGMP-R2R3CC_2019}.

 So, there is a \emph{one-to-one correspondence between $[n,n-r]_qR$ codes and $(R-1)$-saturating sets in $\mathrm{PG}(r-1,q)$}. Therefore,
\begin{align*}
    \ell_q(r,R)=s_q(r-1,R-1),
\end{align*}
in particular,  $\ell_q(4,3)=s_q(3,2)$, $\ell_q(5,3)=s_q(4,2)$.

\emph{Complete arcs} in  $\mathrm{PG}(N,q)$ are an important class of saturating sets. An $n$-arc in $\mathrm{PG}(N,q)$ with $n>N + 1$ is a set of $n$ points such that no $N + 1$ points belong to the same hyperplane of $\mathrm{PG}(N,q)$. An $n$-arc of $\mathrm{PG}(N,q)$ is complete
if it is not contained in an $(n+1)$-arc of $\mathrm{PG}(N,q)$. A complete arc in $\mathrm{PG}(N,q)$ is an $(N-1)$-saturating set.  Points (in homogeneous coordinates) of a complete $n$-arc in $\mathrm{PG}(N,q)$, treated as columns, form a parity check matrix of an $[n,n-(N+1),N+2]_qN$ maximum distance separable (MDS) code \cite{BDGMP-R2R3CC_2019,BGP-2015,DGMP-AMC,DMP-Redund2019,EtzStorm2016,Giul2013Survey,Hirs,HirsSt-old,HirsStor-2001,Klein-Stor,LandSt}. If $N=2,3$ these codes are quasi-perfect.

Let $s_q^\text{arc}(N)$ be \emph{the smallest size of a complete arc} in $\mathrm{PG}(N,q)$. By above,
\begin{align*}
    \ell_q(R+1,R)=s_q(R,R-1)\le\ell_q(R+1,R,R+2)= s_q^\text{arc}(R).
\end{align*}

The known results about upper bounds on $\ell_q(R+1,R,R+2)$ and $s_q^\text{arc}(R)$, $R\ge2$, can be found in \cite{BDFKMP-PIT2014,BDFKMP_ArcFOPArXiv2015,BDFKMP_ArcComputJG2016,BDGMP-R2R3CC_2019,BGP-2015}, see also the references therein.

\subsection{Covering codes with radius 3}

For the \emph{field basis $q$ of an arbitrary structure}, infinite families of covering $[n,n-r]_q3$ codes of
length \eqref{eq1_desired_length} are known only for $r=tR=3t$ \cite{DGMP-AMC,DavOst-IEEE2001}. In particular,
the following parameters $n,r$ are obtained by algebraic constructions \cite[Sect.\,5,\,eq.\,(5.2)]{DGMP-AMC}, \cite[Th.\,12]{DavOst-IEEE2001}:
\begin{align*}
&n= 3q^{(r-3)/3}+q^{(r-6)/3},~r=3t\ge6,~r\ne9,~ q\ge5, \mbox{ and } r=9,~ q=16,~q\ge23;\displaybreak[3]\\
&n= 3q^{(r-3)/3}+2q^{(r-6)/3}+1,~r=9,~q=7,8,11,13,17,19;\displaybreak[3]\\
&n= 3q^{(r-3)/3}+2q^{(r-6)/3}+2,~r=9,~q=5,9.
\end{align*}

If $r=3t+1$ or  $r=3t+2$, infinite families of covering codes of
length \eqref{eq1_desired_length} are known only when $q=(q')^3$, where $q'$ is a prime power \cite{DGMP_ACCT2008,DGMP_Petersb2008,DGMP-AMC,Giul2013Survey}. In particular,
 $[n.n-r,3]_q3$ codes with the following parameters $n$ and $r$  are obtained by algebraic constructions, see \cite{DGMP_ACCT2008,DGMP_Petersb2008,Giul2013Survey}, \cite[Sect.\,5, eqs.\,(5.3),(5.4)]{DGMP-AMC}:
\begin{align}
&n=\left(4+\frac{4}{\sqrt[3]{q}}\right)q^{(r-3)/3},~r=3t+1\ge4,~q=(q')^3\ge64;\displaybreak[3]\label{eq1_rad3_q3}\\
&n=\left(9-\frac{8}{\sqrt[3]{q}}+\frac{4}{\sqrt[3]{q^2}}\right)q^{(r-3)/3},~r=3t+2\ge5,~q=(q')^3\ge27.\notag
\end{align}

For the field basis $q$ of an arbitrary structure, including $q\ne(q')^3$, in the literature, computer results are given for $[n,n-4]_q3$ codes with $q\le563$ \cite[Tab. 1]{DavOst-DESI2010} and $q\le6229$ \cite{BDGMP-R2R3CC_2019}, and also for $[n,n-5]_q3$ codes with $q\le43$ \cite[Tab. 1]{DGMP_Petersb2008}, \cite[Tab.~2]{DavOst-DESI2010} and $q\le761$ \cite{BDGMP-R2R3CC_2019}.

The  results of this paper are used in \cite{DMP-Redund2019} and presented in XVI International Symposium
``Problems of Redundancy in Information and Control Systems'' (Redundancy 2019), Moscow, Russia, 21--25 October 2019.

The paper is organized as follows.  In Section \ref{sec_main_res}, we give the main results of this paper. In Section \ref{sec_gr-matr-alg}, a leximatrix algorithm to obtain parity check matrices of covering codes is described.  In Sections \ref{sec_ell_q(4,3)} and \ref{sec_ell_q(5,3)}, upper bounds on the length functions $\ell_q(4,3)$, $\ell_q(5,3)$ and the $d$-length functions $\ell_q(4,3,5)$, $\ell_q(5,3,5)$, based on leximatrix codes, are given. In Section~\ref{sec_inv_gr-matr-alg}, an inverse leximatrix algorithm to obtain parity check matrices of covering codes is considered and invleximatrix codes are obtained with the help of this algorithm. In Section~\ref{sec_rand_gr_alg} randomized greedy algorithms to obtain parity check matrices of covering codes are presented; new upper bounds improving the bounds of the previous sections are obtained. In Conclusion, the results of this paper are briefly analyzed; some tasks for investigation of the leximatrix algorithm are formulated. In Appendix, tables with sizes of codes obtained in this paper are given.

\section{The main results}\label{sec_main_res}
In this paper, by computer search, we obtain new results for $[n,n-4,5]_q3$ quasi-perfect MDS codes with $q\le7057$  and $[n,n-5,5]_q3$ quasi-perfect Almost MDS codes with $q\le839$. Also, we obtain $[n,n-5,3]_q3$ codes for $q\le401$. This gives  upper bounds on $\ell_q(4,3)$, $\ell_q(4,3,5)$, $\ell_q(5,3)$,  and $\ell_q(5,3,5)$ for a set of values $q$ greater  than in  \cite{DGMP_Petersb2008,DavOst-DESI2010,BDGMP-R2R3CC_2019}. New bounds are better than known ones.

The following Theorem \ref{th2_res_rad3} is based on the results of Sections \ref{sec_gr-matr-alg}--\ref{sec_rand_gr_alg}, see Propositions~\ref{prop7_lexi_r=4}, \ref{prop7_lexi_r=5}, \ref{prop5_invlexi_r=4}, \ref{prop6_mixt_r=4}, and \ref{prop6_mixt_r=5}.

\begin{theorem}\label{th2_res_rad3}
 For the length function $\ell_q(r,3)$, the $d$-length function $\ell_q(r,3,5)$, the smallest size $s_q(r-1,2)$ of a $2$-saturating set in the projective space $\mathrm{PG}(r-1,q)$, and the smallest size $s_q^\text{arc}(3)$ of a complete arc in $\mathrm{PG}(3,q)$, the following upper bounds hold:
 \begin{description}
   \item[(1)] Upper bounds provided by $[n,n-r,5]_q3$ leximatrix and invleximatrix quasi-perfect codes (\textbf{lexi-bounds}).
   \begin{align*}
 &\textbf{\emph{(i)}}~~   \ell_q(4,3)=s_q(3,2)\le \ell_q(4,3,5)=s_q^\emph{\text{arc}}(3)<2.8\sqrt[3]{\ln q}\cdot q^{(4-3)/3}=2.8\sqrt[3]{\ln q}\cdot\sqrt[3]{q}\displaybreak[3]\\
 &\hspace{7cm}=2.8\sqrt[3]{q\ln q}~\text{  for } r=4,~~11\le q\le7057;\displaybreak[3]\\
 &\emph{\textbf{(ii)}}~~\ell_q(5,3)=s_q(4,2)\le\ell_q(5,3,5)<3\sqrt[3]{\ln q}\cdot q^{(5-3)/3}=3\sqrt[3]{\ln q}\cdot\sqrt[3]{q^2}=3\sqrt[3]{q^2\ln q}\displaybreak[3]\\
 &\hspace{7cm}\text{ for }r=5,~~37\le q\le839.
   \end{align*}
   \item[(2)] Upper bounds provided by  $[n,n-4,5]_q3$  quasi-perfect MDS  codes obtained with the help of the leximatrix, invleximatrix and d-Rand-Greedy algorithms.
   \begin{align*}
 &    \ell_q(4,3)=s_q(3,2)\le \ell_q(4,3,5)=s_q^\emph{\text{arc}}(3)<\left\{
    \begin{array}{ccc}
    2.61\sqrt[3]{q\ln q} & \text{if} & 13\le q\le4373 \smallskip\\
    2.65\sqrt[3]{q\ln q} & \text{if} & 4373<q\le7057
    \end{array}
    \right..
\end{align*}
   \item[(3)] Upper bounds provided by  $[n,n-5]_q3$  codes obtained with the help of the leximatrix and Rand-Greedy algorithms.
       \begin{align*}
 &    \ell_q(5,3)=s_q(4,2)< 2.785\sqrt[3]{q^2\ln q} && \text{if } 11\le q\le401;\\
& \ell_q(5,3)=s_q(4,2)\le\ell_q(5,3,5)< 2.884\sqrt[3]{q^2\ln q} && \text{if } 401<q\le839.
\end{align*}
 \end{description}
\end{theorem}

 Note that,  for $r\neq 3t$ and the field basis $q$ of an arbitrary structure, including $q\neq (q^{\prime})^3$ where $q'$ is a prime power,  the new bounds of Theorem \ref{th2_res_rad3} have the form
\begin{align*}
   \ell_q(r,3)< c_\ell\sqrt[3]{\ln q}\cdot q^{(r-3)/3},~~ c_\ell \text{ is a constant independent of }q,~~r=4,5.
\end{align*}
The constants $c_\ell$ in the new bounds are smaller than in the paper \cite{BDGMP-R2R3CC_2019}.

  Our results, in particular, figures and observations in Sections \ref{sec_ell_q(4,3)} and  \ref{sec_ell_q(5,3)}, comparison of leximatrix and invleximatrix codes in Table 3, improvements of the lexi-bounds in Section~\ref{sec_rand_gr_alg}, allow us to conjecture the following.
\begin{conjecture}\label{conj2}
 For the length function $\ell_q(r,3)$, the $d$-length function $\ell_q(r,3,5)$, the smallest size $s_q(r-1,2)$ of a $2$-saturating set in the projective space $\mathrm{PG}(r-1,q)$, and the smallest size $s_q^\text{arc}(3)$ of a complete arc in $\mathrm{PG}(3,q)$, the following upper bounds (\textbf{lexi-bounds}) hold:
\begin{align*}
&\emph{\textbf{(i)}}\quad \ell_q(4,3)=s_q(3,2)\le\ell_q(4,3,5)= s_q^\emph{\text{arc}}(3)<2.8\sqrt[3]{\ln q}\cdot q^{(4-3)/3}=2.8\sqrt[3]{\ln q}\cdot\sqrt[3]{q}\\
&\hspace{8cm}=2.8\sqrt[3]{q\ln q}~\text{ for $r=4$ and \textbf{all} }q\ge11;\displaybreak[3] \\
&\emph{\textbf{(ii)}}\quad\ell_q(5,3)=s_q(4,2)\le\ell_q(5,3,5)<3\sqrt[3]{\ln q}\cdot q^{(5-3)/3}=3\sqrt[3]{\ln q}\cdot\sqrt[3]{q^2}=3\sqrt[3]{q^2\ln q}\\
&\hspace{8cm}\text{ for $r=5$ and  \textbf{all} } q\ge37.
\end{align*}
\end{conjecture}

Let $\mu_q(r,R)$ be the \emph{\textbf{smallest covering density of a $q$-ary linear code of codimension (redundancy) $r$ and covering radius~$R$}}.

The following Theorem \ref{th2_density} is based on the results of Sections \ref{sec_gr-matr-alg}--\ref{sec_rand_gr_alg}, see Propositions~\ref{prop4_dens_lexi_r=4} and \ref{prop5_density_lexi_r=5}.

\begin{theorem}\label{th2_density}
The $[n,n-r,5]_q3$ leximatrix and invleximatrix quasi-perfect codes, providing lexi-bounds of Theorem \emph{\ref{th2_res_rad3}(1)}, give also the following upper bounds on $\mu_q(r,3)$ (\textbf{\emph{density lexi-bounds}}):
\begin{align*}
&\mu_q(4,3)<3.3\cdot\ln q~~\text{ for }~11\le q\le7057;\\
&\mu_q(5,3)<4.2\cdot\ln q~~\text{ for }~37\le q\le839.
\end{align*}
\end{theorem}

 Note that,  for $r\neq 3t$ and the field basis $q$ of an arbitrary structure, including $q\neq (q^{\prime})^3$ where $q'$ is a prime power,  the new bounds of Theorem \ref{th2_density} have the form
\begin{align*}
   \mu_q(r,3)< c_\mu\cdot\ln q,~~ c_\mu \text{ is a constant independent of }q,~~r=4,5.
\end{align*}
\section{A leximatrix algorithm to obtain parity check matrices of covering codes}\label{sec_gr-matr-alg}

The following is a \emph{version of the recursive $g$-parity check matrix algorithm for greedy codes}, see e.g. \cite[p. 25]{BrPlessGreedy}, \cite{MonroePlessGreedCod}, \cite[Section 7]{PlessHandb}.

Let $\F_{q}=\{0,1,\ldots ,q-1\}$ be the Galois field with $q$ elements.

If $q$ is prime, the elements of
 $\F_{q}$ are treated as integers modulo $q$.

If $q=p^{m}$ with $p$ prime and $m\ge2$, the elements of $\F_{p^{m}}$ are
represented by integers as follows: $\F_{p^{m}}=\F_{q}=\{0,1=\alpha^{0},2=\alpha^{1},\ldots,u=\alpha^{u-1},\ldots,q-1=\alpha^{q-2}\},$
where $\alpha$ is a root of a primitive polynomial of $\F_{p^{m}}$.

For a $q$-ary code of codimension $r$, covering radius $R$, and minimum distance $d=R+2$,  we construct  a parity check matrix from nonzero columns $h_i$ of the form
\begin{align}\label{eq2_column}
h_{i}=(x_{1}^{(i)},x_{2}^{(i)},\ldots,x_{r}^{(i)})^{tr},~ x_{u}^{(i)}\in
\F_{q},
\end{align}
 where the first (leftmost) non-zero element is 1; \emph{tr} is the sign of transposition. The number of distinct columns is $(q^r-1)/(q-1)$. We order the columns in the list as
 \begin{align}\label{eq2_list}
    h_1,h_2,\ldots,h_{(q^r-1)/(q-1)}.
 \end{align}
For $h_i$ we put
\begin{align}\label{eq2_number_i}
    i=\sum\limits_{u=1}^r x_{u}^{(i)}q^{r-u}.
\end{align}
 The columns of the list are candidates to be included in the parity check matrix.

By the above arguments connected with the formula for $i$ and the order of the columns,  a column $h_i$ is treated as its number $i$ in our list written in the $q$-ary scale of notation. The considered \textbf{\emph{order of the columns}} is \textbf{\emph{lexicographical}}.

 The first column of the list should be included into the matrix. Then step-by-step,
one takes the next column from the list which cannot be represented as a linear
combination of at most $R$ columns already chosen. The process ends when no new column may
be included into the matrix. The obtained matrix $H_n$ is a parity check matrix of an $[n,n-r,R+2]_qR$ code.

The obtained parity check matrix is called the \textbf{\emph{parity check leximatrix}} or the \textbf{\emph{leximatrix}} for short. We call a \textbf{\emph{leximatrix code}} the corresponding code.

\textbf{For prime $q$}, the following holds: \textbf{length $n$ of a leximatrix code and the form of the leximatrix $H_n$ depend on $q$, $r$, and $R$ only}.  No other factors affect code length and structure. Actually, assume that after some step a current matrix is obtained. At the next step we should remove from our current list all columns that are linear combination of $R$ or less columns of the current matrix. For prime $q$ and the given $r$ and $R$, the result of removing is unequivocal; hence, the next column is taken uniquely.

For non-prime $q$, the length  $n$ of a
leximatrix code depends on $q$ and on the primitive polynomial of the field.
In this paper, we use primitive polynomials that are
created by the program system MAGMA \cite{MAGMA} by default, see Table A. In any case,
the choice of the polynomial changes the leximatrix code length
unessentially.

\begin{table*}[ht]
\noindent\textbf{Table A.} Primitive polynomials used in this paper for leximatrix $[n,n-r,5]_q3$ quasi-perfect codes with non-prime $q$\medskip\\
\begin{tabular}{@{}c@{\,}|@{\,}c@{\,}||@{\,}c@{\,}|@{\,}c@{\,}||@{\,}c@{\,}|@{\,}c@{}}\hline
 $q=p^{m}$&primitive&$q=p^{m}$&primitive&$q=p^{m}$& primitive \\
&polynomial&&polynomial&&polynomial\\ \hline
 $4=2^{2}$& $x^{2}+x+1$&$8=2^{3}$& $x^{3}+x+1$&$9=3^{2}$& $x^{2}+2x+2$\\
$16=2^{4}$& $x^{4}+x^{3}+1$&$25=5^{2}$& $x^{2}+x+2$&$27=3^{3}$& $x^{3}+2x^{2}+x+1$ \\
$32=2^{5}$&$x^{5}+x^{3}+1$&$49=7^{2}$&$x^{2}+x+3$&$64=2^{6}$&$x^{6}+x^{4}+x^{3}+1$\\
$81=3^{4}$&$x^{4}+x+2$&$121=11^{2}$&$x^{2}+4x+2$&$125=5^{3}$&$x^{3}+3x+2$\\
$128=2^{7}$&$x^{7}+x+1$&$169=13^{2}$&$x^{2}+x+2$&$243=3^{5}$&$x^{5}+2x+1$\\
$256=2^{8}$&$x^{8}+x^{4}+x^{3}+$&$289=17^{2}$&$x^{2}+x+3$&$343=7^{3}$&$x^{3}+3x+2$\\
&$x^{2}+1$&&&&\\
$361=19^{2}$&$x^{2}+x+2$&$512=2^{9}$&$x^{9}+x^{4}+1$&$529=23^{2}$&$x^{2}+2x+5$\\
$625=5^{4}$&$x^{4}+x^{2}+2x+2$&$729=3^{6}$&$x^{6}+x+2$&$841=29^{2}$&$x^{2}+24x+2$\\
$961=31^{2}$&$x^{2}+29x+3$&$1024=2^{10}$&$x^{10}+x^{6}+x^{5}+$&$1331=11^{3}$&$x^{3}+2x+9$\\
&&&$x^{3}+x^{2}+x+1$&&\\
$1369=37^{2}$&$x^{2}+33x+2$&$1681=41^{2}$&$x^{2}+38x+6$ &$1849=43^{2}$&$x^{2}+x+3$\\
$2048=2^{11}$&$x^{11}+x^{2}+1$&$2187=3^{7}$&$x^{7}+x^{2}+2x+1$&$2197=13^{3}$&$x^{3}+x^{2}+7$\\
$2209=47^{2}$&$x^{2}+x+13$&$2401=7^{4}$&$x^{4}+5x^{2}+4x+3$&$2809=53^{2}$&$x^{2}+49x+2$\\
$3125=5^{5}$&$x^{5}+4x+2$&$3481=59^{2}$&$x^{2}+58x+2$&$3721=61^{2}$&$x^{2}+60x+2$\\$4096=2^{12}$&$x^{12}+x^{8}+x^{2}+$&$4489=67^{2}$&$x^{2}+63x+2$&$4913=17^{3}$&$x^{3}+x+14$\\
&$x+1$&&&&\\
$5041=71^{2}$&$x^{2}+69x+7$&$5329=73^{2}$&$x^{2}+70x+5$&$6241=79^2$&$x^2+78x+3$\\
$6561=3^8$&$x^8+2x^5+x^4+$&$6859=19^3$&$x^3+4x+17$&$6889=83^2$&$x^2+82x+2$\\
&$2x^2+2x+2$&&&\\
\hline
\end{tabular}
\end{table*}

By the leximatrix algorithm, if $R=1$, we obtain the $q$-ary Hamming code. If $R=2$, we obtain a quasi-perfect $[n,n-r,4]_q2$ code; for $r=3$, such code is an MDS code and corresponds to a complete arc in $\mathrm{PG}(2,q)$. If $R=3$, we obtain a quasi-perfect $[n,n-r,5]_q3$ code; for $r=4$, such code is an MDS code and corresponds to a complete arc in $\mathrm{PG}(3,q)$; for $r=5$, it is an Almost MDS code.

Let $n^\text{L}_q(r,R)$ be \textbf{length of the $q$-ary leximatrix code of codimension $r$ and covering radius~$R$}.

It is assumed that for a non-prime field $\F_q$, one uses the primitive polynomial created by the program system MAGMA \cite{MAGMA} by default; in particular, for non-prime $q\le6889$, the polynomial from Table A should be taken.

We represent length $n^{\text{L}}_q(r,R)$ of an $[n^{\text{L}}_q(r,R),n^\text{L}_q(r,R)-r,R+2]_qR$ leximatrix code in the form
\begin{align}\label{eq2_cLex}
    n^{\text{L}}_q(r,R)=c^{\text{L}}_q(r,R)\sqrt[R]{\ln q}\cdot q^{(r-R)/R},
\end{align}
where $c^{\text{L}}_q(r,R)$ is a coefficient. The coefficient $c^{\text{L}}_q(r,R)$ and length $n^{\text{L}}_q(r,R)$ are entirely given by $r,R,q$ (if $q$ is prime) or by $r,R,q$, and the primitive polynomial of $\F_q$ (if $q$ is non-prime).

\begin{remark}
In the literature on the projective geometry, the columns are considered as points in homogeneous coordinates; the algorithm, described above, is called an ``algorithm with fixed order of points'' (FOP) \cite{BDFKMP_ArcFOPArXiv2015,BDFKMP_ArcComputJG2016,DMP-Redund2019}.
\end{remark}

Let $\mu^\text{L}_q(r,R)$ be \textbf{covering density of the $q$-ary leximatrix code of codimension $r$ and covering radius~$R$}.

By \eqref{eq1_density}. we have
\begin{align}\label{eq3_Ldensity}
   \mu^\text{L}_q(r,R) =\frac{1}{q^{r}}\sum_{i=0}^{R}(q-1)^i\binom{n^{\text{L}}_q(r,R)}{i}\ge1.
\end{align}
We represent covering density $ \mu^\text{L}_q(r,R)$ of an $[n^{\text{L}}_q(r,R),n^\text{L}_q(r,R)-r,R+2]_qR$ leximatrix code in the form
\begin{align}\label{eq3_mLex}
   \mu^\text{L}_q(r,R)=m^{\text{L}}_q(r,R)\cdot\ln q,
\end{align}
where $m^{\text{L}}_q(r,R)$ is a coefficient. The coefficient $m^{\text{L}}_q(r,R)$ and density $\mu^\text{L}_q(r,R)$ are entirely given by $r,R,q$ (if $q$ is prime) or by $r,R,q$, and the primitive polynomial of $\F_q$ (if $q$ is non-prime).

\section{Upper bounds on the length function $\ell_q(4,3)$  and $d$-length function $\ell_q(4,3,5)$ based on leximatrix codes}\label{sec_ell_q(4,3)}

The following properties of the leximatrix algorithm are useful for implementation.

\begin{proposition}\label{prop7_init_part_lexi}
Let $q$ be a prime. Then the $v$-th column of the parity check leximatrix of an
$[n,n-4,5]_q3$ code is the same for all $q\ge q_{0}(v)$ where
$q_{0}(v)$ is large enough.
\end{proposition}
\begin{proof}
Let $H_j=[h^{(1)},h^{(2)},\ldots,h^{(j)}]$ be the matrix obtained in the $j$-th step of the\linebreak leximatrix algorithm. Here $h^{(v)}$ is  a  column of the matrix. A column from the list, not included in $H_j$, is covered by $H_j$ if it can be represented as a linear combination of at most 3  columns of $H_j$. Suppose that $h^{(j)}=h_s$, where $h_s$ is the $s$-th column in the lexicographical list of candidates.
A column $Q=h_u \not \in H_j$ is the next chosen column, if and only if all the columns $h_m$
with $m \in [s+1,u-1]$ are covered by $H_j$. This means that, for
any $m \in [s+1,u-1]$, at least one of the determinants $\det(h^{(v_1)},h^{(v_2)},h^{(v_3)},h_m)$, with $h^{(v_1)},h^{(v_2)},h^{(v_3)} \in
H_j$, is equal to zero modulo $q$. This can happen only in two cases:
\begin{description}
\item{$\bullet$}   $\det(h^{(v_1)},h^{(v_2)},h^{(v_3)},h_m)= 0$, we  say that  $h_m$ is
    ``absolutely" covered by $H_j$;
  \item{$\bullet$}  $\det(h^{(v_1)},h^{(v_2)},h^{(v_3)},h_m)= B\neq0$, but $B \equiv 0 \bmod
    q$.
\end{description}
For $q$ large enough, $q$ does not divide any
of the possible values of $B$ and then, at least for
$j$ relatively small,  the columns covered are just the absolutely
covered columns. Therefore, when $q$ is large enough the
leximatrices  share a certain number of columns.
\end{proof}

The values of $q_{0}(v)$ can be found with the help of
calculations based on the proof of Proposition
\ref{prop7_init_part_lexi}. Also, we can directly consider
leximatrices  for a convenient region of $q$.

\begin{example}\label{ex7}
Values of $q_{0}(v)$, $v\le 20$, together with columns $(x_1^{(v)},x_2^{(v)},x_3^{(v)},x_4^{(v)})^{tr}$, are
given in Table B.
So, for all prime $q\ge233$ (resp. $q\ge1321$) the first 14 (resp. 20) columns of a parity check leximatrix
of an $[n,n-4,5]_q3$ quasi-perfect MDS leximatrix code are as in Table B.

\begin{table*}[htbp]
\noindent\textbf{Table B.} The first 20 columns of the parity check leximatrices of
$[n,n-4,5]_q3$ quasi-perfect MDS leximatrix codes, $q$ prime
 \begin{center}
\begin{tabular}{r|rrrr|r||r|rrrr|r}\hline\noalign{\smallskip}
 $v$&$x_1^{(v)}$ & $x_2^{(v)}$ &$x_3^{(v)}$&$x_4^{(v)}$&$q_{0}(v)$&$v$&$x_1^{(v)}$ & $x_2^{(v)}$ &$x_3^{(v)}$&$x_4^{(v)}$&$q_{0}(v)$  \\
 \hline
 1& 0&0&0&1&2&    11&1&7&11&8&67  \\
 2& 0&0&1&0&2&    12&1&8&6&13&109 \\
 3& 0 &1&0&0&2&   13&1&9&13&16&199\\
4&  1&0&0&0&2&    14&1&10&12&22&233\\
5& 1 &1&1&1&2&    15&1&11&7&29&269 \\
6& 1 &2&3&4&5&    16&1&12&22&15&769 \\
7& 1 &3&2&5&11&   17&1&13&16&20&769 \\
8& 1 &4&5&3&29&   18&1&14&17&7&1283\\
9&1&5&4&2&41&     19&1&15&21&10&1283\\
10&1&6&8&9&41&    20&1&16&9&38&1321\\
\hline
 \end{tabular}
\end{center}
\end{table*}
\end{example}

\begin{proposition}\label{prop7_lexi_r=4}
\begin{description}
\item[(i)] For $q=9$, there exists a $[7,7-4,4]_93$ code of length $n=7<2.8\sqrt[3]{9\ln 9}$.
\item[(ii)]There exist $[n^{\emph{\text{L}}}_q(4,3),n^{\emph{\text{L}}}_q(4,3)-4,5]_q3$ quasi-perfect MDS \smallskip  leximatrix codes of length $n^{\emph{\text{L}}}_q(4,3)<2.8\sqrt[3]{q\ln q}$~for $q=8$ and $11\le q\le7057$.
\end{description}
\end{proposition}
\begin{proof}
\begin{description}
\item[(i)] The existence of the code is noted in \cite[Tab.\,1]{DavOst-DESI2010}, see also the references therein.
\item[(ii)]
The needed codes are obtained by computer search, using the leximatrix algorithm, Proposition \ref{prop7_init_part_lexi}, and Example \ref{ex7}.
\end{description}
\end{proof}

Proposition \ref{prop7_lexi_r=4} implies the assertions of Theorem \ref{th2_res_rad3}(1i) on the upper \textbf{\emph{lexi-bound}} on the length function $\ell_q(4,3)$  and the $d$-length function $\ell_q(4,3,5)$.

Lengths $n^{\text{L}}_q(4,3)$ of the $[n^{\text{L}}_q(4,3),n^{\text{L}}_q(4,3)-4,5]_q3$ leximatrix quasi-perfect MDS codes are collected in Table 1 (see Appendix) and presented in Figure~\ref{fig_PG3qFOPsize} by the bottom solid black curve.
The bound
\begin{align*}
n^{\text{L}}_q(4,3)<2.8\sqrt[3]{q\ln q},
\end{align*}
called the \emph{\textbf{lexi-bound}},  is shown in Figure \ref{fig_PG3qFOPsize} by the top dashed red curve.
\begin{figure}[htb]
\includegraphics[width=\textwidth]{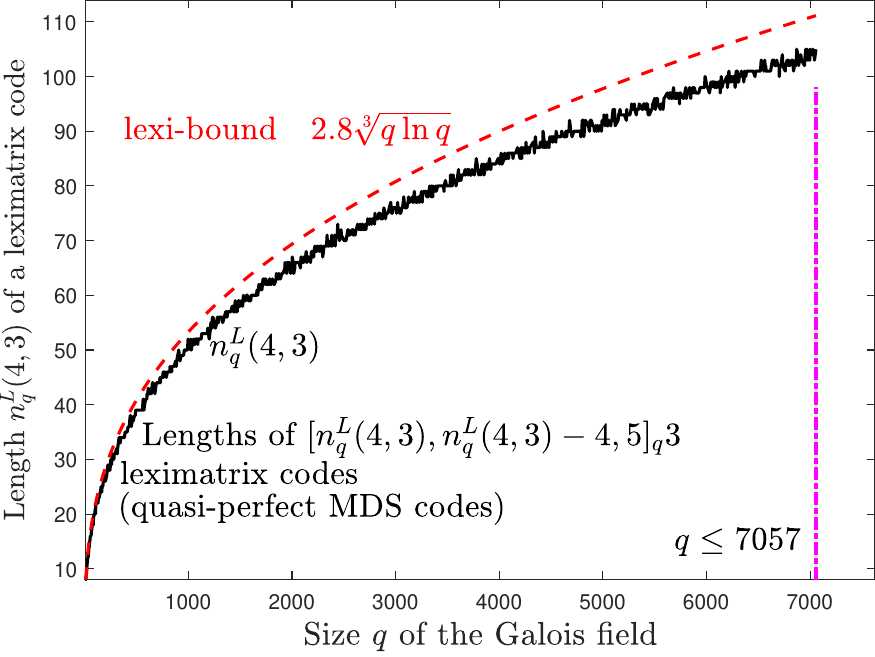}
\caption{Lengths $n^\text{L}_q(4,3)$ of the $[n^\text{L}_q(4,3),n^\text{L}_q(4,3)-4,5]_q3$ leximatrix quasi-perfect  MDS codes (\emph{bottom solid black curve}) vs the lexi-bound $2.8\sqrt[3]{q\ln q}$ (\emph{top dashed red curve});\newline $11\le q\leq 7057$. \emph{Vertical magenta line} marks region $q\le7057$}
\label{fig_PG3qFOPsize}
\end{figure}

We denote by  $\delta_q(4,3)$ the difference between the lexi-bound $2.8\sqrt[3]{q\ln q}$  and length $n^\text{L}_q(4,3)$ of the leximatrix code. Let $\delta_q^{\%}(4,3)$ be the corresponding percent difference. Thus,
\begin{align*}
&\delta_q(4,3)=2.8\sqrt[3]{q\ln q}-n^\text{L}_q(4,3);\displaybreak[3]\\
&\delta_q^{\%}(4,3)=\frac{2.8\sqrt[3]{q\ln q}-n^\text{L}_q(4,3)}{2.8\sqrt[3]{q\ln q}}100\%.
\end{align*}
The difference $\delta_q(4,3)$ and the percent difference $\delta_q^{\%}(4,3)$ are presented in Figures \ref{fig_PG3qFOPdlt} and~\ref{fig_PG3qFOPperc}.
\begin{figure}[htb]
\includegraphics[width=\textwidth]{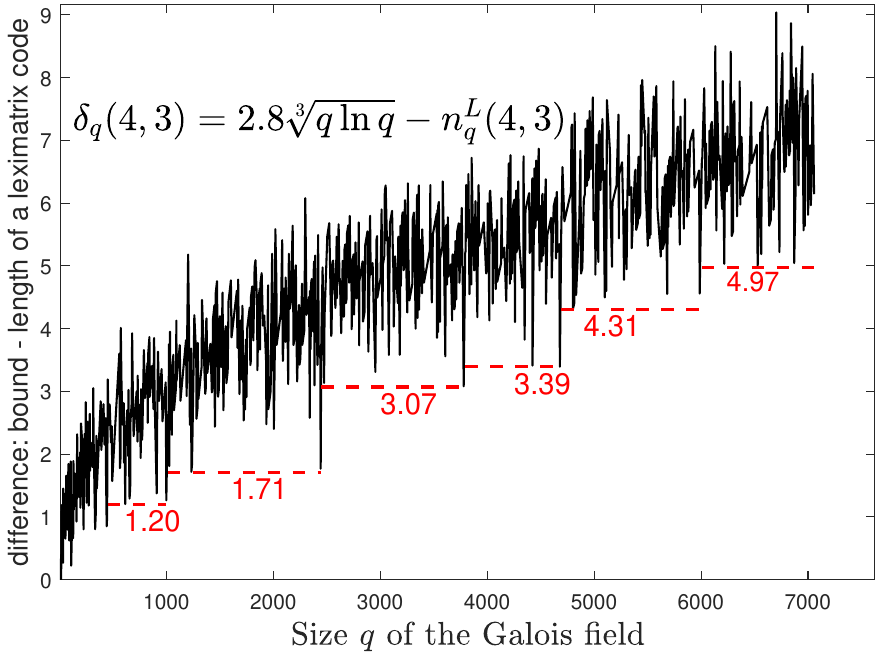}
\caption{Difference $\delta_q(4,3)$ between the lexi-bound $2.8\sqrt[3]{q\ln q}$ and length $n^\text{L}_q(4,3)$ of an $[n^\text{L}_q(4,3),n^\text{L}_q(4,3)-4,5]_q3$ leximatrix code; $11\le q\leq 7057$}
\label{fig_PG3qFOPdlt}
\end{figure}

\begin{figure}[htb]
\includegraphics{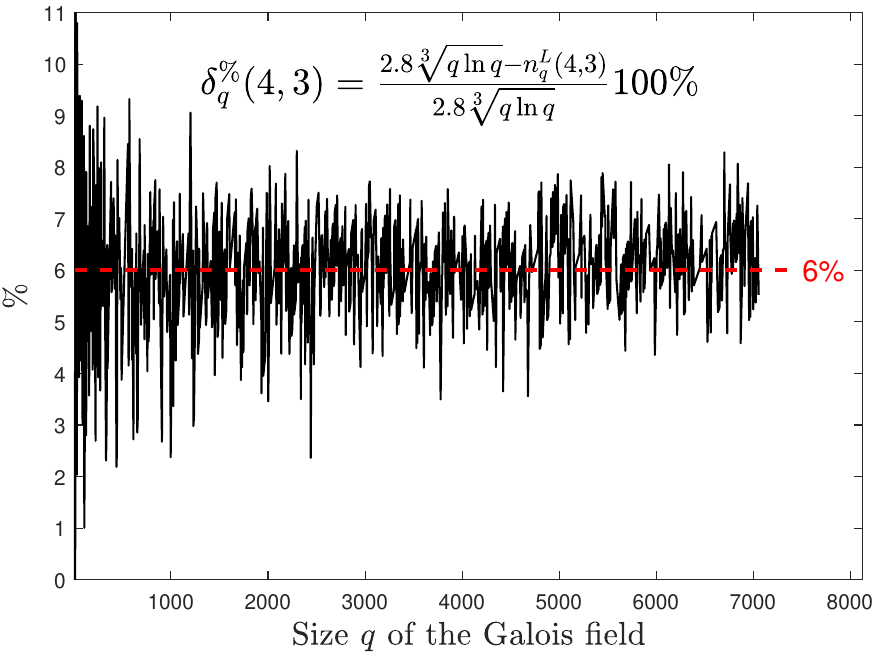}
\caption{Percent difference $\delta_q^{\%}(4,3)=\frac{2.8\sqrt[3]{q\ln q}-n^\text{L}_q(4,3)}{2.8\sqrt[3]{q\ln q}}100\%$ between the lexi-bound $2.8\sqrt[3]{q\ln q}$ and length $n^\text{L}_q(4,3)$ of an $[n^\text{L}_q(4,3),n^\text{L}_q(4,3)-4,5]_q3$ leximatrix code;\newline $11\le q\leq 7057$}
\label{fig_PG3qFOPperc}
\end{figure}

By \eqref{eq2_cLex}, we represent length of an $[n^{L}_q(4,3),n^{L}_q(4,3)-4,5]_q3$ leximatrix code in the form
\begin{align}\label{eq3_coef}
    n^{\text{L}}_q(4,3)=c^{\text{L}}_q(4,3)\sqrt[3]{q\ln q},
\end{align}
where $c^{\text{L}}_q(4,3)$ is a coefficient entirely given by $q$ (if $q$ is prime) or by $q$ and the primitive polynomial of the field $\F_q$ (if $q$ is non-prime). The coefficients $c_q^\text{L}(4,3)=n^{\text{L}}_q(4,3)/\sqrt[3]{q\ln q}$ are shown in Figure~\ref{fig_PG3qFOPcoef}.
\begin{figure}[htb]
\includegraphics[width=\textwidth]{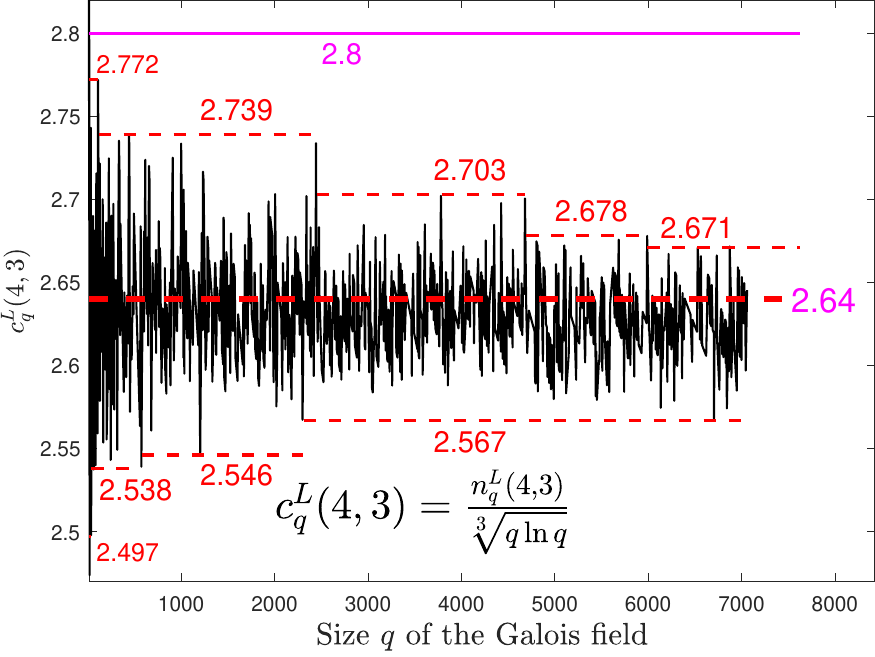}
\caption{Coefficients  $c_q^\text{L}(4,3)=n^\text{L}_q(4,3)/\sqrt[3]{q\ln q}$ for the $[n^\text{L}_q(4,3),n^\text{L}_q(4,3)-4,5]_q3$ leximatrix quasi-perfect  MDS codes;  $11\le q\leq 7057$}
\label{fig_PG3qFOPcoef}
\end{figure}

\begin{observation}\label{observ3}
\begin{description}
\item[(i)] The difference $\delta_q(4,3)$  tends to increase when $q$ grows, see Figures $\ref{fig_PG3qFOPsize}$ and $\ref{fig_PG3qFOPdlt}$.

\item[(ii)] The percent difference $\delta_q^{\%}(4,3)$ oscillates around the horizontal line $y=6\%$. When $q$ increases, the oscillation amplitude tends to decrease, see Figure $\ref{fig_PG3qFOPperc}$.

\item[(iii)] Coefficients $c_q^\text{L}(4,3)$ oscillate around the horizontal line $y=2.64$ with a small amplitude. \textbf{When $q$ increases, the oscillation
amplitude tends to decrease}, see Figure~$\ref{fig_PG3qFOPcoef}$.
\end{description}
\end{observation}

Observation \ref{observ3} gives rise to Conjecture~\ref{conj2}(i) on the length function $\ell_q(4,3)$  and the $d$-length function $\ell_q(4,3,5)$.

Note that Observations \ref{observ3}(ii) and \ref{observ3}(iii) are connected with each other. Actually,
$$\delta_q^{\%}(4,3)=\frac{2.8\sqrt[3]{q\ln q}-n^\text{L}_q(4,3)}{2.8\sqrt[3]{q\ln q}}100\%=\left(1-\frac{c_q^\text{L}(4,3)}{2.8}\right)100\%.$$

\begin{remark}\label{rem3_oscil}
  It is interesting that the oscillation of the coefficients $c_q^\text{L}(4,3)$ around a horizontal line, in principle, is similar to the oscillation of the values $h^\text{L}(q)$ around a horizontal line in \cite[Fig.\,6, Observation 3.5]{BDFKMP_ArcFOPArXiv2015}, \cite[Fig.\,5, Observation 3.7]{BDFKMP_ArcComputJG2016}.

  In the papers \cite{BDFKMP_ArcFOPArXiv2015,BDFKMP_ArcComputJG2016}, small complete $t_2^L(2,q)$-arcs in the projective plane $\mathrm{PG}(2,q)$ are constructed by computer search using algorithm with fixed order of points (FOP). These arcs correspond to $[t_2^L(2,q),t_2^L(2,q)-3,4]_q2$ quasi-perfect MDS codes while the algorithm FOP is analogous to the leximatrix algorithm of Section \ref{sec_gr-matr-alg}. Moreover, the value $h^\text{L}(q)$ is defined in \cite{BDFKMP_ArcFOPArXiv2015,BDFKMP_ArcComputJG2016} as  $h^\text{L}(q)=t_2^L(2,q)/\sqrt{3q\ln q}$. So, see \eqref{eq3_coef}, the coefficients $c_q^\text{L}(4,3)$ and the values $h^\text{L}(q)$ have the similar nature. It is possible that the oscillations mentioned also have similar reasons.

  However, in the present time the \textbf{enigma of the oscillations} is incomprehensible,
\end{remark}

\begin{proposition}\label{prop4_dens_lexi_r=4}
There exist $[n^{\emph{\text{L}}}_q(4,3),n^{\emph{\text{L}}}_q(4,3)-4,5]_q3$ quasi-perfect MDS leximatrix codes of covering density  $\mu^{\text{L}}_q(4,3)<3.3\cdot\ln q$ for  $11\le q\le7057$.
\end{proposition}
\begin{proof}
The needed codes are the codes of Proposition \ref{prop7_lexi_r=4}.
\end{proof}

Proposition \ref{prop4_dens_lexi_r=4} implies the assertion of Theorem \ref{th2_density} on the upper \textbf{\emph{density lexi-bound}} on the covering density  $\mu_q(4,3)$ .

Covering densities $\mu^{\text{L}}_q(4,3)$ of the $[n^{\text{L}}_q(4,3),n^{\text{L}}_q(4,3)-4,5]_q3$ leximatrix quasi-perfect MDS codes  are presented in Figure~\ref{fig_PG3qFOPdensity} by the bottom solid black curve. The values $\mu^{\text{L}}_q(4,3)$ are obtained by \eqref{eq3_Ldensity} where lengths $n^{\text{L}}_q(4,3)$ are taken from Table 1 (see Appendix).
The bound
\begin{align*}
\mu^{\text{L}}_q(4,3)<3.3\cdot\ln q,
\end{align*}
called the \emph{\textbf{density lexi-bound}},  is shown in Figure \ref{fig_PG3qFOPdensity} by the top dashed red curve.
\begin{figure}[htb]
\includegraphics[width=\textwidth]{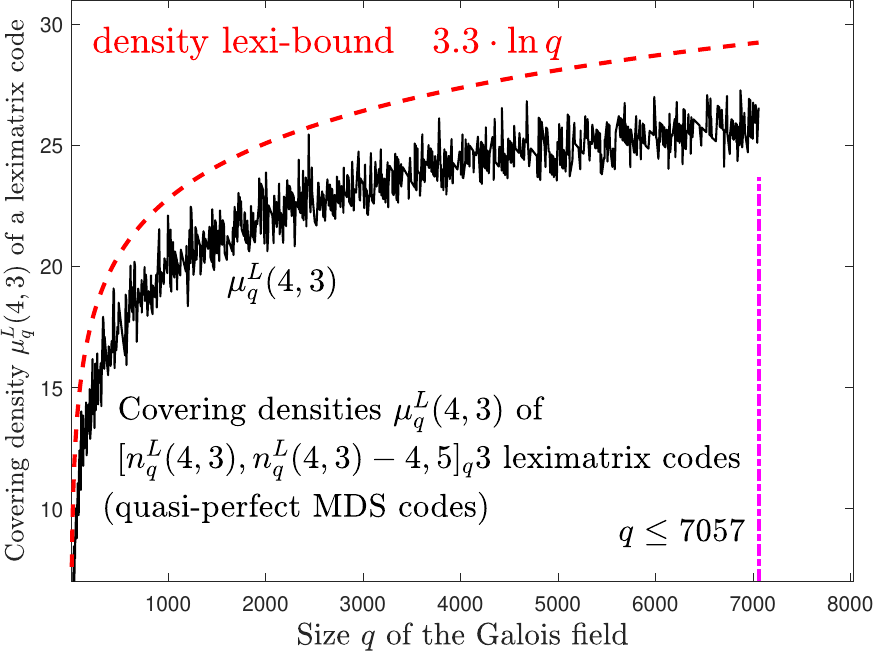}
\caption{Covering densities $\mu^\text{L}_q(4,3)$ of the $[n^\text{L}_q(4,3),n^\text{L}_q(4,3)-4,5]_q3$ leximatrix quasi-perfect  MDS codes  (\emph{bottom solid black curve}) vs the density lexi-bound $3.3\cdot\ln q$ (\emph{top dashed red curve}); $11\le q\leq 7057$. \emph{Vertical magenta line} marks region $q\le7057$}
\label{fig_PG3qFOPdensity}
\end{figure}

By \eqref{eq3_mLex}, we represent covering density of an $[n^{L}_q(4,3),n^{L}_q(4,3)-4,5]_q3$ leximatrix code in the form
\begin{align*}
    \mu^{\text{L}}_q(4,3)=m^{\text{L}}_q(4,3)\cdot\ln q,
\end{align*}
where $m^{\text{L}}_q(4,3)$ is a coefficient entirely given by $q$ (if $q$ is prime) or by $q$ and the primitive polynomial of the field $\F_q$ (if $q$ is non-prime). The coefficients $m_q^\text{L}(4,3)=\frac{\mu^{\text{L}}_q(4,3)}{\ln q}$ are shown in Figure~\ref{fig_PG3qFOPdensdiv}.

\begin{figure}[htb]
\includegraphics[width=\textwidth]{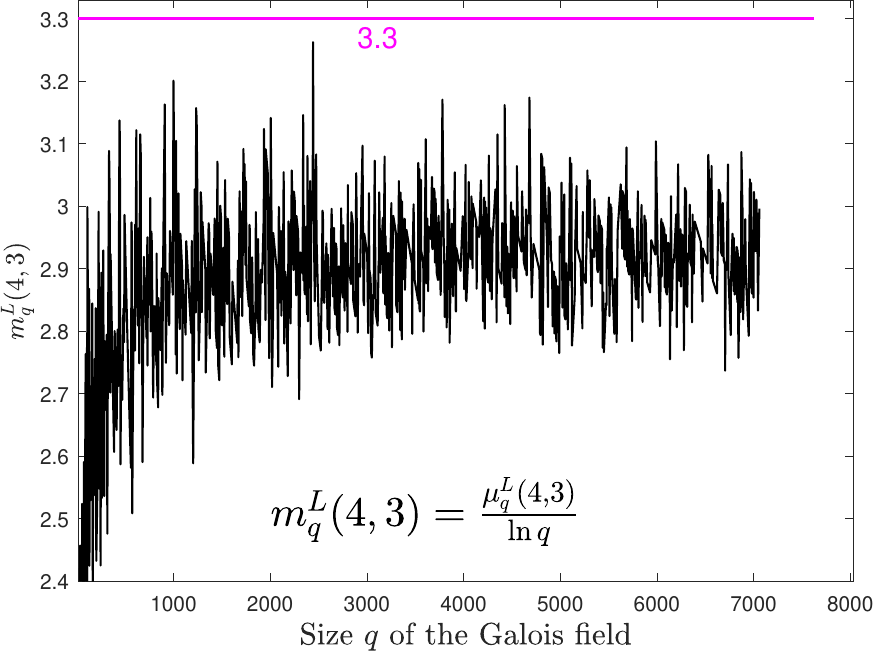}
\caption{Coefficients  $m_q^\text{L}(4,3)=\mu^\text{L}_q(4,3)/\ln q$ for covering density of the\newline
 $[n^\text{L}_q(4,3),n^\text{L}_q(4,3)-4,5]_q3$ leximatrix quasi-perfect  MDS codes;  $11\le q\leq 7057$}
\label{fig_PG3qFOPdensdiv}
\end{figure}

\section{Upper bounds on the length function $\ell_q(5,3)$ and $d$-length function $\ell_q(5,3,5)$ based on leximatrix codes}\label{sec_ell_q(5,3)}

\begin{proposition}\label{prop7_lexi_r=5}
\begin{description}
\item[\textbf{(i)}]
There exist $[n,n-5,4]_q3$ codes with $n<3\sqrt[3]{q^2\ln q}$~for $5\le q<37$.\smallskip
\item[\textbf{(ii)}]
There exist $[n^{\text{L}}_q(5,3),n^{\text{L}}_q(5,3)-5,5]_q3$ quasi-perfect Almost MDS leximatrix codes with $n^{\text{L}}_q(5,3)<3\sqrt[3]{q^2\ln q}$ for $37\le q\le839$.
\end{description}
\end{proposition}

\begin{proof}
\begin{description}
\item[(i)] The existence of the codes is noted in \cite[Tab.\,1]{DGMP_Petersb2008}, \cite[Tab.\,2]{DavOst-DESI2010}, see also the references therein.
\item[(ii)]
The needed codes are obtained by computer search, using the leximatrix algorithm.
\end{description}
\end{proof}

Proposition \ref{prop7_lexi_r=5} implies the assertions of Theorem \ref{th2_res_rad3}(1ii) on the upper \textbf{\emph{lexi-bound}} on the length function $\ell_q(5,3)$  and the $d$-length function $\ell_q(5,3,5)$.

Lengths $n^{\emph{\text{L}}}_q(5,3)$ of the $[n^{\emph{\text{L}}}_q(5,3),n^{\emph{\text{L}}}_q(5,3)-5,5]_q3$ leximatrix Almost MDS codes are collected in Table 2 (see Appendix) and presented in Figure~\ref{fig_PG4qFOPsize} by the bottom solid black curve.
The bound $3\sqrt[3]{q^2\ln q}$, called the \emph{\textbf{lexi-bound}},  is shown in Figure \ref{fig_PG4qFOPsize} by the top dashed red curve.
\begin{figure}[htb]
\includegraphics[width=\textwidth]{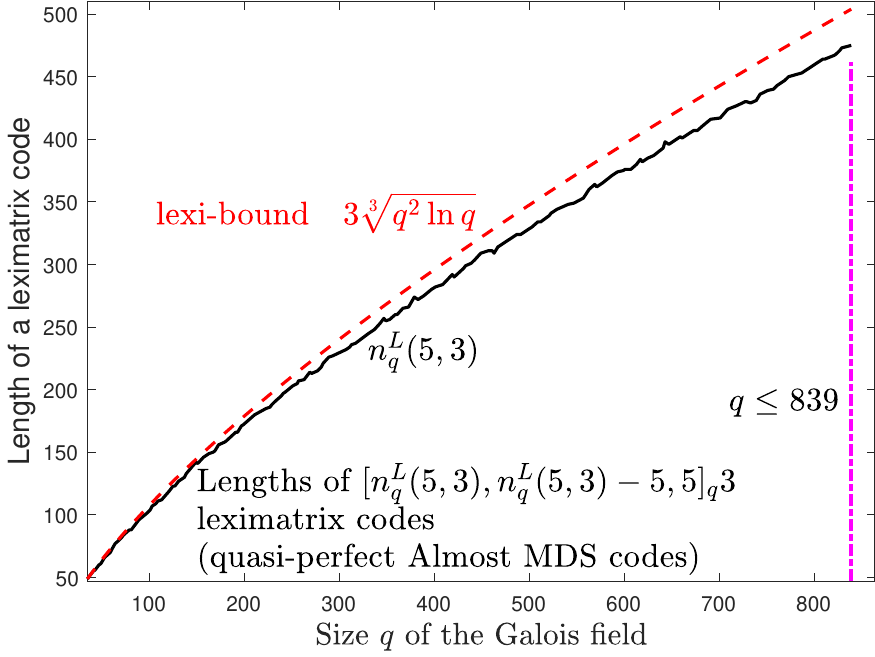}
\caption{Lengths $n_{q}^\text{L}(5,3)$ of the $[n_{q}^\text{L}(5,3),n_{q}^\text{L}(5,3)-5,5]_q3$  leximatrix quasi-perfect Almost MDS codes (\emph{bottom solid black curve}) vs the lexi-bound $3\sqrt[3]{q^2\ln q}$ (\emph{top dashed red curve}); $37\le q\le839$. \emph{Vertical magenta line} marks region $q\le839$}
\label{fig_PG4qFOPsize}
\end{figure}

We denote by  $\delta_q(5,3)$ the difference between the lexi-bound $3\sqrt[3]{q^2\ln q}$   and length $n^\text{L}_q(5,3)$ of the leximatrix code. Let $\delta_q^{\%}(5,3)$ be the corresponding percent difference. Thus,
\begin{align*}
&\delta_q(5,3)=3\sqrt[3]{q^2\ln q} -n^\text{L}_q(5,3);\\
&\delta_q^{\%}(5,3)=\frac{3\sqrt[3]{q^2\ln q}-n^\text{L}_q(5,3)}{3\sqrt[3]{q^2\ln q}}100\%.
\end{align*}
The difference $\delta_q(5,3)$ and the percent difference $\delta_q^{\%}(5,3)$ are presented in Figures \ref{fig_PG4qFOPdlt} and~\ref{fig_PG4qFOPperc}.
\begin{figure}[htb]
\includegraphics[width=\textwidth]{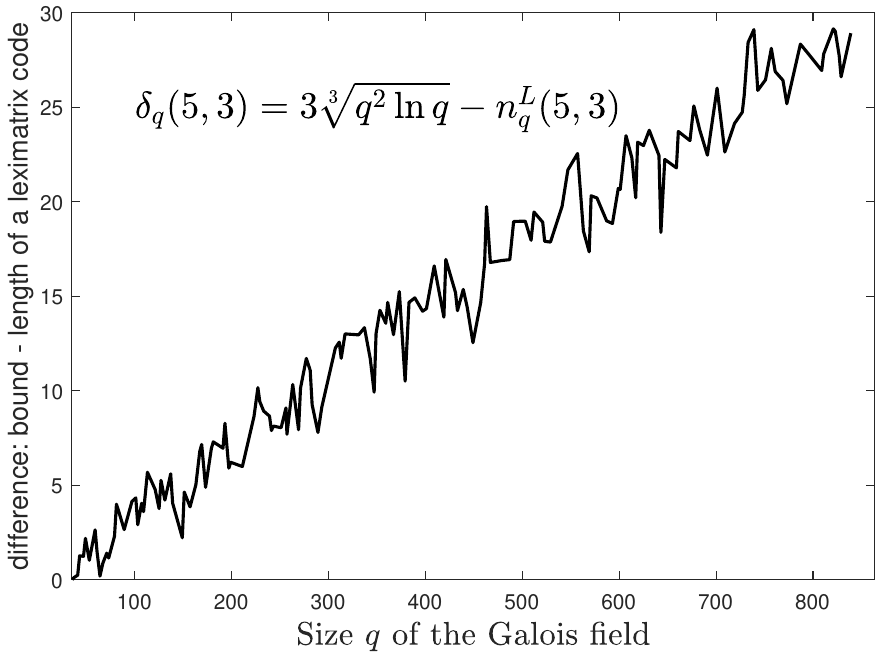}
\caption{Difference $\delta_q(5,3)$ between the lexi-bound $3\sqrt[3]{q^2\ln q}$ and length $n^\text{L}_q(5,3)$ of an $[n^\text{L}_q(5,3),n^\text{L}_q(5,3)-5,5]_q3$ leximatrix code; $37\le q\le839$}
\label{fig_PG4qFOPdlt}
\end{figure}

\begin{figure}[htb]
\includegraphics[width=\textwidth]{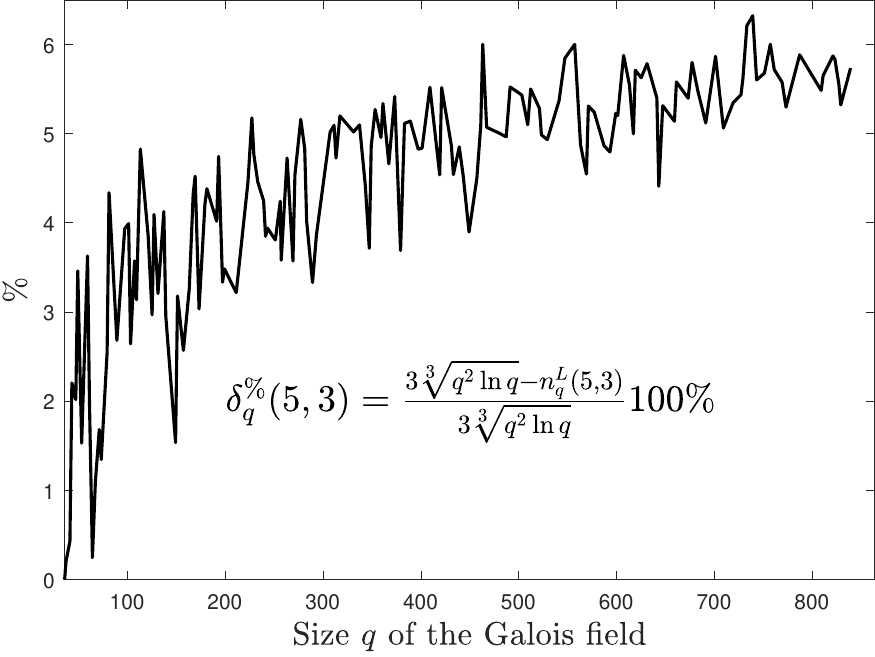}
\caption{Percent difference $\delta_q^{\%}(5,3)=\frac{3\sqrt[3]{q^2\ln q}-n^\text{L}_q(5,3)}{3\sqrt[3]{q^2\ln q}}100\%$ between the lexi-bound $3\sqrt[3]{q^2\ln q}$ and length $n^\text{L}_q(5,3)$ of an $[n^\text{L}_q(5,3),n^\text{L}_q(5,3)-5,5]_q3$ leximatrix code;\newline $37\le q\le839$}
\label{fig_PG4qFOPperc}
\end{figure}

By \eqref{eq2_cLex}, we represent length of an $[n^{L}_q(5,3),n^{L}_q(5,3)-5,5]_q3$ leximatrix code in the form
\begin{align*}
    n^{\text{L}}_q(5,3)=c^{\text{L}}_q(5,3)\sqrt[3]{q^2\ln q},
\end{align*}
where $c^{\text{L}}_q(5,3)$ is a coefficient entirely given by $q$ (if $q$ is prime) or by $q$ and the primitive polynomial of the field $\F_q$ (if $q$ is non-prime). The coefficients $c_q^\text{L}(5,3)=n^{\text{L}}_q(5,3)/\sqrt[3]{q^2\ln q}$ are shown in Figure~\ref{fig_PG4qFOPcoef}.
\begin{figure}[htb]
\includegraphics[width=\textwidth]{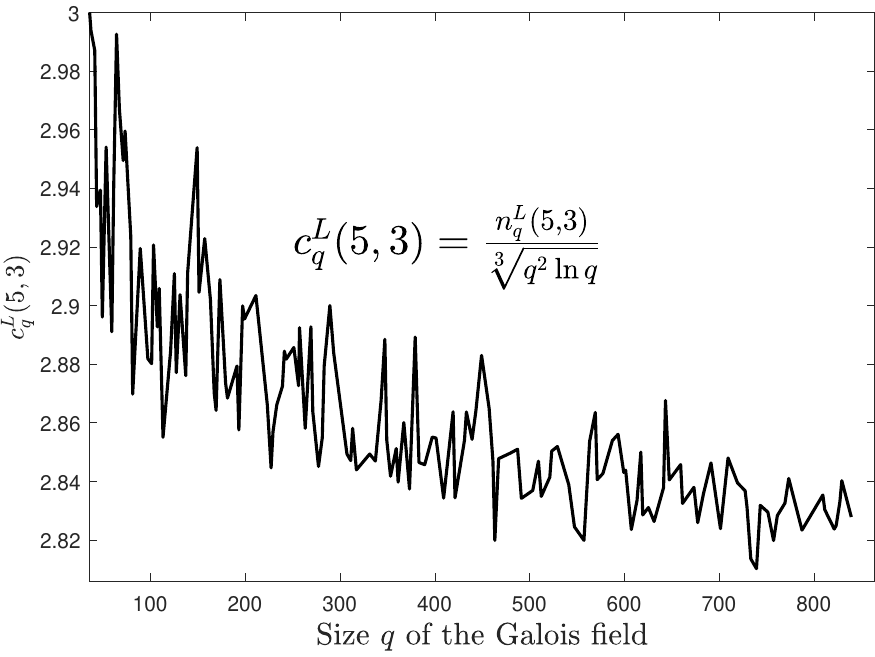}
\caption{Coefficients  $c_q^\text{L}(5,3)=n^\text{L}_q(5,3)/\sqrt[3]{q^2\ln q}$ for the $[n^\text{L}_q(5,3),n^\text{L}_q(5,3)-5,5]_q3$ leximatrix quasi-perfect Almost MDS codes; $37\le q\le839$}
\label{fig_PG4qFOPcoef}
\end{figure}

\begin{observation}\label{observ4}
\begin{description}
\item[(i)] The difference $\delta_q(5,3)$ tends to increase when $q$ grows, see Figures $\ref{fig_PG4qFOPsize}$ and $\ref{fig_PG4qFOPdlt}$.

\item[(ii)] The percent difference $\delta_q^{\%}(5,3)$ tends to increase when $q$ grows, see Figure $\ref{fig_PG4qFOPperc}$.

\item[(iii)] Coefficients $c_q^\text{L}(5,3)$ tend to decrease when $q$ grows, see Figure $\ref{fig_PG4qFOPcoef}$.
\end{description}
\end{observation}

Observation \ref{observ4} gives rise to Conjecture~\ref{conj2}(ii)  on the length function $\ell_q(5,3)$  and the $d$-length function $\ell_q(5,3,5)$.

Note that Observations \ref{observ4}(ii) and \ref{observ4}(iii) directly follow each from other. Actually,
$$\delta_q^{\%}(5,3)=\frac{3\sqrt[3]{q^2\ln q}-n^\text{L}_q(5,3)}{3\sqrt[3]{q^2\ln q}}100\%=\left(1-\frac{c_q^\text{L}(5,3)}{3}\right)100\%.$$

\begin{proposition}\label{prop5_density_lexi_r=5}
There exist $[n^{\text{L}}_q(5,3),n^{\text{L}}_q(5,3)-5,5]_q3$ quasi-perfect Almost MDS leximatrix codes of covering density $\mu^{\text{L}}_q(5,3)<4.2\cdot\ln q$  for $37\le q\le839$.
\end{proposition}

\begin{proof}
The needed codes are the codes of Proposition \ref{prop7_lexi_r=5}.
\end{proof}

Proposition \ref{prop5_density_lexi_r=5} implies the assertion of Theorem \ref{th2_density} on the upper \textbf{\emph{density lexi-bound}} on covering density $\mu_q(5,3)$.

Covering densities $\mu^{\text{L}}_q(5,3)$ of the $[n^{\text{L}}_q(5,3),n^{\text{L}}_q(5,3)-5,5]_q3$ leximatrix quasi-perfect Almost MDS codes  are presented in Figure~\ref{fig_PG4qFOPdensity} by the bottom solid black curve. The values $\mu^{\text{L}}_q(5,3)$ are obtained by \eqref{eq3_Ldensity} where lengths $n^{\text{L}}_q(5,3)$ are taken from Table 2 (see Appendix).
The bound
\begin{align*}
\mu^{\text{L}}_q(5,3)<4.2\cdot\ln q,
\end{align*}
called the \emph{\textbf{density lexi-bound}},  is shown in Figure \ref{fig_PG4qFOPdensity} by the top dashed red curve.
\begin{figure}[htb]
\includegraphics[width=\textwidth]{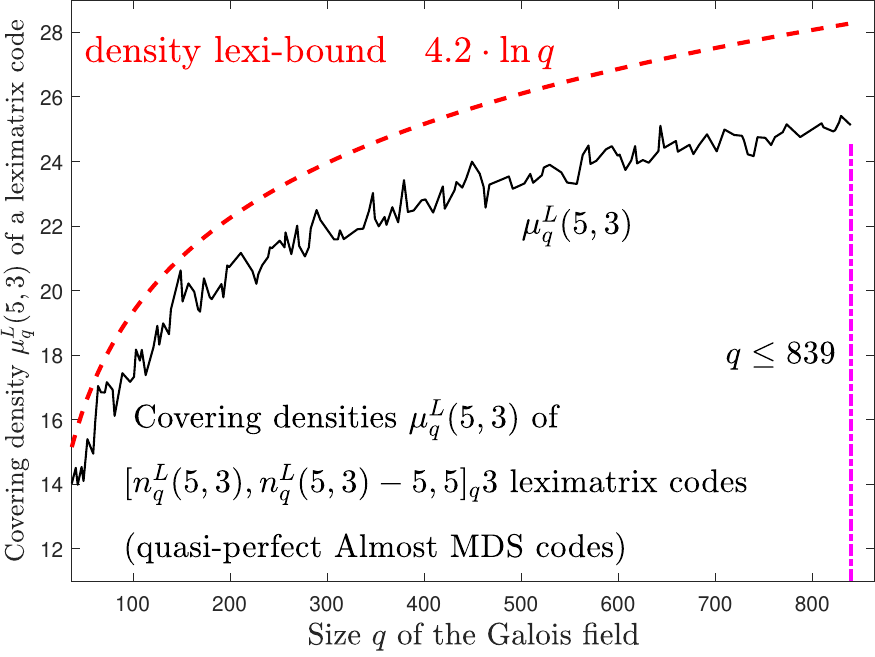}
\caption{Covering densities $\mu^\text{L}_q(5,3)$ of the $[n^\text{L}_q(5,3),n^\text{L}_q(5,3)-5,5]_q3$ leximatrix quasi-perfect Almost MDS codes  (\emph{bottom solid black curve}) vs the density lexi-bound $4.2\cdot\ln q$ (\emph{top dashed red curve}); $37\le q\leq 839$. \emph{Vertical magenta line} marks region $q\le839$}
\label{fig_PG4qFOPdensity}
\end{figure}

By \eqref{eq3_mLex}, we represent covering density of an $[n^{L}_q(5,3),n^{L}_q(5,3)-5,5]_q3$ leximatrix code in the form
\begin{align*}
    \mu^{\text{L}}_q(5,3)=m^{\text{L}}_q(5,3)\cdot\ln q,
\end{align*}
where $m^{\text{L}}_q(5,3)$ is a coefficient entirely given by $q$ (if $q$ is prime) or by $q$ and the primitive polynomial of the field $\F_q$ (if $q$ is non-prime). The coefficients $m_q^\text{L}(5,3)=\frac{\mu^{\text{L}}_q(5,3)}{\ln q}$ are shown in Figure~\ref{fig_PG4qFOPdensdiv}.
\begin{figure}[htb]
\includegraphics[width=\textwidth]{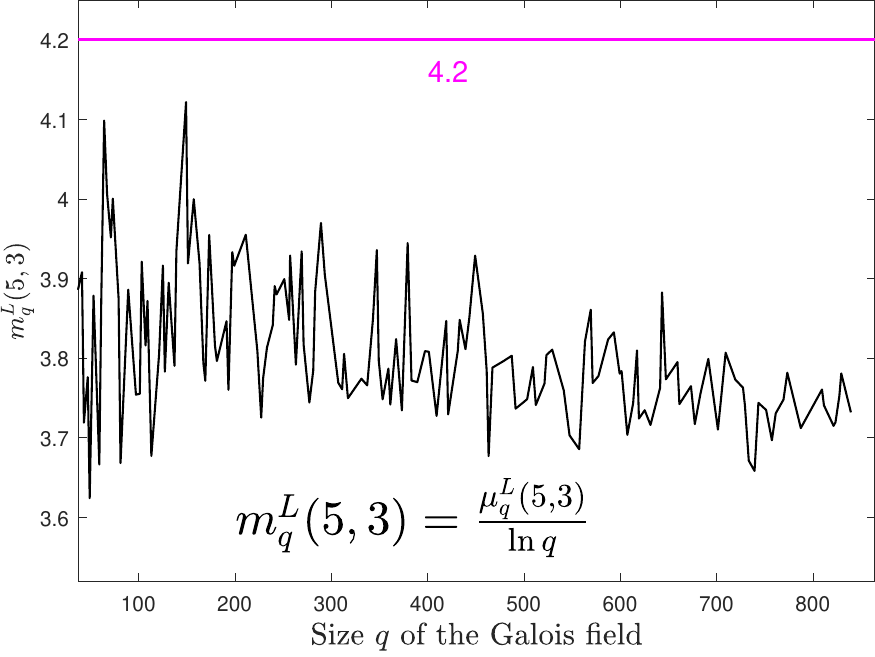}
\caption{Coefficients  $m_q^\text{L}(5,3)=\mu^\text{L}_q(5,3)/\ln q$ for covering density of the\newline $[n^\text{L}_q(5,3),n^\text{L}_q(5,3)-5,5]_q3$ leximatrix quasi-perfect Almost MDS codes;  $37\le q\leq 839$}
\label{fig_PG4qFOPdensdiv}
\end{figure}

\begin{observation}\label{observ5}
\begin{description}
\item[(i)] The difference $4.2\cdot\ln q-\mu_q^\text{L}(5,3)$ tends to increase when $q$ grows, see Figure $\ref{fig_PG4qFOPdensity}$.
\item[(ii)] Coefficients $m_q^\text{L}(5,3)$ tend to decrease when $q$ grows, see Figure $\ref{fig_PG4qFOPdensdiv}$.
\end{description}
\end{observation}

\section{An inverse leximatrix algorithm to obtain parity check matrices of covering codes}\label{sec_inv_gr-matr-alg}

An inverse leximatrix algorithm is a modification of the leximatrix algorihm of Section \ref{sec_gr-matr-alg}.

Let $\F_{q}=\{0,1,\ldots ,q-1\}$ be the Galois field with $q$ elements.

If $q$ is prime, the elements of
 $\F_{q}$ are treated as integers modulo $q$.

If $q=p^{m}$ with $p$ prime and $m\ge2$, the elements of $\F_{p^{m}}$ are
represented by integers as follows: $\F_{p^{m}}=\F_{q}=\{0,1=\alpha^{0},2=\alpha^{1},\ldots,u=\alpha^{u-1},\ldots,q-1=\alpha^{q-2}\},$
where $\alpha$ is a root of a primitive polynomial of $\F_{p^{m}}$.

For a $q$-ary code of codimension $r$, covering radius $R$, and minimum distance $d=R+2$,  we construct  a parity check matrix from nonzero columns $h_i$ of the form
\begin{align}\label{eq5_column}
h_{i}=(x_{1}^{(i)},x_{2}^{(i)},\ldots,x_{r}^{(i)})^{tr},~ x_{u}^{(i)}\in
\F_{q},
\end{align}
 where the first (leftmost) non-zero element is 1; \emph{tr} is the sign of transposition. The number of distinct columns is $(q^r-1)/(q-1)$. We order the columns in the list as
 \begin{align}\label{eq5_list}
    h_1,h_2,\ldots,h_{(q^r-1)/(q-1)}.
 \end{align}

One sees that the forms of the columns of a parity check matrix in  \eqref{eq2_column} and \eqref{eq5_column} coincide with each other. Also, external view of the list of the columns in \eqref{eq2_list} and \eqref{eq5_list} is the same. However, contrary to \eqref{eq2_number_i}, we represent a number $i$ of a column $h_i$ as follows:
\begin{align}\label{eq5_number_i_inver}
    i=\frac{q^r-1}{q-1}-\sum\limits_{u=1}^r x_{u}^{(i)}q^{r-u}.
\end{align}
We call the  \textbf{\emph{order of the columns}} corresponding to \eqref{eq5_number_i_inver}  \textbf{\emph{the inverse lexicographical order}}.

Apart the fixed order of columns (cf.  \eqref{eq2_number_i} and \eqref{eq5_number_i_inver} ), the inverse leximatrix algorithm is similar to the leximatrix algorithm.

The first column of the list should be included into the matrix. Then step-by-step,
one takes the next column from the list which cannot be represented as a linear
combination of at most $R$ columns already chosen. The process ends when no new column may
be included into the matrix. The obtained matrix $H_n$ is a parity check matrix of an $[n,n-r,R+2]_qR$ code.

The obtained parity check matrix is called the \textbf{\emph{parity check invleximatrix}} or the \textbf{\emph{invleximatrix}} for short.  We call an \textbf{\emph{invleximatrix code}} the corresponding code.

\textbf{For prime $q$}, the following holds: \textbf{length $n$ of an invleximatrix code and the form of the invleximatrix $H_n$ depend on $q$, $r$, and $R$ only}.  No other factors affect code length and structure. Actually, assume that after some step a current matrix is obtained. At the next step we should remove from our current list all columns that are linear combination of $R$ or less columns of the current matrix. For prime $q$ and the given  $r$ and $R$, the result of removing is unequivocal; hence, the next column is taken uniquely.

For non-prime $q$, the length  $n$ of an
invleximatrix code depends on $q$ and on the primitive polynomial of the field.
In this paper, we use primitive polynomials that are
created by the program system MAGMA \cite{MAGMA} by default, see Table A. In any case,
the choice of the polynomial changes the invleximatrix code length
unessentially.

By the invleximatrix algorithm, if $R=1$, we obtain the $q$-ary Hamming code. If $R=2$, we obtain a quasi-perfect $[n,n-r,4]_q2$ code; for $r=3$ such code is an MDS code and corresponds to a complete arc in $\mathrm{PG}(2,q)$. If $R=3$, we obtain a quasi-perfect $[n,n-r,5]_q3$ code; for $r=4$ such code is an MDS code and corresponds to a complete arc in $\mathrm{PG}(3,q)$; for $r=5$ it is an Almost MDS code.

Let $n^\text{IL}_q(r,R)$ be \textbf{length of the $q$-ary invleximatrix code of codimension $r$ and covering radius~$R$}. It is assumed that for a non-prime field $\F_q$, one uses the primitive polynomial created by the program system MAGMA \cite{MAGMA} by default; in particular, for non-prime $q\le6889$, the polynomial from Table A should be taken.

We represent length of an $[n^\text{IL}_q(r,R),n^\text{IL}_q(r,R)-r,R+2]_qR$ invleximatrix code in the form
\begin{align}\label{eq5_cinvLex}
    n^{\text{IL}}_q(r,R)=c^{\text{IL}}_q(r,R)\sqrt[R]{\ln q}\cdot q^{(r-R)/r},
\end{align}
where $c^{\text{IL}}_q(r,R)$ is a coefficient entirely given by $r,R,q$ (if $q$ is prime) or by $r,R,q$, and the primitive polynomial of $\F_q$ (if $q$ is non-prime).
\begin{proposition}\label{prop5_invlexi_r=4}
There exist $[n^{\emph{\text{IL}}}_q(4,3),n^{\emph{\text{IL}}}_q(4,3)-4,5]_q3$ quasi-perfect MDS \smallskip  invleximatrix codes of length $n^{\emph{\text{IL}}}_q(4,3)<2.8\sqrt[3]{q\ln q}$~for  $127\le q\leq 6101$ and $q=6143,6217,6287$, $6299,6529,6563$.
\end{proposition}

\begin{proof}
The needed codes are obtained by computer search, using the inverse leximatrix algorithm.
\end{proof}

Proposition \ref{prop5_invlexi_r=4} as well as Proposition \ref{prop7_lexi_r=4} implies the assertions of Theorem \ref{th2_res_rad3}(1i) on the upper \textbf{\emph{lexi-bound}} on the length function $\ell_q(4,3)$  and the $d$-length function $\ell_q(4,3,5)$.

Lengths of the $[n^{\text{IL}}_q(4,3),n^{\text{IL}}_q(4,3)-4,5]_q3$ invleximatrix quasi-perfect MDS codes are collected in Table 3 (see Appendix). The cases $$n^\text{IL}_q(4,3)<n^\text{L}_q(4,3)$$ are noted in Table 3 in \textbf{\emph{bold italic}} font.

We have relatively\textbf{\emph{ many the cases $n^\text{\emph{IL}}_q(4,3)<n^\text{\emph{L}}_q(4,3)$; this strengthens our assurance in truth of Conjecture}} \textbf{\ref{conj2}(i)}.

\section{Randomized greedy algorithms to obtain parity check matrices of covering codes}\label{sec_rand_gr_alg}
\subsection{Randomized greedy algorithms}
Randomized greedy algorithms are described (in geometrical language) in \cite{BDFKMP-PIT2014,BDFKMP_ArcFOPArXiv2015,BDFKMP_ArcComputJG2016}, see also
the references therein.

In every step a randomized greedy  algorithm maximizes an
objective function $f$ but some steps are executed in a random
manner. The number of these steps, their ordinal numbers, and
some other parameters of the algorithm have been taken
intuitively. Also, if the same maximum of $f$ can be obtained
in distinct ways, one way is chosen randomly.

We begin to construct a parity check matrix of an $[n,n-r]_qR$ code by using a starting matrix
 $ H_{0} $. In the $i$-th step one column is added to the current matrix
$H_{i-1}$ and we obtain a matrix $H_{i}$. We say that an $r$-dimensional \textbf{\emph{column is $R$-covered}} if it can be represented as linear combination at most $R$ columns of the current parity check matrix.
 As the value of
the objective function $f$ we consider \emph{\textbf{the number of $R$-covered
columns}}.

On every \textquotedblleft random\textquotedblright\ $i$-th
step we take $ d_{q,i}$ \emph{randomly chosen columns} of
$F_q^r$ \emph{not covered by }$ H_{i-1}$ and compute
the objective function $f$ adding each of these $ d_{q,i} $
columns to $H_{i-1}$. The column providing the maximum of $f$ is
included into~$H_{i}.$ On every \textquotedblleft
non-random\textquotedblright\ $j$-th step we consider \emph{all
columns not covered by }$H_{j-1}$ and add to $H_{j-1}$ the column
providing the maximum of $f.$

 As $H_{0}$ we can use a matrix  obtained in previous stages of the search.

A generator of random numbers is used for a random choice. To
get codes with distinct lengths, the starting conditions of the
generator are changed for the same matrix $H_{0}$. In this way the
algorithm works in a convenient limited region of the search
space to obtain examples decreasing the size of the matrix from
which the fixed starting submatrix have been taken.

To obtain codes with new lengths, sufficiently many
attempts should be made with  randomized greedy algorithms.
``Predicted'' lengths  could be useful for understanding if a
good result has been obtained.  If the result is not close to the
predicted size, the attempts are continued.

We consider the following two versions of the randomized greedy algorithms:

$\bullet$ \textbf{\emph{Rand-Greedy algorithm.}}  In this version, one does not take into account if a new column is $R$-covered. Therefore, the constructed   code has minimum distance $d=3$.

$\bullet$ \textbf{\emph{d-Rand-Greedy algorithm.}}  In this version, we chose a \emph{new column from columns that are not $R$-covered}. Therefore, minimum distance of the obtained code is  $d=R+2$.

The randomized greedy algorithms give better results than the leximatrix and inverse leximatrix algorithms but the randomized greedy algorithms take  essentially greater computer time.

Lengths of codes obtained by randomized greedy algorithms depend of many factors connected with parameters of the algorithms.

Let $n^\text{G}_q(r,R)$ be \textbf{length of a $q$-ary code of codimension $r$ and covering radius~$R$ obtained by the Rand-Greedy algorithm}.

We represent length of an $[n^\text{G}_q(r,R),n^\text{G}_q(r,R)-r,3]_qR$  code obtained by the Rand-Greedy algorithm in the form
\begin{align*}
    n^{\text{G}}_q(r,R)=c^{\text{G}}_q(r,R)\sqrt[R]{\ln q}\cdot q^{(r-R)/r},
\end{align*}
where $c^{\text{G}}_q(r,R)$ is a coefficient dependent on parameters of the Rand-Greedy algorithm.

Let $n^{d\text{G}}_q(r,R)$ be \textbf{length of a $q$-ary code of codimension $r$ and covering radius~$R$ obtained by the {\mathversion{bold}$d$}-Rand-Greedy algorithm}.

We represent length of an $[n^{d\text{G}}_q(r,R),n^{d\text{G}}_q(r,R)-r,R+2]_qR$  code obtained by the $d$-Rand-Greedy algorithm in the form
\begin{align*}
    n^{d\text{G}}_q(r,R)=c^{d\text{G}}_q(r,R)\sqrt[R]{\ln q}\cdot q^{(r-R)/r},
\end{align*}
where $c^{d\text{G}}_q(r,R)$ is a coefficient dependent on parameters of the $d$-Rand-Greedy algorithm.

Let \textbf{$\overline{n}_q(r,R)$ be length of the shortest \emph{known} $q$-ary code of codimension $r$ and covering radius $R$.}

Let \textbf{$\overline{n}_q(r,R,d)$ be length of the shortest \emph{known} $q$-ary code of codimension $r$, covering radius $R$, and minimum distance $d$.}

Clearly,
\begin{align*}
    \overline{n}_q(r,R)\le\overline{n}_q(r,R,d).
\end{align*}
We represent length $\overline{n}_q(r,R,d)$ in the form
\begin{align*}
   \overline{n}_q(r,R,d)=\overline{c}_q(r,R,d)\sqrt[R]{\ln q}\cdot q^{(r-R)/R},
\end{align*}
where $\overline{c}_q(r,R,d)$ is a coefficient.

\subsection{The shortest known $[\overline{n}_q(4,3,5),\overline{n}_q(4,3,5)-4,5]_q3$ and \\$[\overline{n}_q(4,3),\overline{n}_q(4,3)-4]_q3$  codes}
For $2\leq q\leq 7057$,  lengths $\overline{n}_q(4,3,5)$ of the \textbf{\emph{shortest known}}  $[\overline{n}_q(4,3,5),\overline{n}_q(4,3,5)-4,5]_q3$ quasi-perfect MDS codes,  obtained by the leximatrix, inverse leximatrix, and $d$-Rand-Greedy algorithms, are as follows
\begin{align}\label{eq6_mixt}
   \overline{n}_q(4,3,5)=\min\{n^\text{L}_q(4,3),n^\text{IL}_q(4,3),n^{d\text{G}}_q(4,3)\}.
\end{align}
 \begin{proposition}\label{prop6_mixt_r=4}
 There exist $[\overline{n}_q(4,3,5),\overline{n}_q(4,3,5)-4,5]_q3$ quasi-perfect MDS  codes of length
\begin{align*}
    \overline{n}_q(4,3,5)<\left\{
    \begin{array}{ccc}
    2.61\sqrt[3]{q\ln q} & \text{if} & 13\le q\le4373 \smallskip\\
    2.65\sqrt[3]{q\ln q} & \text{if} & 4373<q\le7057
    \end{array}
    \right..
\end{align*}
\end{proposition}
\begin{proof}
The needed codes are obtained by computer search, using the approach of \eqref{eq6_mixt}.  To obtain codes with $q\le4451$ we used the $d$-Rand-Greedy algorithm. For\\ $4451<q\le7057$ we used, in preference, the leximatrix and inverse leximatrix algorithms, see Sections \ref{sec_ell_q(5,3)} and \ref{sec_inv_gr-matr-alg}, but for $q=4489,4679,4877,4889,4913,5801,6653$ we applied the $d$-Rand-Greedy algorithm. For $q=841$, the complete 42-arc of \cite{Sonin} is used.
\end{proof}

Proposition \ref{prop6_mixt_r=4} implies the assertions of Theorem \ref{th2_res_rad3}(2) on upper  bounds on the length function $\ell_q(4,3)$  and the $d$-length function $\ell_q(4,3,5)$.

Proposition \ref{prop6_mixt_r=4} \textbf{\emph{improves the lexi-bound of Theorem \emph{\ref{th2_res_rad3}(1i)}; this strengthens our assurance in truth of Conjecture}} \textbf{\ref{conj2}(i)}.

The lengths $\overline{n}_q(4,3)$ of the shortest known $[\overline{n}_q(4,3),\overline{n}_q(4,3)-4]_q3$  codes we obtain using results of computer search for  $\overline{n}_q(4,3,5)$, data from \cite[Tab.\,1]{DavOst-DESI2010}, the Rand-Greedy algorithm, and formula \eqref{eq1_rad3_q3} for $q=(q')^3$ where $(q')$ is a prime power.

The \textbf{\emph{smallest known lengths }}$\overline{n}_q(4,3)$ are given in Table 4  (see Appendix) where the cases
$$\overline{n}_q(4,3,5)=\overline{n}_q(4,3)+j$$
 are noted by the superscript ``$+j$''. For the rest of $q$ we have  $\overline{n}_q(4,3,5)=\overline{n}_q(4,3)$. So, in fact, \emph{Table 4 gives also \textbf{smallest known lengths} $\overline{n}_q(4,3,5)$.}

Note, that in Table 4, the improvements of code distance up to $d=5$ in comparison with \cite[Tab.\,1]{DavOst-DESI2010} are noted in bold italic font.
Also, in Table 4, the cases $\ell_q(4,3)=\overline{n}_q(4,3)$ are noted by the subscript ``$\bullet$'', see \cite[Tab.\,1]{DavOst-DESI2010}.

Coefficients $\overline{c}_q(4,3,5)$ corresponding to the codes of Table 4 (taking into account the superscripts ``$+j$'') are shown in Figure~\ref{fig_PG3qBestCoef_d5}.

\begin{figure}[htb]
\includegraphics[width=\textwidth]{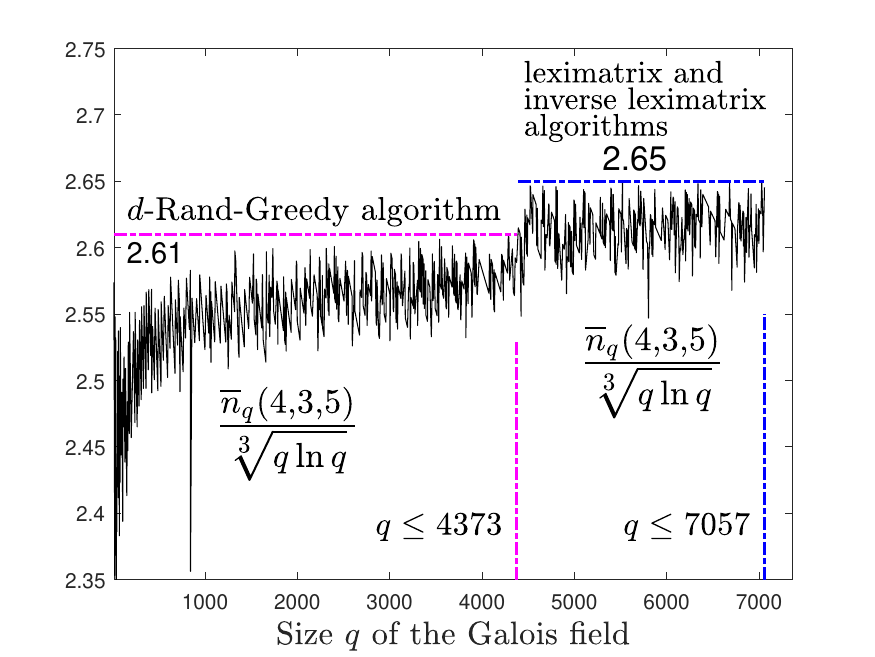}
\caption{Coefficients  $\overline{c}_q(4,3,5)= \overline{n}_q(4,3,5)/\sqrt[3]{q\ln q}$ for $[\overline{n}_q(4,3,5),\overline{n}_q(4,3,5)-4,5]_q3$ quasi-perfect  MDS codes;  $13\le q\leq 7057$}
\label{fig_PG3qBestCoef_d5}
\end{figure}

\subsection{The shortest known $[\overline{n}_q(5,3),\overline{n}_q(5,3)-5]_q3$ codes}
For $3\leq q\leq 839$,  lengths $\overline{n}_q(5,3)$ of the\textbf{ \emph{shortest known}}  $[\overline{n}_q(5,3),\overline{n}_q(5,3)-5,3]_q3$ codes, obtained by the leximatrix and Rand-Greedy algorithms, are as follows
\begin{align}\label{eq6_mixt_r=5}
&\overline{n}_q(5,3)=\min\{n^\text{L}_q(5,3),n^\text{G}_q(5,3)\} \text{ if }11\le q\le401;\\
&\overline{n}_q(5,3)=\overline{n}_q(5,3,5)=n^\text{L}_q(5,3) \text{ if }401<q\le839.\notag
\end{align}
\begin{proposition}\label{prop6_mixt_r=5}There exist $[\overline{n}_q(5,3),\overline{n}_q(5,3)-5]_q3$   codes of length
\begin{align*}
&\overline{n}_q(5,3)< 2.785\sqrt[3]{q^2\ln q}  \text{ if }  11\le q\le401;\\
&\overline{n}_q(5,3)= \overline{n}_q(5,3,5)<   2.884\sqrt[3]{q^2\ln q}   \text{ if }   401<q\le839.
\end{align*}
\end{proposition}
\begin{proof}
The needed codes are obtained by computer search, using the approach of \eqref{eq6_mixt_r=5}.  To obtain codes with $q\le401$ we used the Rand-Greedy algorithm; it gives $[n,n-5,3]_q3$ codes with minimum distance $d=3$. For $401<q\le839$ we used the leximatrix  algorithm, see Section \ref{sec_ell_q(5,3)}.
\end{proof}
Proposition \ref{prop6_mixt_r=5} implies the assertions of Theorem \ref{th2_res_rad3}(3) on upper  bounds on the length function $\ell_q(5,3)$ and the $d$-length function $\ell_q(5,3,5)$.

Proposition \ref{prop6_mixt_r=5} \textbf{\emph{improves the lexi-bound of Theorem \emph{\ref{th2_res_rad3}(1ii)}; this strengthens our assurance in truth of Conjecture}} \textbf{\ref{conj2}(ii)}.

Lengths $\overline{n}_q(5,3)$ of the\textbf{ \emph{shortest known}}  $[\overline{n}_q(5,3),\overline{n}_q(5,3)-5,3]_q3$ codes are collected in Table 5,
 where the improvements of code length in comparison with \cite[Tab.\,2]{DavOst-DESI2010} are noted in bold italic font. Also,in Table 5, the cases $\ell_q(5,3)=\overline{n}_q(5,3)$ are noted by the subscript~``$\bullet$'', see \cite[Tab.\,2]{DavOst-DESI2010}.

For $q\le839$, coefficients $\overline{c}_q(5,3,3)$ and $\overline{c}_q(5,3,5)$ corresponding to the codes of Table~5 are shown in Figure \ref{fig_PG4qBestCoef}.

\begin{figure}[htb]
\includegraphics[width=\textwidth]{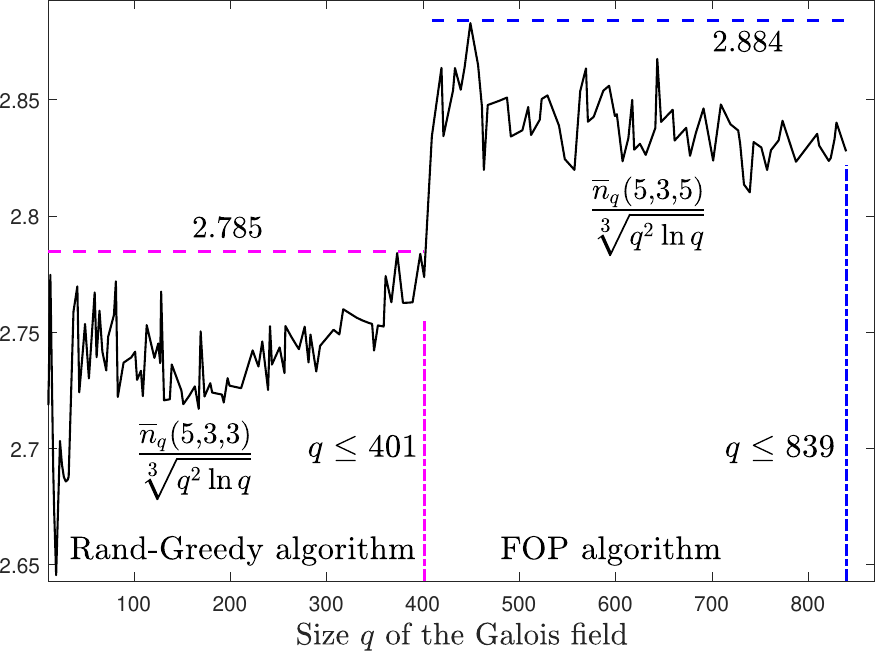}
\caption{Coefficients  $\overline{c}_q(5,3,3)= \overline{n}_q(5,3,3)/\sqrt[3]{q^2\ln q}$ for $[\overline{n}_q(5,3,3),\overline{n}_q(5,3,3)-5,3]_q3$ codes, $11\le q\leq 401$; and $\overline{c}_q(5,3,5)= \overline{n}_q(5,3,5)/\sqrt[3]{q^2\ln q}$ for $[\overline{n}_q(5,3,5),\overline{n}_q(5,3,5)-5,5]_q3$ Almost  MDS codes, $401< q\leq 839$}
\label{fig_PG4qBestCoef}
\end{figure}

\section{Conclusion}
\textbf{\emph{The length function}} $\ell_q(r,R)$ is the smallest
length of a $ q $-ary linear code of covering radius $R$ and
codimension $r$.
\textbf{\emph{The $d$-length function}} $\ell_q(r,R,d)$ is the smallest
length of a $ q $-ary linear code with codimension $r$, covering radius $R$, and minimum distance~$d$.
   In this paper, we consider upper bounds on the length functions $\ell_q(4,3)$, $\ell_q(5,3)$ and the $d$-length functions $\ell_q(4,3,5)$, $\ell_q(5,3,5)$. For $r\ne3t$ and  $q\ne(q')^3$, where $q'$ is a prime power, upper bounds on
   $\ell_q(r,3)$ and $\ell_q(r,3,5)$ are open problems.

By computer search in wide regions of $q$, we obtained following short codes of covering radius $R=3$: $[n,n-4,5]_q3$ quasi-perfect MDS codes, $[n,n-5,5]_q3$ quasi-perfect Almost MDS codes, and  $[n,n-5,3]_q3$ codes.

 For $r\neq 3t$ and the field basis $q$ of an arbitrary structure, including $q\neq (q^{\prime})^3$, the new codes imply upper bounds (called the \emph{\textbf{lexi-bounds}}) of the form
   \begin{align}\label{eq8_length_func}
  \ell_q(r,3)< c_\ell\sqrt[3]{\ln q}\cdot q^{(r-3)/3},~~c_\ell\text{ is a constant independent of } q,~~r=4,5\neq 3t.
\end{align}
Also, the new codes imply the following upper bounds (called the \emph{\textbf{density lexi-bounds}}) on the smallest
covering density $\mu_q(r,3)$ of a $ q $-ary linear code of covering radius $3$ and
codimension $r$:
   \begin{align}\label{eq8_density}
\mu_q(r,3)<c_\mu\cdot\ln q,~~c_\mu\text{ is a constant independent of } q,~~r=4,5\neq 3t.
\end{align}

\emph{In comparison with upper bounds of \cite{DGMP_ACCT2008,DGMP_Petersb2008,DGMP-AMC} for the special field basis $q=(q^{\prime})^3$, bounds \eqref{eq8_length_func} and \eqref{eq8_density} have ``extra'' multipliers $\sqrt[3]{\ln q}$ and $\ln q$, respectively. \textbf{This is a ``price'' of an arbitrary structure of the field basis $q$.}}

In computer search, we use the  step-by-step leximatrix and inverse leximatrix algorithms to obtain parity check matrices of codes. The algorithms are versions of the recursive $g$-parity check matrix algorithm for greedy codes. Also, we apply the randomized greedy algorithms.

In future, it would be useful to investigate and understand properties of the leximatrix and inverse leximatrix algorithms and structure of leximatrices and invleximatrices.

 In particular, the following is of great interest:

$\bullet$ Initial part of the parity check leximatrices and invleximatrices, see Proposition \ref{prop7_init_part_lexi} and Example \ref{ex7}.

$\bullet$  The working mechanism and its quantitative estimates
       for the leximatrix and inverse leximatrix algorithms; see, for instance, the papers \cite{BDFKMP-PIT2014,DFMP-ConjCap} where the working
mechanisms  of greedy algorithms for complete arcs in the projective plane $\mathrm{PG}(2,q)$ and for complete caps in the projective spaces $\mathrm{PG}(N,q)$ are studied.

$\bullet$ The oscillation of the coefficients $c_q^\text{L}(4,3)$ around a horizontal line and its likenesses with the oscillation of the values $h^\text{L}(q)$ around a horizontal line in \cite[Fig.\,6, Observation~3.5]{BDFKMP_ArcFOPArXiv2015}, \cite[Fig.\,5, Observation 3.7]{BDFKMP_ArcComputJG2016}, see Figure \ref{fig_PG3qFOPcoef} and Remark \ref{rem3_oscil}.

It is important to emphasize  that although the \textbf{\emph{lexi-bounds}} of Theorem \ref{th2_res_rad3}(1) are obtained by computer search, several factors give us insurance that these bounds \textbf{\emph{are truth for all~$q$}} (see Conjecture \ref{conj2}). In particular we note figures and observations in Sections~\ref{sec_ell_q(4,3)} and  \ref{sec_ell_q(5,3)}, comparison of leximatrix and invleximatrix codes in Table 3, improvements of the lexi-bounds in Section~\ref{sec_rand_gr_alg}.

\newpage
\section*{Appendix}
\noindent\textbf{Table 1.} Lengths $n^\text{L}_q(4,3)$ of the $[n^\text{L}_q(4,3),n^\text{L}_q(4,3)-4,5]_q3$ leximatrix quasi-perfect\smallskip\\ MDS codes, $2\le q\leq 7057$ \medskip

\noindent \begin{tabular}
{@{}r@{\,}c@{\,}|r@{\,}c@{\,}|r@{\,}c@{\,}|r@{\,}c@{\,}|r@{\,}c@{\,}|r@{\,}c@{}}
 \hline
$q~~$ & $n^\text{L}_q(4,3)$ &
$q~~$ & $n^\text{L}_q(4,3)$ &
$q~~$ & $n^\text{L}_q(4,3)$ &
$q~~$ & $n^\text{L}_q(4,3)$ &
$q~~$ & $n^\text{L}_q(4,3)$ &
$q~~$ & $n^\text{L}_{q_{\vphantom{H_{L}}}}(4,3)^{ \vphantom{H^{L}}}$ \\
 \hline
 2 & 5  &  3 & 5  &  4 & 5  &  5 & 6  &  7 & 8  &  8 & 7  \\
 9 & 9  &  11 & 8  &  13 & 9  &  16 & 9  &  17 & 9  &  19 & 10  \\
 23 & 11  &  25 & 11  &  27 & 12  &  29 & 12  &  31 & 13  &  32 & 12  \\
 37 & 13  &  41 & 14  &  43 & 14  &  47 & 15  &  49 & 15  &  53 & 16  \\
 59 & 16  &  61 & 16  &  64 & 17  &  67 & 17  &  71 & 18  &  73 & 18  \\
 79 & 18  &  81 & 18  &  83 & 19  &  89 & 20  &  97 & 20  &  101 & 21  \\
 103 & 20  &  107 & 22  &  109 & 22  &  113 & 22  &  121 & 22  &  125 & 23  \\
 127 & 23  &  128 & 22  &  131 & 23  &  137 & 23  &  139 & 23  &  149 & 24  \\
 151 & 24  &  157 & 25  &  163 & 24  &  167 & 25  &  169 & 25  &  173 & 25  \\
 179 & 26  &  181 & 26  &  191 & 26  &  193 & 27  &  197 & 27  &  199 & 26  \\
 211 & 27  &  223 & 29  &  227 & 28  &  229 & 28  &  233 & 28  &  239 & 29  \\
 241 & 29  &  243 & 28  &  251 & 30  &  256 & 29  &  257 & 29  &  263 & 30  \\
 269 & 30  &  271 & 31  &  277 & 30  &  281 & 30  &  283 & 31  &  289 & 31  \\
 293 & 31  &  307 & 32  &  311 & 32  &  313 & 31  &  317 & 32  &  331 & 34  \\
 337 & 34  &  343 & 33  &  347 & 34  &  349 & 34  &  353 & 34  &  359 & 34  \\
 361 & 34  &  367 & 34  &  373 & 34  &  379 & 34  &  383 & 34  &  389 & 35  \\
 397 & 35  &  401 & 35  &  409 & 35  &  419 & 36  &  421 & 36  &  431 & 36  \\
 433 & 37  &  439 & 38  &  443 & 38  &  449 & 36  &  457 & 37  &  461 & 37  \\
 463 & 37  &  467 & 37  &  479 & 38  &  487 & 38  &  491 & 39  &  499 & 39  \\
 503 & 39  &  509 & 39  &  512 & 39  &  521 & 39  &  523 & 39  &  529 & 39  \\
 541 & 39  &  547 & 39  &  557 & 39  &  563 & 41  &  569 & 41  &  571 & 39  \\
 577 & 40  &  587 & 41  &  593 & 41  &  599 & 41  &  601 & 42  &  607 & 42  \\
 613 & 43  &  617 & 42  &  619 & 42  &  625 & 42  &  631 & 42  &  641 & 43  \\
 643 & 42  &  647 & 43  &  653 & 44  &  659 & 44  &  661 & 43  &  673 & 43  \\
 677 & 42  &  683 & 43  &  691 & 44  &  701 & 44  &  709 & 44  &  719 & 44  \\
 727 & 45  &  729 & 44  &  733 & 45  &  739 & 45  &  743 & 45  &  751 & 45  \\
 757 & 46  &  761 & 45  &  769 & 46  &  773 & 46  &  787 & 45  &  797 & 46  \\
 809 & 46  &  811 & 46  &  821 & 46  &  823 & 47  &  827 & 46  &  829 & 46  \\
 839 & 46  &  841 & 47  &  853 & 47  &  857 & 47  &  859 & 47  &  863 & 47  \\
 877 & 48  &  881 & 47  &  883 & 47  &  887 & 48  &  907 & 50  &  911 & 49  \\
 919 & 48  &  929 & 49  &  937 & 49  &  941 & 49  &  947 & 49  &  953 & 49  \\
 961 & 50  &  967 & 50  &  971 & 50  &  977 & 50  &  983 & 50  &  991 & 50  \\
 997 & 52  &  1009 & 51  &  1013 & 51  &  1019 & 51  &  1021 & 50  &  1024 & 52  \\
 1031 & 50  &  1033 & 51  &  1039 & 51  &  1049 & 52  &  1051 & 51  &  1061 & 51  \\
 \hline
 \end{tabular}
 \newpage

\noindent\textbf{Table 1.} Continue 1\medskip

\noindent \begin{tabular}
{@{}r@{\,}c@{\,}|@{\,}r@{\,}c@{\,}|@{\,}r@{\,}c@{\,}|@{\,}r@{\,}c@{\,}|@{\,}r@{\,}c@{\,}|@{\,}r@{\,}c@{}}
 \hline
$q~~$ & $n^\text{L}_q(4,3)$ &
$q~~$ & $n^\text{L}_q(4,3)$ &
$q~~$ & $n^\text{L}_q(4,3)$ &
$q~~$ & $n^\text{L}_q(4,3)$ &
$q~~$ & $n^\text{L}_q(4,3)$ &
$q~~$ & $n^\text{L}_{q_{\vphantom{H_{L}}}}(4,3)^{ \vphantom{H^{L}}}$ \\
 \hline
 1063 & 51  &  1069 & 52  &  1087 & 52  &  1091 & 51  &  1093 & 52  &  1097 & 52  \\
 1103 & 52  &  1109 & 52  &  1117 & 52  &  1123 & 52  &  1129 & 53  &  1151 & 53  \\
 1153 & 53  &  1163 & 53  &  1171 & 53  &  1181 & 53  &  1187 & 54  &  1193 & 53  \\
 1201 & 52  &  1213 & 54  &  1217 & 55  &  1223 & 55  &  1229 & 54  &  1231 & 56  \\
 1237 & 56  &  1249 & 55  &  1259 & 54  &  1277 & 55  &  1279 & 56  &  1283 & 56  \\
 1289 & 56  &  1291 & 55  &  1297 & 56  &  1301 & 56  &  1303 & 56  &  1307 & 56  \\
 1319 & 56  &  1321 & 56  &  1327 & 56  &  1331 & 55  &  1361 & 57  &  1367 & 57  \\
 1369 & 56  &  1373 & 56  &  1381 & 57  &  1399 & 57  &  1409 & 57  &  1423 & 58  \\
 1427 & 58  &  1429 & 58  &  1433 & 57  &  1439 & 57  &  1447 & 57  &  1451 & 59  \\
 1453 & 59  &  1459 & 57  &  1471 & 57  &  1481 & 59  &  1483 & 59  &  1487 & 59  \\
 1489 & 59  &  1493 & 58  &  1499 & 58  &  1511 & 59  &  1523 & 58  &  1531 & 60  \\
 1543 & 59  &  1549 & 59  &  1553 & 59  &  1559 & 60  &  1567 & 60  &  1571 & 60  \\
 1579 & 59  &  1583 & 59  &  1597 & 59  &  1601 & 59  &  1607 & 60  &  1609 & 60  \\
 1613 & 60  &  1619 & 60  &  1621 & 60  &  1627 & 60  &  1637 & 60  &  1657 & 60  \\
 1663 & 61  &  1667 & 61  &  1669 & 60  &  1681 & 62  &  1693 & 61  &  1697 & 62  \\
 1699 & 62  &  1709 & 61  &  1721 & 63  &  1723 & 62  &  1733 & 63  &  1741 & 62  \\
 1747 & 63  &  1753 & 62  &  1759 & 62  &  1777 & 62  &  1783 & 63  &  1787 & 63  \\
 1789 & 62  &  1801 & 62  &  1811 & 63  &  1823 & 62  &  1831 & 62  &  1847 & 63  \\
 1849 & 64  &  1861 & 63  &  1867 & 63  &  1871 & 63  &  1873 & 64  &  1877 & 63  \\
 1879 & 63  &  1889 & 63  &  1901 & 64  &  1907 & 64  &  1913 & 64  &  1931 & 65  \\
 1933 & 66  &  1949 & 64  &  1951 & 66  &  1973 & 66  &  1979 & 65  &  1987 & 64  \\
 1993 & 65  &  1997 & 66  &  1999 & 65  &  2003 & 67  &  2011 & 66  &  2017 & 64  \\
 2027 & 65  &  2029 & 66  &  2039 & 66  &  2048 & 66  &  2053 & 66  &  2063 & 66  \\
 2069 & 66  &  2081 & 65  &  2083 & 66  &  2087 & 67  &  2089 & 67  &  2099 & 66  \\
 2111 & 67  &  2113 & 66  &  2129 & 67  &  2131 & 67  &  2137 & 68  &  2141 & 67  \\
 2143 & 66  &  2153 & 67  &  2161 & 67  &  2179 & 66  &  2187 & 68  &  2197 & 68  \\
 2203 & 67  &  2207 & 68  &  2209 & 67  &  2213 & 68  &  2221 & 69  &  2237 & 68  \\
 2239 & 68  &  2243 & 69  &  2251 & 69  &  2267 & 68  &  2269 & 69  &  2273 & 69  \\
 2281 & 69  &  2287 & 69  &  2293 & 68  &  2297 & 67  &  2309 & 69  &  2311 & 69  \\
 2333 & 69  &  2339 & 71  &  2341 & 69  &  2347 & 70  &  2351 & 69  &  2357 & 69  \\
 2371 & 70  &  2377 & 69  &  2381 & 69  &  2383 & 71  &  2389 & 69  &  2393 & 70  \\
 2399 & 70  &  2401 & 70  &  2411 & 71  &  2417 & 69  &  2423 & 71  &  2437 & 71  \\
 2441 & 73  &  2447 & 71  &  2459 & 70  &  2467 & 71  &  2473 & 72  &  2477 & 71  \\
 2503 & 70  &  2521 & 70  &  2531 & 71  &  2539 & 72  &  2543 & 72  &  2549 & 71  \\
 2551 & 71  &  2557 & 71  &  2579 & 72  &  2591 & 71  &  2593 & 72  &  2609 & 71  \\
 2617 & 72  &  2621 & 72  &  2633 & 73  &  2647 & 72  &  2657 & 73  &  2659 & 73  \\
 2663 & 72  &  2671 & 72  &  2677 & 73  &  2683 & 73  &  2687 & 72  &  2689 & 72  \\
 2693 & 72  &  2699 & 72  &  2707 & 73  &  2711 & 73  &  2713 & 72  &  2719 & 73  \\
 \hline
 \end{tabular}
\newpage

\noindent\textbf{Table 1.} Continue 2\medskip

\noindent \begin{tabular}
{@{}r@{\,}c@{\,}|@{\,}r@{\,}c@{\,}|@{\,}r@{\,}c@{\,}|@{\,}r@{\,}c@{\,}|@{\,}r@{\,}c@{\,}|@{\,}r@{\,}c@{}}
 \hline
$q~~$ & $n^\text{L}_q(4,3)$ &
$q~~$ & $n^\text{L}_q(4,3)$ &
$q~~$ & $n^\text{L}_q(4,3)$ &
$q~~$ & $n^\text{L}_q(4,3)$ &
$q~~$ & $n^\text{L}_q(4,3)$ &
$q~~$ & $n^\text{L}_{q_{\vphantom{H_{L}}}}(4,3)^{ \vphantom{H^{L}}}$ \\
 \hline
 2729 & 73  &  2731 & 74  &  2741 & 73  &  2749 & 73  &  2753 & 74  &  2767 & 73  \\
 2777 & 74  &  2789 & 74  &  2791 & 74  &  2797 & 73  &  2801 & 75  &  2803 & 74  \\
 2809 & 74  &  2819 & 74  &  2833 & 74  &  2837 & 75  &  2843 & 75  &  2851 & 75  \\
 2857 & 74  &  2861 & 74  &  2879 & 74  &  2887 & 76  &  2897 & 75  &  2903 & 74  \\
 2909 & 75  &  2917 & 75  &  2927 & 75  &  2939 & 76  &  2953 & 77  &  2957 & 76  \\
 2963 & 75  &  2969 & 75  &  2971 & 76  &  2999 & 76  &  3001 & 76  &  3011 & 75  \\
 3019 & 77  &  3023 & 76  &  3037 & 76  &  3041 & 75  &  3049 & 75  &  3061 & 76  \\
 3067 & 76  &  3079 & 78  &  3083 & 77  &  3089 & 76  &  3109 & 76  &  3119 & 77  \\
 3121 & 77  &  3125 & 78  &  3137 & 77  &  3163 & 78  &  3167 & 77  &  3169 & 77  \\
 3181 & 79  &  3187 & 77  &  3191 & 78  &  3203 & 77  &  3209 & 77  &  3217 & 78  \\
 3221 & 78  &  3229 & 77  &  3251 & 79  &  3253 & 78  &  3257 & 77  &  3259 & 78  \\
 3271 & 79  &  3299 & 79  &  3301 & 78  &  3307 & 78  &  3313 & 78  &  3319 & 79  \\
 3323 & 79  &  3329 & 80  &  3331 & 79  &  3343 & 78  &  3347 & 80  &  3359 & 78  \\
 3361 & 80  &  3371 & 79  &  3373 & 79  &  3389 & 80  &  3391 & 79  &  3407 & 80  \\
 3413 & 80  &  3433 & 80  &  3449 & 80  &  3457 & 80  &  3461 & 80  &  3463 & 80  \\
 3467 & 79  &  3469 & 80  &  3481 & 81  &  3491 & 80  &  3499 & 80  &  3511 & 80  \\
 3517 & 80  &  3527 & 80  &  3529 & 82  &  3533 & 80  &  3539 & 82  &  3541 & 80  \\
 3547 & 80  &  3557 & 82  &  3559 & 81  &  3571 & 81  &  3581 & 81  &  3583 & 80  \\
 3593 & 81  &  3607 & 83  &  3613 & 81  &  3617 & 81  &  3623 & 82  &  3631 & 81  \\
 3637 & 82  &  3643 & 82  &  3659 & 82  &  3671 & 83  &  3673 & 82  &  3677 & 82  \\
 3691 & 83  &  3697 & 83  &  3701 & 82  &  3709 & 83  &  3719 & 82  &  3721 & 82  \\
 3727 & 82  &  3733 & 82  &  3739 & 83  &  3761 & 82  &  3767 & 83  &  3769 & 83  \\
 3779 & 85  &  3793 & 83  &  3797 & 83  &  3803 & 82  &  3821 & 83  &  3823 & 82  \\
 3833 & 84  &  3847 & 83  &  3851 & 84  &  3853 & 82  &  3863 & 83  &  3877 & 84  \\
 3881 & 84  &  3889 & 83  &  3907 & 85  &  3911 & 84  &  3917 & 83  &  3919 & 83  \\
 3923 & 84  &  3929 & 84  &  3931 & 84  &  3943 & 84  &  3947 & 84  &  3967 & 84  \\
 3989 & 85  &  4001 & 85  &  4003 & 84  &  4007 & 85  &  4013 & 85  &  4019 & 86  \\
 4021 & 84  &  4027 & 84  &  4049 & 85  &  4051 & 86  &  4057 & 85  &  4073 & 85  \\
 4079 & 86  &  4091 & 85  &  4093 & 86  &  4096 & 86  &  4099 & 86  &  4111 & 86  \\
 4127 & 86  &  4129 & 86  &  4133 & 85  &  4139 & 86  &  4153 & 86  &  4157 & 86  \\
 4159 & 86  &  4177 & 87  &  4201 & 85  &  4211 & 87  &  4217 & 85  &  4219 & 87  \\
 4229 & 86  &  4231 & 87  &  4241 & 86  &  4243 & 86  &  4253 & 86  &  4259 & 88  \\
 4261 & 87  &  4271 & 86  &  4273 & 87  &  4283 & 87  &  4289 & 86  &  4297 & 87  \\
 4327 & 88  &  4337 & 88  &  4339 & 86  &  4349 & 89  &  4357 & 87  &  4363 & 87  \\
 4373 & 87  &  4391 & 87  &  4397 & 88  &  4409 & 88  &  4421 & 87  &  4423 & 90  \\
 4441 & 87  &  4447 & 88  &  4451 & 88  &  4457 & 87  &  4463 & 88  &  4481 & 87  \\
 4483 & 88  &  4489 & 89  &  4493 & 88  &  4507 & 89  &  4513 & 88  &  4517 & 88  \\
 4519 & 89  &  4523 & 89  &  4547 & 88  &  4549 & 90  &  4561 & 89  &  4567 & 89  \\
 \hline
 \end{tabular}
\newpage

\noindent\textbf{Table 1.} Continue 3\medskip

\renewcommand{\arraystretch}{0.95}
\noindent \begin{tabular}
{@{}r@{\,}c@{\,}|@{\,}r@{\,}c@{\,}|@{\,}r@{\,}c@{\,}|@{\,}r@{\,}c@{\,}|@{\,}r@{\,}c@{\,}|@{\,}r@{\,}c@{}}
 \hline
$q~~$ & $n^\text{L}_q(4,3)$ &
$q~~$ & $n^\text{L}_q(4,3)$ &
$q~~$ & $n^\text{L}_q(4,3)$ &
$q~~$ & $n^\text{L}_q(4,3)$ &
$q~~$ & $n^\text{L}_q(4,3)$ &
$q~~$ & $n^\text{L}_{q_{\vphantom{H_{L}}}}(4,3)^{ \vphantom{H^{L}}}$ \\
\hline
 4583 & 89  &  4591 & 89  &  4597 & 89  &  4603 & 90  &  4621 & 89  &  4637 & 89  \\
 4639 & 89  &  4643 & 90  &  4649 & 89  &  4651 & 89  &  4657 & 90  &  4663 & 90  \\
 4673 & 90  &  4679 & 92  &  4691 & 90  &  4703 & 89  &  4721 & 90  &  4723 & 90  \\
 4729 & 90  &  4733 & 90  &  4751 & 90  &  4759 & 90  &  4783 & 90  &  4787 & 89  \\
 4789 & 89  &  4793 & 89  &  4799 & 91  &  4801 & 92  &  4813 & 92  &  4817 & 89  \\
 4831 & 92  &  4861 & 91  &  4871 & 90  &  4877 & 92  &  4889 & 92  &  4903 & 91  \\
 4909 & 91  &  4913 & 91  &  4919 & 90  &  4931 & 91  &  4933 & 91  &  4937 & 90  \\
 4943 & 91  &  4951 & 91  &  4957 & 90  &  4967 & 91  &  4969 & 91  &  4973 & 91  \\
 4987 & 90  &  4993 & 92  &  4999 & 92  &  5003 & 92  &  5009 & 92  &  5011 & 93  \\
 5021 & 91  &  5023 & 92  &  5039 & 93  &  5041 & 91  &  5051 & 91  &  5059 & 92  \\
 5077 & 91  &  5081 & 92  &  5087 & 92  &  5099 & 94  &  5101 & 92  &  5107 & 93  \\
 5113 & 94  &  5119 & 91  &  5147 & 92  &  5153 & 93  &  5167 & 94  &  5171 & 93  \\
 5179 & 93  &  5189 & 93  &  5197 & 93  &  5209 & 93  &  5227 & 92  &  5231 & 94  \\
 5233 & 93  &  5237 & 93  &  5261 & 93  &  5273 & 94  &  5279 & 95  &  5281 & 94  \\
 5297 & 94  &  5303 & 95  &  5309 & 94  &  5323 & 93  &  5329 & 94  &  5333 & 94  \\
 5347 & 94  &  5351 & 95  &  5381 & 94  &  5387 & 94  &  5393 & 95  &  5399 & 95  \\
 5407 & 95  &  5413 & 94  &  5417 & 94  &  5419 & 95  &  5431 & 95  &  5437 & 93  \\
 5441 & 94  &  5443 & 94  &  5449 & 93  &  5471 & 94  &  5477 & 94  &  5479 & 95  \\
 5483 & 95  &  5501 & 96  &  5503 & 95  &  5507 & 94  &  5519 & 96  &  5521 & 95  \\
 5527 & 96  &  5531 & 95  &  5557 & 94  &  5563 & 95  &  5569 & 95  &  5573 & 95  \\
 5581 & 94  &  5591 & 96  &  5623 & 97  &  5639 & 96  &  5641 & 97  &  5647 & 97  \\
 5651 & 97  &  5653 & 97  &  5657 & 97  &  5659 & 96  &  5669 & 96  &  5683 & 98  \\
 5689 & 96  &  5693 & 97  &  5701 & 96  &  5711 & 96  &  5717 & 97  &  5737 & 96  \\
 5741 & 95  &  5743 & 97  &  5749 & 97  &  5779 & 96  &  5783 & 96  &  5791 & 97  \\
 5801 & 98  &  5807 & 96  &  5813 & 97  &  5821 & 97  &  5827 & 97  &  5839 & 98  \\
 5843 & 97  &  5849 & 96  &  5851 & 97  &  5857 & 97  &  5861 & 97  &  5867 & 97  \\
 5869 & 98  &  5879 & 97  &  5881 & 98  &  5897 & 97  &  5903 & 97  &  5923 & 97  \\
 5927 & 97  &  5939 & 98  &  5953 & 98  &  5981 & 98  &  5987 & 100 &  6007 & 98  \\
 6011 & 98  &  6029 & 97  &  6037 & 98  &  6043 & 99  &  6047 & 98  &  6053 & 99  \\
 6067 & 99  &  6073 & 99  &  6079 & 98  &  6089 & 99  &  6091 & 98  &  6101 & 98  \\
 6113 & 99  &  6121 & 99  &  6131 & 98  &  6133 & 97  &  6143 & 100 &  6151 & 98  \\
 6163 & 99  &  6173 & 99  &  6197 & 100 &  6199 & 100 &  6203 & 98  &  6211 & 100  \\
 6217 & 101  &  6221 & 100  &  6229 & 99  &  6241 & 100  &  6247 & 99  &  6257 & 100  \\
 6263 & 100  &  6269 & 100  &  6271 & 100  &  6277 & 98  &  6287 & 101  &  6299 & 101  \\
 6301 & 99  &  6311 & 99  &  6317 & 100  &  6323 & 100  &  6329 & 100  &  6337 & 101  \\
 6343 & 100  &  6353 & 100  &  6359 & 100  &  6361 & 99  &  6367 & 101  &  6373 & 100  \\
 6379 & 100  &  6389 & 101  &  6397 & 101  &  6421 & 101  &  6427 & 101  &  6449 & 101  \\
 6451 & 101  &  6469 & 100  &  6473 & 101  &  6481 & 101  &  6491 & 101  &  6521 & 101  \\
 6529 & 103  &  6547 & 102  &  6551 & 102  &  6553 & 101  &  6561 & 102  &  6563 & 103  \\
 6569 & 101  &  6571 & 102  &  6577 & 101  &  6581 & 101  &  6599 & 101  &  6607 & 101\\
 \hline
 \end{tabular}
\newpage

\noindent\textbf{Table 1.} Continue 4\medskip

\renewcommand{\arraystretch}{0.95}
\noindent \begin{tabular}
{@{}r@{\,}c@{\,}|@{\,}r@{\,}c@{\,}|@{\,}r@{\,}c@{\,}|@{\,}r@{\,}c@{\,}|@{\,}r@{\,}c@{\,}|@{\,}r@{\,}c@{}}
 \hline
$q~~$ & $n^\text{L}_q(4,3)$ &
$q~~$ & $n^\text{L}_q(4,3)$ &
$q~~$ & $n^\text{L}_q(4,3)$ &
$q~~$ & $n^\text{L}_q(4,3)$ &
$q~~$ & $n^\text{L}_q(4,3)$ &
$q~~$ & $n^\text{L}_{q_{\vphantom{H_{L}}}}(4,3)^{ \vphantom{H^{L}}}$ \\
\hline
 6619 & 101  &  6637 & 103  &  6653 & 103  &  6659 & 102  &  6661 & 102  &  6673 & 102  \\
 6679 & 103  &  6689 & 102  &  6691 & 102  &  6701 & 102  &  6703 & 100  &  6709 & 102  \\
 6719 & 102  &  6733 & 104  &  6737 & 102  &  6761 & 101  &  6763 & 102  &  6779 & 103  \\
 6781 & 103  &  6791 & 102  &  6793 & 103  &  6803 & 102  &  6823 & 103  &  6827 & 103  \\
 6829 & 102  &  6833 & 103  &  6841 & 101  &  6857 & 103  &  6859 & 102  &  6863 & 102  \\
 6869 & 103  &  6871 & 105  &  6883 & 104  &  6889 & 102  &  6899 & 103  &  6907 & 103  \\
 6911 & 104  &  6917 & 103  &  6947 & 102  &  6949 & 103  &  6959 & 105  &  6961 & 103  \\
 6967 & 104  &  6971 & 105  &  6977 & 103  &  6983 & 103  &  6991 & 104  &  6997 & 103  \\
 7001 & 105  &  7013 & 104  &  7019 & 104  &  7027 & 105  &  7039 & 104  &  7043 & 103  \\
 7057 & 105  &  &  &  &  &  &  &  &  &  &  \\
 \hline
 \end{tabular}
\newpage

\noindent\textbf{Table 2.}   Lengths $n^\text{L}_q(5,3)$ of the $[n^\text{L}_q(5,3),n^\text{L}_q(5,3)-5,5]_q3$ leximatrix quasi-perfect Almost MDS codes,  $3\leq q\leq 839$  \medskip

 \begin{tabular}
{@{}r@{}r@{\,\,}|@{\,}r@{}r@{\,\,}|@{\,}r@{}r@{\,\,}|@{\,}r@{}r@{\,\,}|@{\,}r@{}r@{\,\,}|@{\,}r@{}r@{\,\,}|@{\,\,}r@{}r@{}}
 \hline
$q~~$ & $n^\text{L}_q(5,3)$ &
$q~~$ & $n^\text{L}_q(5,3)$ &
$q~~$ & $n^\text{L}_q(5,3)$ &
$q~~$ & $n^\text{L}_q(5,3)$ &
$q~~$ & $n^\text{L}_q(5,3)$ &
$q~~$ & $n^\text{L}_q(5,3)$ &
$q~~$ & $n^\text{L}_{q_{\vphantom{H_{L}}}}(5,3)^{ \vphantom{H^{L}}}$ \\
 \hline
 3 & 11  &  4 & 10  &  5 & 11  &  7 & 16  &  8 & 17  &  9 & 19  &  11 & 22  \\
 13 & 24  &  16 & 28  &  17 & 28  &  19 & 31  &  23 & 36  &  25 & 37  &  27 & 40  \\
 29 & 43  &  31 & 46  &  32 & 46  &  37 & 51  &  41 & 55  &  43 & 56  &  47 & 60  \\
 49 & 61  &  53 & 66  &  59 & 70  &  61 & 73  &  64 & 77  &  67 & 79  &  71 & 82  \\
 73 & 84  &  79 & 88  &  81 & 88  &  83 & 90  &  89 & 96  &  97 & 101  &  101 & 104  \\
 103 & 107  &  107 & 109  &  109 & 111  &  113 & 112  &  121 & 119  &  125 & 123  &  127 & 123  \\
 128 & 124  &  131 & 127  &  137 & 130  &  139 & 133  &  149 & 142  &  151 & 141  &  157 & 146  \\
 163 & 149  &  167 & 150  &  169 & 151  &  173 & 156  &  179 & 158  &  181 & 159  &  191 & 166  \\
 193 & 166  &  197 & 171  &  199 & 172  &  211 & 180  &  223 & 185  &  227 & 186  &  229 & 188  \\
 233 & 191  &  239 & 195  &  241 & 197  &  243 & 198  &  251 & 203  &  256 & 205  &  257 & 207  \\
 263 & 208  &  269 & 214  &  271 & 213  &  277 & 215  &  281 & 218  &  283 & 221  &  289 & 226  \\
 293 & 227  &  307 & 232  &  311 & 234  &  313 & 236  &  317 & 237  &  331 & 245  &  337 & 248  \\
 343 & 253  &  347 & 257  &  349 & 255  &  353 & 256  &  359 & 260  &  361 & 260  &  367 & 265  \\
 373 & 266  &  379 & 274  &  383 & 272  &  389 & 275  &  397 & 280  &  401 & 282  &  409 & 284  \\
 419 & 292  &  421 & 290  &  431 & 297  &  433 & 299  &  439 & 301  &  443 & 304  &  449 & 309  \\
 457 & 311  &  461 & 311  &  463 & 309  &  467 & 314  &  479 & 320  &  487 & 324  &  491 & 324  \\
 499 & 328  &  503 & 330  &  509 & 334  &  512 & 334  &  521 & 339  &  523 & 341  &  529 & 344  \\
 541 & 348  &  547 & 349  &  557 & 353  &  563 & 360  &  569 & 364  &  571 & 362  &  577 & 365  \\
 587 & 371  &  593 & 374  &  599 & 375  &  601 & 376  &  607 & 376  &  613 & 380  &  617 & 384  \\
 619 & 382  &  625 & 385  &  631 & 387  &  641 & 393  &  643 & 398  &  647 & 396  &  653 & 399  \\
 659 & 402  &  661 & 401  &  673 & 407  &  677 & 407  &  683 & 411  &  691 & 416  &  701 & 417  \\
 709 & 424  &  719 & 427  &  727 & 430  &  729 & 430  &  733 & 429  &  739 & 431  &  743 & 436  \\
 751 & 439  &  757 & 440  &  761 & 443  &  769 & 447  &  773 & 450  &  787 & 453  &  797 & 458  \\
 809 & 464  &  811 & 464  &  821 & 467  &  823 & 468  &  827 & 471  &  829 & 473  &  839 & 475  \\
 \hline
 \end{tabular}
 \newpage

 \noindent\textbf{Table 3.} Lengths $n^\text{IL}_q(4,3)$ of the $[n^\text{IL}_q(4,3),n^\text{IL}_q(4,3)-4,5]_q3$ invleximatrix quasi-perfect MDS codes, $7\le q\leq 6203$ and $q=6217,6287,6299,6529,6563,6637,6653,6733,6871,$
$6959,6971,7001$. The cases $n^\text{IL}_q(4,3)<n^\text{L}_q(4,3)$ are noted in bold italic font\medskip

\noindent \begin{tabular}
{@{}r@{}c@{\,}|@{\,}r@{}c@{\,}|@{\,}r@{}c@{\,}|@{\,}r@{}c@{\,}|@{\,}r@{}c@{\,}|@{\,}r@{}c@{}}
 \hline
$q~~$ & $n^\text{IL}_q(4,3)$ &
$q~~$ & $n^\text{IL}_q(4,3)$ &
$q~~$ & $n^\text{IL}_q(4,3)$ &
$q~~$ & $n^\text{IL}_q(4,3)$ &
$q~~$ & $n^\text{IL}_q(4,3)$ &
$q~~$ & $n^\text{IL}_q(4,3)^{ \vphantom{H^{L}}}$ \\
 \hline
 7 & 8  &  8 & 7  &  \textbf{\emph{9}} & \textbf{\emph{8}}  &  11 & 9  &  13 & 10  &  16 & 10  \\
 17 & 11  &  19 & 10  &  23 & 12  &  25 & 11  &  27 & 12  &  29 & 13  \\
 31 & 13  &  32 & 12  &  37 & 14  &  41 & 14  &  43 & 14  &  47 & 15  \\
 49 & 15  &  53 & 16  &  59 & 16  &  61 & 17  &  64 & 17  &  67 & 17  \\
 \textbf{\emph{71}} & \textbf{\emph{17}}  &  73 & 18  &  79 & 19  &  81 & 19  &  83 & 19  &  \textbf{\emph{89}} & \textbf{\emph{19}}  \\
 97 & 21  &  \textbf{\emph{101}} & \textbf{\emph{20}}  &  103 & 21  &  \textbf{\emph{107}} & \textbf{\emph{21}}  &  \textbf{\emph{109}} & \textbf{\emph{21}}  &  113 & 23  \\
 121 & 22  &  \textbf{\emph{125}} & \textbf{\emph{22}}  &  127 & 23  &  128 & 23  &  \textbf{\emph{131}} & \textbf{\emph{22}}  &  137 & 23  \\
 139 & 23  &  149 & 24  &  151 & 24  &  \textbf{\emph{157}} & \textbf{\emph{24}}  &  163 & 25  &  167 & 25  \\
 \textbf{\emph{169}} & \textbf{\emph{24}}  &  173 & 26  &  179 & 26  &  \textbf{\emph{181}} & \textbf{\emph{25}}  &  191 & 27  &  \textbf{\emph{193}} & \textbf{\emph{26}}  \\
 197 & 28  &  199 & 26  &  211 & 28  &  \textbf{\emph{223}} & \textbf{\emph{28}}  &  227 & 29  &  229 & 29  \\
 233 & 28  &  239 & 29  &  241 & 29  &  243 & 29  &  251 & 30  &  256 & 29  \\
 257 & 30  &  263 & 31  &  269 & 30  &  \textbf{\emph{271}} & \textbf{\emph{30}}  &  277 & 31  &  281 & 30  \\
 283 & 31  &  289 & 31  &  293 & 32  &  \textbf{\emph{307}} & \textbf{\emph{31}}  &  311 & 32  &  313 & 32  \\
 317 & 32  &  \textbf{\emph{331}} & \textbf{\emph{32}}  &  \textbf{\emph{337}} & \textbf{\emph{32}}  &  343 & 34  &  \textbf{\emph{347}} & \textbf{\emph{33}}  &  349 & 34  \\
 353 & 34  &  359 & 34  &  361 & 34  &  367 & 34  &  373 & 35  &  379 & 35  \\
 383 & 35  &  389 & 35  &  397 & 35  &  401 & 35  &  409 & 36  &  419 & 36  \\
 421 & 36  &  431 & 36  &  433 & 37  &  \textbf{\emph{439}} & \textbf{\emph{37}}  &  \textbf{\emph{443}} & \textbf{\emph{37}}  &  449 & 37  \\
 457 & 37  &  \textbf{\emph{461}} & \textbf{\emph{36}}  &  463 & 38  &  467 & 38  &  \textbf{\emph{479}} & \textbf{\emph{37}}  &  487 & 38  \\
 \textbf{\emph{491}} & \textbf{\emph{38}}  &  \textbf{\emph{499}} & \textbf{\emph{38}}  &  \textbf{\emph{503}} & \textbf{\emph{37}}  &  509 & 39  &  512 & 39  &  521 & 40  \\
 523 & 39  &  529 & 40  &  541 & 40  &  547 & 40  &  557 & 40  &  \textbf{\emph{563}} & \textbf{\emph{40}}  \\
 569 & 41  &  571 & 40  &  577 & 41  &  \textbf{\emph{587}} & \textbf{\emph{40}}  &  593 & 41  &  599 & 41  \\
 601 & 42  &  607 & 43  &  613 & 43  &  \texttt{\emph{617}} & \textbf{\emph{41}}  &  619 & 42  &  \textbf{\emph{625}} & \textbf{\emph{41}}  \\
 631 & 42  &  641 & 43  &  643 & 42  &  \textbf{\emph{647}} & \textbf{\emph{42}}  &  \textbf{\emph{653}} & \textbf{\emph{43}}  &  \textbf{\emph{659}} & \textbf{\emph{43}}  \\
 661 & 43  &  \textbf{\emph{673}} & \textbf{\emph{42}}  &  677 & 44  &  683 & 43  &  691 & 45  &  \textbf{\emph{701}} & \textbf{\emph{43}}  \\
 \textbf{\emph{709}} & \textbf{\emph{43}}  &  719 & 44  &  \textbf{\emph{727}} & \textbf{\emph{44}}  &  729 & 44  &  733 & 45  &  \textbf{\emph{739}} & \textbf{\emph{44}}  \\
 743 & 46  &  \textbf{\emph{751}} & \textbf{\emph{44}}  &  \textbf{\emph{757}} & \textbf{\emph{45}}  &  761 & 45  &  \textbf{\emph{769}} & \textbf{\emph{45}}  &  \textbf{\emph{773}} & \textbf{\emph{44}}  \\
 787 & 45  &  797 & 46  &  809 & 47  &  811 & 47  &  821 & 46  &  823 & 48  \\
 827 & 47  &  829 & 47  &  839 & 47  &  \textbf{\emph{841}} & \textbf{\emph{46}}  &  853 & 48  &  857 & 48  \\
 859 & 48  &  863 & 48  &  877 & 48  &  881 & 49  &  883 & 48  &  887 & 48  \\
 \textbf{\emph{907}} & \textbf{\emph{49}}  &  \textbf{\emph{911}} & \textbf{\emph{48}}  &  919 & 48  &  929 & 49  &  937 & 49  &  941 & 49  \\
 947 & 50  &  953 & 49  &  961 & 50  &  \textbf{\emph{967}} & \textbf{\emph{49}}  &  \textbf{\emph{971}} & \textbf{\emph{49}}  &  977 & 50  \\
 983 & 50  &  991 & 50  &  \textbf{\emph{997}} & \textbf{\emph{50}}  &  1009 & 51  &  \textbf{\emph{1013}} & \textbf{\emph{50}}  &  \textbf{\emph{1019}} & \textbf{\emph{50}}  \\
 1021 & 51  &  \textbf{\emph{1024}} & \textbf{\emph{50}}  &  1031 & 51  &  1033 & 51  &  1039 & 52  &  \textbf{\emph{1049}} & \textbf{\emph{51}}  \\
 1051 & 51  &  1061 & 51  &  1063 & 52  &  \textbf{\emph{1069}} & \textbf{\emph{51}}  &  1087 & 52  &  1091 & 52  \\
1093 & 52  &  1097 & 52  &  1103 & 52  &  1109 & 52  &  1117 & 53  &  1123 & 52  \\
 \textbf{\emph{1129}} & \textbf{\emph{52}}  &  \textbf{\emph{1151}} & \textbf{\emph{52}}  &  1153 & 54  &  1163 & 53  &  1171 & 53  &  1181 & 53  \\
 \textbf{\emph{1187}} & \textbf{\emph{53}}  &  1193 & 54  &  1201 & 54  &  1213 & 54  &  \textbf{\emph{1217}} & \textbf{\emph{54}}  &  \textbf{\emph{1223}} & \textbf{\emph{54}}  \\
 \hline
 \end{tabular}
 \newpage

\textbf{Table 3.} Continue 1 \medskip

\noindent \begin{tabular}
{@{}r@{}c@{\,}|@{\,}r@{}c@{\,}|@{\,}r@{}c@{\,}|@{\,}r@{}c@{\,}|@{\,}r@{}c@{\,}|@{\,}r@{}c@{}}
\hline
$q~~$ & $n^\text{IL}_q(4,3)$ &
$q~~$ & $n^\text{IL}_q(4,3)$ &
$q~~$ & $n^\text{IL}_q(4,3)$ &
$q~~$ & $n^\text{IL}_q(4,3)$ &
$q~~$ & $n^\text{IL}_q(4,3)$ &
$q~~$ & $n^\text{IL}_q(4,3)^{ \vphantom{H^{L}}}$ \\
 \hline
 1229 & 54  &  \textbf{\emph{1231}} & \textbf{\emph{53}}  &  \textbf{\emph{1237}} & \textbf{\emph{54}}  &  \textbf{\emph{1249}} & \textbf{\emph{54}}  &  1259 & 55  &  \textbf{\emph{1277}} & \textbf{\emph{54}}  \\
 \textbf{\emph{1279}} & \textbf{\emph{55}}  &  \textbf{\emph{1283}} & \textbf{\emph{55}}  &  \textbf{\emph{1289}} & \textbf{\emph{54}}  &  1291 & 56  &  1297 & 57  &  \textbf{\emph{1301}} & \textbf{\emph{55}}  \\
 \textbf{\emph{1303}} & \textbf{\emph{55}}  &  1307 & 56  &  \textbf{\emph{1319}} & \textbf{\emph{55}}  &  1321 & 57  &  1327 & 56  &  1331 & 56  \\
 \textbf{\emph{1361}} & \textbf{\emph{56}}  &  \textbf{\emph{1367}} & \textbf{\emph{56}}  &  1369 & 56  &  1373 & 58  &  1381 & 57  &  \textbf{\emph{1399}} & \textbf{\emph{56}}  \\
 \textbf{\emph{1409}} & \textbf{\emph{56}}  &  \textbf{\emph{1423}} & \textbf{\emph{57}}  &  1427 & 59  &  1429 & 58  &  1433 & 57  &  1439 & 57  \\
 1447 & 58  &  \textbf{\emph{1451}} & \textbf{\emph{58}}  &  \textbf{\emph{1453}} & \textbf{\emph{58}}  &  1459 & 58  &  1471 & 59  &  1481 & 59  \\
 1483 & 59  &  1487 & 59  &  1489 & 59  &  1493 & 58  &  1499 & 59  &  1511 & 60  \\
 1523 & 59  &  \textbf{\emph{1531}} & \textbf{\emph{59}}  &  1543 & 60  &  \textbf{\emph{1549}} & \textbf{\emph{58}}  &  1553 & 59  &  1559 & 60  \\
 1567 & 60  &  1571 & 60  &  1579 & 60  &  1583 & 59  &  1597 & 60  &  1601 & 60  \\
 1607 & 60  &  \textbf{\emph{1609}} & \textbf{\emph{59}}  &  1613 & 60  &  1619 & 61  &  1621 & 60  &  1627 & 61  \\
 1637 & 60  &  1657 & 61  &  \textbf{\emph{1663}} & \textbf{\emph{60}}  &  \textbf{\emph{1667}} & \textbf{\emph{60}}  &  1669 & 61  &  \textbf{\emph{1681}} & \textbf{\emph{61}}  \\
 \textbf{\emph{1693}} & \textbf{\emph{60}}  &  \textbf{\emph{1697}} & \textbf{\emph{61}}  &  \textbf{\emph{1699}} & \textbf{\emph{61}}  &  1709 & 62  &  1721 & 63  &  1723 & 62  \\
 \textbf{\emph{1733}} & \textbf{\emph{62}}  &  1741 & 62  &  1747 & 63  &  1753 & 62  &  1759 & 63  &  1777 & 62  \\
 \textbf{\emph{1783}} & \textbf{\emph{62}}  &  \textbf{\emph{1787}} & \textbf{\emph{62}}  &  1789 & 62  &  1801 & 63  &  1811 & 64  &  1823 & 63  \\
 1831 & 63  &  1847 & 63  &  \textbf{\emph{1849}} & \textbf{\emph{63}}  &  1861 & 63  &  1867 & 64  &  1871 & 63  \\
 \textbf{\emph{1873}} & \textbf{\emph{63}}  &  1877 & 64  &  1879 & 63  &  1889 & 63  &  1901 & 64  &  1907 & 64  \\
 1913 & 64  &  \textbf{\emph{1931}} & \textbf{\emph{64}}  &  \textbf{\emph{1933}} & \textbf{\emph{65}}  &  1949 & 64  &  \textbf{\emph{1951}} & \textbf{\emph{64}}  &  \textbf{\emph{1973}} & \textbf{\emph{64}}  \\
 \textbf{\emph{1979}} & \textbf{\emph{64}}  &  1987 & 66  &  1993 & 66  &  \textbf{\emph{1997}} & \textbf{\emph{65}}  &  1999 & 66  &  \textbf{\emph{2003}} & \textbf{\emph{66}}  \\
 \textbf{\emph{2011}} & \textbf{\emph{65}}  &  2017 & 66  &  2027 & 67  &  2029 & 66  &  2039 & 66  &  2048 & 67  \\
 2053 & 66  &  2063 & 67  &  \textbf{\emph{2069}} & \textbf{\emph{65}}  &  2081 & 66  &  2083 & 66  &  \textbf{\emph{2087}} & \textbf{\emph{66}}  \\
 \textbf{\emph{2089}} & \textbf{\emph{66}}  &  2099 & 67  &  \textbf{\emph{2111}} & \textbf{\emph{66}}  &  2113 & 66  &  2129 & 67  &  2131 & 68  \\
 \textbf{\emph{2137}} & \textbf{\emph{66}}  &  2141 & 68  &  2143 & 67  &  2153 & 67  &  \textbf{\emph{2161}} & \textbf{\emph{66}}  &  2179 & 68  \\
 2187 & 68  &  \textbf{\emph{2197}} & \textbf{\emph{66}}  &  2203 & 68  &  \textbf{\emph{2207}} & \textbf{\emph{67}}  &  2209 & 69  &  \textbf{\emph{2213}} & \textbf{\emph{67}}  \\
 2221 & 69  &  2237 & 68  &  2239 & 69  &  \textbf{\emph{2243}} & \textbf{\emph{68}}  &  \textbf{\emph{2251}} & \textbf{\emph{68}}  &  \textbf{\emph{2267}} & \textbf{\emph{67}}  \\
 \textbf{\emph{2269}} & \textbf{\emph{68}}  &  \textbf{\emph{2273}} & \textbf{\emph{68}}  &  2281 & 70  &  2287 & 69  &  2293 & 68  &  2297 & 69  \\
 2309 & 70  &  2311 & 69  &  2333 & 69  &  \textbf{\emph{2339}} & \textbf{\emph{68}}  &  2341 & 70  &  2347 & 70  \\
 \textbf{\emph{2351}} & \textbf{\emph{68}}  &  2357 & 70  &  \textbf{\emph{2371}} & \textbf{\emph{69}}  &  2377 & 70  &  2381 & 69  &  \textbf{\emph{2383}} & \textbf{\emph{70}}  \\
 2389 & 70  &  2393 & 70  &  \textbf{\emph{2399}} & \textbf{\emph{69}}  &  2401 & 70  &  2411 & 71  &  2417 & 71  \\
 \textbf{\emph{2423}} & \textbf{\emph{70}}  &  \textbf{\emph{2437}} & \textbf{\emph{70}}  &  \textbf{\emph{2441}} & \textbf{\emph{70}}  &  2447 & 71  &  2459 & 71  &  2467 & 71  \\
 \textbf{\emph{2473}} & \textbf{\emph{70}}  &  \textbf{\emph{2477}} & \textbf{\emph{70}}  &  2503 & 70  &  2521 & 72  &  2531 & 72  &  2539 & 72  \\
 2543 & 72  &  2549 & 71  &  2551 & 71  &  2557 & 72  &  2579 & 73  &  2591 & 72  \\
 \textbf{\emph{2593}} & \textbf{\emph{71}}  &  2609 & 72  &  \textbf{\emph{2617}} & \textbf{\emph{71}}  &  \textbf{\emph{2621}} & \textbf{\emph{71}}  &  \textbf{\emph{2633}} & \textbf{\emph{72}}  &  2647 & 74  \\
 2657 & 73  &  2659 & 74  &  2663 & 72  &  2671 & 73  &  2677 & 73  &  2683 & 73  \\
 2687 & 72  &  2689 & 73  &  2693 & 73  &  2699 & 72  &  2707 & 73  &  2711 & 74  \\
 2713 & 73  &  2719 & 73  &  2729 & 74  &  \textbf{\emph{2731}} & \textbf{\emph{73}}  &  2741 & 73  &  2749 & 74  \\
 \textbf{\emph{2753}} & \textbf{\emph{73}}  &  2767 & 74  &  \textbf{\emph{2777}} & \textbf{\emph{73}}  &  2789 & 74  &  \textbf{\emph{2791}} & \textbf{\emph{72}}  &  2797 & 74  \\
 \textbf{\emph{2801}} & \textbf{\emph{74}}  &  \textbf{\emph{2803}} & \textbf{\emph{73}}  &  2809 & 74  &  \textbf{\emph{2819}} & \textbf{\emph{73}}  &  2833 & 75  &  2837 & 76  \\
 \hline
 \end{tabular}
 \newpage

\textbf{Table 3.} Continue 2 \medskip

\noindent \begin{tabular}
{@{}r@{}c@{\,}|@{\,}r@{}c@{\,}|@{\,}r@{}c@{\,}|@{\,}r@{}c@{\,}|@{\,}r@{}c@{\,}|@{\,}r@{}c@{}}
  \hline
$q~~$ & $n^\text{IL}_q(4,3)$ &
$q~~$ & $n^\text{IL}_q(4,3)$ &
$q~~$ & $n^\text{IL}_q(4,3)$ &
$q~~$ & $n^\text{IL}_q(4,3)$ &
$q~~$ & $n^\text{IL}_q(4,3)$ &
$q~~$ & $n^\text{IL}_q(4,3)^{ \vphantom{H^{L}}}$ \\
 \hline
 2843 & 76  &  \textbf{\emph{2851}} & \textbf{\emph{74}}  &  2857 & 74  &  \textbf{\emph{2861}} & \textbf{\emph{73}}  &  2879 & 75  &  \textbf{\emph{2887}} & \textbf{\emph{75}}  \\
 2897 & 75  &  2903 & 75  &  2909 & 76  &  2917 & 76  &  \textbf{\emph{2927}} & \textbf{\emph{74}}  &  \textbf{\emph{2939}} & \textbf{\emph{75}}  \\
 \textbf{\emph{2953}} & \textbf{\emph{76}}  &  \textbf{\emph{2957}} & \textbf{\emph{75}}  &  2963 & 76  &  2969 & 75  &  2971 & 76  &  2999 & 76  \\
 3001 & 76  &  3011 & 76  &  3019 & 78  &  3023 & 77  &  3037 & 76  &  3041 & 77  \\
 3049 & 76  &  3061 & 76  &  3067 & 77  &  \textbf{\emph{3079}} & \textbf{\emph{77}}  &  \textbf{\emph{3083}} & \textbf{\emph{76}}  &  3089 & 77  \\
 3109 & 77  &  \textbf{\emph{3119}} & \textbf{\emph{76}}  &  3121 & 77  &  \textbf{\emph{3125}} & \textbf{\emph{76}}  &  3137 & 78  &  \textbf{\emph{3163}} & \textbf{\emph{77}}  \\
 3167 & 77  &  3169 & 79  &  \textbf{\emph{3181}} & \textbf{\emph{78}}  &  3187 & 77  &  \textbf{\emph{3191}} & \textbf{\emph{77}}  &  3203 & 77  \\
 3209 & 77  &  3217 & 78  &  3221 & 80  &  3229 & 78  &  \textbf{\emph{3251}} & \textbf{\emph{78}}  &  3253 & 78  \\
 3257 & 78  &  3259 & 78  &  3271 & 79  &  3299 & 79  &  3301 & 79  &  3307 & 80  \\
 3313 & 79  &  3319 & 80  &  3323 & 79  &  \textbf{\emph{3329}} & \textbf{\emph{79}}  &  \textbf{\emph{3331}} & \textbf{\emph{78}}  &  3343 & 80  \\
 \textbf{\emph{3347}} & \textbf{\emph{78}}  &  3359 & 81  &  \textbf{\emph{3361}} & \textbf{\emph{79}}  &  3371 & 81  &  3373 & 79  &  3389 & 80  \\
 3391 & 79  &  3407 & 80  &  3413 & 81  &  3433 & 80  &  3449 & 80  &  3457 & 81  \\
 3461 & 80  &  3463 & 80  &  3467 & 80  &  3469 & 80  &  \textbf{\emph{3481}} & \textbf{\emph{79}}  &  3491 & 80  \\
 3499 & 80  &  3511 & 81  &  3517 & 80  &  3527 & 81  &  \textbf{\emph{3529}} & \textbf{\emph{81}}  &  3533 & 82  \\
 3539 & 82  &  3541 & 80  &  3547 & 81  &  \textbf{\emph{3557}} & \textbf{\emph{81}}  &  3559 & 82  &  \textbf{\emph{3571}} & \textbf{\emph{80}}  \\
 \textbf{\emph{3581}} & \textbf{\emph{80}}  &  3583 & 81  &  3593 & 81  &  \textbf{\emph{3607}} & \textbf{\emph{81}}  &  \textbf{\emph{3613}} & \textbf{\emph{79}}  &  3617 & 81  \\
 3623 & 82  &  3631 & 82  &  \textbf{\emph{3637}} & \textbf{\emph{81}}  &  3643 & 83  &  3659 & 82  &  \textbf{\emph{3671}} & \textbf{\emph{81}}  \\
 \textbf{\emph{3673}} & \textbf{\emph{81}}  &  3677 & 82  &  3691 & 83  &  \textbf{\emph{3697}} & \textbf{\emph{82}}  &  3701 & 83  &  \textbf{\emph{3709}} & \textbf{\emph{82}}  \\
 3719 & 82  &  3721 & 83  &  3727 & 82  &  3733 & 83  &  3739 & 84  &  3761 & 82  \\
 3767 & 83  &  3769 & 83  &  \textbf{\emph{3779}} & \textbf{\emph{82}}  &  3793 & 83  &  3797 & 84  &  3803 & 83  \\
 \textbf{\emph{3821}} & \textbf{\emph{82}}  &  3823 & 83  &  \textbf{\emph{3833}} & \textbf{\emph{82}}  &  3847 & 83  &  3851 & 84  &  3853 & 83  \\
 3863 & 84  &  \textbf{\emph{3877}} & \textbf{\emph{83}}  &  3881 & 84  &  3889 & 83  &  3907 & 85  &  \textbf{\emph{3911}} & \textbf{\emph{83}}  \\
 3917 & 83  &  3919 & 84  &  3923 & 85  &  3929 & 84  &  3931 & 85  &  3943 & 84  \\
 3947 & 84  &  3967 & 84  &  3989 & 85  &  4001 & 86  &  \textbf{\emph{4003}} & \textbf{\emph{83}}  &  \textbf{\emph{4007}} & \textbf{\emph{84}}  \\
 4013 & 85  &  \textbf{\emph{4019}} & \textbf{\emph{85}}  &  4021 & 84  &  4027 & 84  &  \textbf{\emph{4049}} & \textbf{\emph{83}}  &  \textbf{\emph{4051}} & \textbf{\emph{84}}  \\
 4057 & 85  &  4073 & 85  &  \textbf{\emph{4079}} & \textbf{\emph{85}}  &  4091 & 85  &  \textbf{\emph{4093}} & \textbf{\emph{85}}  &  \textbf{\emph{4096}} & \textbf{\emph{85}}  \\
 \textbf{\emph{4099}} & \textbf{\emph{85}}  &  \textbf{\emph{4111}} & \textbf{\emph{85}}  &  \textbf{\emph{4127}} & \textbf{\emph{85}}  &  \textbf{\emph{4129}} & \textbf{\emph{85}}  &  4133 & 86  &  \textbf{\emph{4139}} & \textbf{\emph{85}}  \\
 \textbf{\emph{4153}} & \textbf{\emph{85}}  &  4157 & 86  &  4159 & 86  &  \textbf{\emph{4177}} & \textbf{\emph{86}}  &  4201 & 85  &  \textbf{\emph{4211}} & \textbf{\emph{85}}  \\
4217 & 87  &  4219 & 87  &  4229 & 86  &  \textbf{\emph{4231}} & \textbf{\emph{86}}  &  4241 & 87  &  4243 & 87  \\
 4253 & 87  &  \textbf{\emph{4259}} & \textbf{\emph{87}}  &  4261 & 87  &  4271 & 88  &  \textbf{\emph{4273}} & \textbf{\emph{86}}  &  \textbf{\emph{4283}} & \textbf{\emph{86}}  \\
 4289 & 87  &  4297 & 87  &  \textbf{\emph{4327}} & \textbf{\emph{87}}  &  \textbf{\emph{4337}} & \textbf{\emph{87}}  &  4339 & 88  &  \textbf{\emph{4349}} & \textbf{\emph{87}}  \\
 4357 & 87  &  4363 & 87  &  4373 & 88  &  4391 & 87  &  4397 & 88  &  4409 & 88  \\
 4421 & 88  &  \textbf{\emph{4423}} & \textbf{\emph{88}}  &  4441 & 89  &  4447 & 88  &  4451 & 88  &  4457 & 87  \\
 4463 & 89  &  4481 & 89  &  4483 & 89  &  4489 & 89  &  4493 & 88  &  \textbf{\emph{4507}} & \textbf{\emph{88}}  \\
 4513 & 89  &  4517 & 88  &  4519 & 89  &  4523 & 89  &  4547 & 89  &  \textbf{\emph{4549}} & \textbf{\emph{89}}  \\
 \hline
 \end{tabular}
 \newpage

\textbf{Table 3.} Continue 3 \medskip

\noindent \begin{tabular}
{@{}r@{}c@{\,}|@{\,}r@{}c@{\,}|@{\,}r@{}c@{\,}|@{\,}r@{}c@{\,}|@{\,}r@{}c@{\,}|@{\,}r@{}c@{}}
  \hline
$q~~$ & $n^\text{IL}_q(4,3)$ &
$q~~$ & $n^\text{IL}_q(4,3)$ &
$q~~$ & $n^\text{IL}_q(4,3)$ &
$q~~$ & $n^\text{IL}_q(4,3)$ &
$q~~$ & $n^\text{IL}_q(4,3)$ &
$q~~$ & $n^\text{IL}_q(4,3)^{ \vphantom{H^{L}}}$ \\
 \hline
 4561 & 90  &  4567 & 90  &  4583 & 89  &  \textbf{\emph{4591}} & \textbf{\emph{88}}  &  4597 & 89  &  \textbf{\emph{4603}} & \textbf{\emph{88}}  \\
 \textbf{\emph{4621}} & \textbf{\emph{88}}  &  \textbf{\emph{4637}} & \textbf{\emph{88}}  &  4639 & 89  &  \textbf{\emph{4643}} & \textbf{\emph{89}}  &  4649 & 90  &  4651 & 90  \\
 4657 & 90  &  \textbf{\emph{4663}} & \textbf{\emph{89}}  &  4673 & 90  &  \textbf{\emph{4679}} & \textbf{\emph{90}}  &  \textbf{\emph{4691}} & \textbf{\emph{89}}  &  4703 & 90  \\
 4721 & 90  &  4723 & 91  &  \textbf{\emph{4729}} & \textbf{\emph{89}}  &  \textbf{\emph{4733}} & \textbf{\emph{89}}  &  4751 & 90  &  4759 & 90  \\
 4783 & 90  &  4787 & 90  &  4789 & 90  &  4793 & 90  &  4799 & 91  &  \textbf{\emph{4801}} & \textbf{\emph{89}}  \\
 \textbf{\emph{4813}} & \textbf{\emph{91}}  &  4817 & 91  &  \textbf{\emph{4831}} & \textbf{\emph{90}}  &  \textbf{\emph{4861}} & \textbf{\emph{89}}  &  4871 & 92  &  4877 & 92  \\
 4889 & 92  &  4903 & 92  &  \textbf{\emph{4909}} & \textbf{\emph{90}}  &  4913 & 92  &  4919 & 92  &  \textbf{\emph{4931}} & \textbf{\emph{90}}  \\
 4933 & 92  &  4937 & 91  &  4943 & 91  &  4951 & 91  &  4957 & 91  &  4967 & 91  \\
 4969 & 91  &  \textbf{\emph{4973}} & \textbf{\emph{90}}  &  4987 & 91  &  4993 & 93  &  4999 & 92  &  \textbf{\emph{5003}} & \textbf{\emph{91}}  \\
 5009 & 92  &  \textbf{\emph{5011}} & \textbf{\emph{92}}  &  5021 & 92  &  \textbf{\emph{5023}} & \textbf{\emph{91}}  & \textbf{\emph{ 5039}} & \textbf{\emph{91}}  &  5041 & 92  \\
 5051 & 93  &  5059 & 93  &  5077 & 93  &  5081 & 94  &  5087 & 92  &  \textbf{\emph{5099}} & \textbf{\emph{93}}  \\
 5101 & 94  &  5107 & 93  &  \textbf{\emph{5113}} & \textbf{\emph{93}}  &  5119 & 93  &  5147 & 95  &  5153 & 94  \\
 \textbf{\emph{5167}} & \textbf{\emph{92}}  &  5171 & 93  &  5179 & 93  &  5189 & 93  &  5197 & 94  &  \textbf{\emph{5209}} & \textbf{\emph{92}}  \\
 5227 & 93  &  \textbf{\emph{5231}} & \textbf{\emph{93}}  &  5233 & 94  &  5237 & 94  &  5261 & 93  &  \textbf{\emph{5273}} & \textbf{\emph{93}}  \\
 \textbf{\emph{5279}} & \textbf{\emph{94}}  &  \textbf{\emph{5281}} & \textbf{\emph{93}}  &  \textbf{\emph{5297}} & \textbf{\emph{93}}  &  \textbf{\emph{5303}} & \textbf{\emph{94}}  &  \textbf{\emph{5309}} & \textbf{\emph{93}}  &  5323 & 94  \\
 5329 & 94  &  \textbf{\emph{5333}} & \textbf{\emph{93}}  &  5347 & 95  &  \textbf{\emph{5351}} & \textbf{\emph{94}}  &  5381 & 95  &  \textbf{\emph{5387}} & \textbf{\emph{93}}  \\
 5393 & 95  &  5399 & 95  &  \textbf{\emph{5407}} & \textbf{\emph{94}}  &  5413 & 95  &  5417 & 95  &  5419 & 95  \\
\textbf{\emph{ 5431}} & \textbf{\emph{94}}  &  5437 & 96  &  5441 & 94  &  \textbf{\emph{5443}} & \textbf{\emph{93}}  &  5449 & 95  &  5471 & 95  \\
 5477 & 95  &  5479 & 97  &  5483 & 96  &  \textbf{\emph{5501}} & \textbf{\emph{95}}  &  5503 & 96  &  5507 & 94  \\
 5519 & 96  &  5521 & 96  &  \textbf{\emph{5527}} & \textbf{\emph{95}}  &  5531 & 96  &  5557 & 96  &  5563 & 96  \\
 5569 & 95  &  5573 & 96  &  5581 & 97  &  5591 & 96  &  \textbf{\emph{5623}} & 95  &  \textbf{\emph{5639}} & \textbf{\emph{95}}  \\
 \textbf{\emph{5641}} & \textbf{\emph{96}}  &  \textbf{\emph{5647}} & \textbf{\emph{96}}  &  \textbf{\emph{5651}} & \textbf{\emph{95}}  & \textbf{\emph{ 5653}} & \textbf{\emph{96}}  &  \textbf{\emph{5657}} & \textbf{\emph{95}}  &  5659 & 96  \\
 \textbf{\emph{5669}} & \textbf{\emph{95}}  &  \textbf{\emph{5683}} & \textbf{\emph{96}}  &  5689 & 96  &  5693 & 97  &  5701 & 96  &  5711 & 96  \\
 5717 & 98  &  5737 & 96  &  5741 & 97  &  \textbf{\emph{5743}} & \textbf{\emph{95}}  &  5749 & 98  &  5779 & 97  \\
 5783 & 97  &  \textbf{\emph{5791}} & \textbf{\emph{96}}  &  5801 & 99  &  5807 & 97  &  \textbf{\emph{5813}} & \textbf{\emph{96}}  &  5821 & 97  \\
 5827 & 98  &  \textbf{\emph{5839}} & \textbf{\emph{96}}  &  5843 & 98  &  5849 & 96  &  5851 & 97  &  5857 & 97  \\
 5861 & 97  &  5867 & 98  &  5869 & 98  &  5879 & 99  &  \textbf{\emph{5881}} & \textbf{\emph{97}}  &  5897 & 97  \\
 5903 & 98  &  5923 & 98  &  5927 & 97  &  \textbf{\emph{5939}} & \textbf{\emph{97}}  &  5953 & 98  &  \textbf{\emph{5981}} & \textbf{\emph{97}}\\
 \textbf{\emph{5987}} & \textbf{\emph{98}}  &  6007 & 99    &  6011 & 99  &  6029 & 98  &  6037 & 98  &  6043 & 99 \\
 6047 & 98  &  \textbf{\emph{6053}} & \textbf{\emph{98}}  &  6067 & 99  &  6073 & 100  &  6079 & 98  &  6089 & 100  \\
 \textbf{\emph{6091}} & \textbf{\emph{97}}  &  6101 & 99  &    6113 & 99  &  6121 & 99  &  6131 & 99  &  6133 & 98   \\
 \textbf{\emph{6143}} & \textbf{\emph{98}}&  6151 & 99  &  6163 & 100  &  \textbf{\emph{6173}} & \textbf{\emph{98}}  &  \textbf{\emph{6197}} & \textbf{\emph{99}}  &  6199 & 100  \\
 6203 & 100  &  &&&&&&&& & \\\hline
  \textbf{\emph{6217}} & \textbf{\emph{99}}  &  \textbf{\emph{6287}} & \textbf{\emph{100}}  & \textbf{\emph{6299}} & \textbf{\emph{100}}  & \textbf{\emph{6529}} & \textbf{\emph{100}}& \textbf{\emph{6563}} &  \textbf{\emph{100}}&\textbf{\emph{6637}} & \textbf{\emph{102}} \\
 6653 & 103  &  \textbf{\emph{6733}} & \textbf{\emph{102}}  &  \textbf{\emph{6871}} & \textbf{\emph{103}}  &  \textbf{\emph{6959}} & \textbf{\emph{103}}  &  \textbf{\emph{6971}} & \textbf{\emph{102}} & \textbf{\emph{7001}} & \textbf{\emph{104}} \\
 \hline
 \end{tabular}
\newpage

\textbf{Table 4.}   Lengths $\overline{n}_q(4,3)$ of the \textbf{\emph{shortest known}} $[\overline{n}_q(4,3),\overline{n}_q(4,3)-4]_q3$ codes, $2\le q\leq 7057$.
The cases $\overline{n}_q(4,3,5)=\overline{n}_q(4,3)+j$ are noted by the superscript ``$+j$''. For the rest of $q$ we have  $\overline{n}_q(4,3,5)=\overline{n}_q(4,3)$. The improvements of code distance up to $d=5$ in comparison with \cite[Tab.\,1]{DavOst-DESI2010} are noted in bold italic font. For $q=841$ the complete 42-arc of \cite{Sonin} is used. The cases $\ell_q(4,3)=\overline{n}_q(4,3)$ are noted by the subscript~``$\bullet$'' \medskip

 \renewcommand{\arraystretch}{0.97}
 \begin{tabular}
{@{}r@{\,}c@{\,}|r@{\,}c@{\,}|r@{\,}c@{\,}|r@{\,}c@{\,}|r@{\,}c@{\,}|r@{\,}c@{}}
 \hline
$q~~~$ & $\overline{n}_q(4,3)$ &
$q~~~$ & $\overline{n}_q(4,3)$ &
$q~~~$ & $\overline{n}_q(4,3)$ &
$q~~~$ & $\overline{n}_q(4,3)$ &
$q~~~$ & $\overline{n}_q(4,3)$ &
$q~~~$ & $\overline{n}_q(4,3)^{\vphantom{H^{L}}}$ \\
 \hline
 \textbf{\emph{2}} & \textbf{\emph{5}}$_\bullet$  &  3 & 5$_\bullet$  &  4 & 5$_\bullet$  &  5 & 6$_\bullet$  &  7 & 7$^{+1}_\bullet$  &  8 & 7$_\bullet$  \\
 9 & 7$^{+1}_\bullet$  &  11 & 8$_\bullet$  &  13 & 8  &  16 & 9  &  17 & 9  &  19 & 9  \\
 23 & 10  &  25 & 11  &  27 & 11  &  29 & 11  &  31 & 11$^{+1}$  &  32 & 12  \\
 37 & 12  &  41 & 13  &  43 & 13  &  47 & 14  &  49 & 14  &  53 & 15  \\
 59 & 15  &  61 & 15$^{+1}$  &  64 & 16  &  67 & 16  &  71 & 16  &  73 & 16$^{+1}$  \\
 79 & 17  &  81 & 17$^{+1}$  &  83 & 17$^{+1}$  &  89 & 18  &  97 & 19  &  101 & 19  \\
 103 & 19  &  107 & 19  &  109 & 20  &  113 & 20  &  121 & 20$^{+1}$  &  125 & 21  \\
 127 & 21  &  128 & 21  &  131 & 21  &  137 & 22  &  139 & 22  &  149 & 22  \\
 \textbf{\emph{151}} & \textbf{\emph{22}}  &  157 & 23  &  163 & 23  &  167 & 24  &  169 & 24  &  173 & 24  \\
 179 & 24  &  181 & 24$^{+1}$  &  191 & 25  &  193 & 25  &  197 & 25  &  199 & 25  \\
 211 & 26  &  223 & 27  &  227 & 27  &  229 & 27  &  233 & 27  &  239 & 27  \\
 241 & 28  &  243 & 28  &  251 & 28  &  256 & 28  &  257 & 28  &  263 & 28  \\
 269 & 29  &  271 & 29  &  277 & 29  &  281 & 29  &  283 & 29  &  289 & 29$^{+1}$  \\
 293 & 29$^{+1}$  &  307 & 30  &  311 & 30$^{+1}$  &  313 & 30$^{+1}$  &  317 & 30$^{+1}$  &  331 & 31  \\
 337 & 31$^{+1}$  &  343 & 31$^{+1}$  &  347 & 32  &  349 & 32  &  353 & 32  &  359 & 32  \\
 361 & 32$^{+1}$  &  367 & 32$^{+1}$  &  373 & 33  &  379 & 33  &  383 & 33  &  389 & 33$^{+1}$  \\
 397 & 34  &  401 & 34  &  409 & 34  &  419 & 34$^{+1}$  &  \textbf{\emph{421}} & \textbf{\emph{34}}  &  431 & 35  \\
 \textbf{\emph{433}} & \textbf{\emph{35}}  &  439 & 35  &  443 & 35  &  449 & 35  &  457 & 35$^{+1}$  &  461 & 36  \\
 \textbf{\emph{463}} & \textbf{\emph{36}}  &  \textbf{\emph{467}} & \textbf{\emph{36}}  &  \textbf{\emph{479}} & \textbf{\emph{36}} &  487 & 36  &  491 & 36$^{+1}$  &  499 & 37  \\
 503 & 37  &  509 & 37  &  512 & 36$^{+1}$  &\textbf{\emph{521}} & \textbf{\emph{37}}  &  \textbf{\emph{523}} & \textbf{\emph{38}}  &  529 & 38   \\
 541 & 38  &  \textbf{\emph{547}} & \textbf{\emph{38}}  &  557 & 39  &  563 & 39  &  569 & 39  &  571 & 39  \\
 577 & 39  &  587 & 39  &  593 & 39  &  599 & 40  &  601 & 40  &  607 & 40  \\
 613 & 40  &  617 & 40  &  619 & 40  &  625 & 41  &  631 & 41  &  641 & 41  \\
 643 & 41  &  647 & 41  &  653 & 41  &  659 & 41  &  661 & 41  &  673 & 41  \\
 677 & 42  &  683 & 42  &  691 & 42  &  701 & 42  &  709 & 43  &  719 & 43  \\
 727 & 42  &  729 & 40$^{+3}$  &  733 & 43  &  739 & 43  &  743 & 43  &  751 & 43  \\
 757 & 43  &  761 & 43  &  769 & 44  &  773 & 44  &  787 & 44  &  797 & 45  \\
 809 & 45  &  811 & 45  &  821 & 45  &  823 & 45  &  827 & 45  &  829 & 45  \\
 839 & 45$^{+1}$  &  \emph{841} & \emph{42}  &  853 & 45  &  857 & 46  &  859 & 46  &  863 & 46  \\
 877 & 46  &  881 & 46  &  883 & 46  &  887 & 46  &  907 & 47  &  911 & 47  \\
 919 & 47  &  929 & 47  &  937 & 47  &  941 & 48  &  947 & 48  &  953 & 48  \\
 961 & 48  &  967 & 48  &  971 & 48  &  977 & 48  &  983 & 48  &  991 & 48  \\
 997 & 48  &  1009 & 49  &  1013 & 49  &  1019 & 49  &  1021 & 49  &  1024 & 49  \\
 1031 & 49  &  1033 & 49  &  1039 & 49  &  1049 & 50  &  1051 & 49  &  1061 & 50  \\
 1063 & 49  &  1069 & 50  &  1087 & 50  &  1091 & 50  &  1093 & 50  &  1097 & 50  \\
 \hline
 \end{tabular}
 \newpage

\textbf{Table 4.} Continue 1 \medskip\\

 \renewcommand{\arraystretch}{1.00}
 \begin{tabular}
{@{}r@{\,}c@{\,}|r@{\,}c@{\,}|r@{\,}c@{\,}|r@{\,}c@{\,}|r@{\,}c@{\,}|r@{\,}c@{}}
 \hline
$q~~~$ & $\overline{n}_q(4,3)$ &
$q~~~$ & $\overline{n}_q(4,3)$ &
$q~~~$ & $\overline{n}_q(4,3)$ &
$q~~~$ & $\overline{n}_q(4,3)$ &
$q~~~$ & $\overline{n}_q(4,3)$ &
$q~~~$ & $\overline{n}_q(4,3)^{ \vphantom{H^{L}}}$ \\
 \hline
 1103 & 50  &  1109 & 51  &  1117 & 51  &  1123 & 51  &  1129 & 51  &  1151 & 51  \\
 1153 & 51  &  1163 & 51  &  1171 & 52  &  1181 & 52  &  1187 & 52  &  1193 & 52  \\
 1201 & 52  &  1213 & 52  &  1217 & 52  &  1223 & 52  &  1229 & 53  &  1231 & 53  \\
 1237 & 53  &  1249 & 52  &  1259 & 53  &  1277 & 53  &  1279 & 53  &  1283 & 53  \\
 1289 & 54  &  1291 & 54  &  1297 & 54  &  1301 & 54  &  1303 & 54  &  1307 & 54  \\
 1319 & 54  &  1321 & 54$^{+1}$  &  1327 & 55  &  1331 & 48$^{+7}$  &  1361 & 54  &  1367 & 54  \\
 1369 & 55  &  1373 & 55  &  1381 & 55  &  1399 & 55  &  1409 & 55  &  1423 & 55  \\
 1427 & 56  &  1429 & 56  &  1433 & 56  &  1439 & 56  &  1447 & 56  &  1451 & 56  \\
 1453 & 56  &  1459 & 56  &  1471 & 56  &  1481 & 57  &  1483 & 57  &  1487 & 57  \\
 1489 & 57  &  1493 & 57  &  1499 & 57  &  1511 & 57  &  1523 & 57$^{+1}$  &  1531 & 57  \\
 1543 & 57  &  1549 & 57  &  1553 & 58  &  1559 & 58  &  1567 & 57  &  1571 & 58  \\
 1579 & 58  &  1583 & 58  &  1597 & 58  &  1601 & 58  &  1607 & 58  &  1609 & 58  \\
 1613 & 58  &  1619 & 59  &  1621 & 59  &  1627 & 58  &  1637 & 58  &  1657 & 58  \\
 1663 & 59$^{+1}$  &  1667 & 59  &  1669 & 59  &  1681 & 59  &  1693 & 60  &  1697 & 59  \\
 1699 & 59  &  1709 & 60  &  1721 & 60  &  1723 & 60  &  1733 & 60$^{+1}$  &  1741 & 60  \\
 1747 & 60  &  1753 & 60  &  1759 & 60  &  1777 & 61  &  1783 & 61  &  1787 & 60  \\
 1789 & 61  &  1801 & 61  &  1811 & 61  &  1823 & 61  &  1831 & 61  &  1847 & 61  \\
 1849 & 62  &  1861 & 61  &  1867 & 61  &  1871 & 62  &  1873 & 62  &  1877 & 61  \\
 1879 & 62  &  1889 & 62  &  1901 & 62  &  1907 & 62  &  1913 & 62  &  1931 & 62  \\
 1933 & 63  &  1949 & 63  &  1951 & 63  &  1973 & 63  &  1979 & 63  &  1987 & 64  \\
 1993 & 64  &  1997 & 63  &  1999 & 63  &  2003 & 63  &  2011 & 63  &  2017 & 63  \\
 2027 & 63  &  2029 & 64  &  2039 & 64  &  2048 & 64  &  2053 & 64  &  2063 & 64  \\
 2069 & 64  &  2081 & 65  &  2083 & 65  &  2087 & 64  &  2089 & 64  &  2099 & 65  \\
 2111 & 65  &  2113 & 65  &  2129 & 65  &  2131 & 65  &  2137 & 66  &  2141 & 65  \\
 2143 & 65  &  2153 & 65  &  2161 & 65  &  2179 & 66  &  2187 & 66  &  2197 & 56$^{+10}$  \\
 2203 & 66  &  2207 & 66  &  2209 & 66  &  2213 & 66  &  2221 & 65  &  2237 & 67  \\
 2239 & 66  &  2243 & 67  &  2251 & 66  &  2267 & 66  &  2269 & 67  &  2273 & 66  \\
 2281 & 66  &  2287 & 66  &  2293 & 67  &  2297 & 67  &  2309 & 67  &  2311 & 68  \\
 2333 & 67  &  2339 & 68  &  2341 & 67  &  2347 & 68  &  2351 & 68  &  2357 & 68  \\
 2371 & 68  &  2377 & 68  &  2381 & 68  &  2383 & 68  &  2389 & 68  &  2393 & 68  \\
 2399 & 69  &  2401 & 68  &  2411 & 68  &  2417 & 69  &  2423 & 68  &  2437 & 69  \\
 2441 & 68  &  2447 & 68  &  2459 & 69  &  2467 & 69  &  2473 & 69  &  2477 & 69  \\
 2503 & 69  &  2521 & 70  &  2531 & 69  &  2539 & 70  &  2543 & 70  &  2549 & 69  \\
 2551 & 70  &  2557 & 70  &  2579 & 70  &  2591 & 70  &  2593 & 69  &  2609 & 70  \\
 2617 & 71  &  2621 & 70  &  2633 & 70  &  2647 & 70  &  2657 & 70  &  2659 & 70  \\
 2663 & 70  &  2671 & 70  &  2677 & 71  &  2683 & 71  &  2687 & 71  &  2689 & 71  \\
 \hline
 \end{tabular}
 \newpage

\textbf{Table 4.} Continue 2 \medskip

 \begin{tabular}
{@{}r@{\,}c@{\,}|r@{\,}c@{\,}|r@{\,}c@{\,}|r@{\,}c@{\,}|r@{\,}c@{\,}|r@{\,}c@{}}
 \hline
$q~~~$ & $\overline{n}_q(4,3)$ &
$q~~~$ & $\overline{n}_q(4,3)$ &
$q~~~$ & $\overline{n}_q(4,3)$ &
$q~~~$ & $\overline{n}_q(4,3)$ &
$q~~~$ & $\overline{n}_q(4,3)$ &
$q~~~$ & $\overline{n}_q(4,3)^{ \vphantom{H^{L}}}$ \\
 \hline
 2693 & 71  &  2699 & 72  &  2707 & 72  &  2711 & 71  &  2713 & 71  &  2719 & 71  \\
 2729 & 71  &  2731 & 71  &  2741 & 72  &  2749 & 72  &  2753 & 71  &  2767 & 72  \\
 2777 & 72  &  2789 & 72  &  2791 & 72  &  2797 & 73  &  2801 & 72  &  2803 & 72  \\
 2809 & 73  &  2819 & 73  &  2833 & 73  &  2837 & 73  &  2843 & 73  &  2851 & 72  \\
 2857 & 72  &  2861 & 73  &  2879 & 72  &  2887 & 72  &  2897 & 73  &  2903 & 73  \\
 2909 & 74  &  2917 & 73  &  2927 & 74  &  2939 & 73  &  2953 & 73  &  2957 & 73  \\
 2963 & 74  &  2969 & 74  &  2971 & 74  &  2999 & 74  &  3001 & 73  &  3011 & 75  \\
 3019 & 75  &  3023 & 75  &  3037 & 74  &  3041 & 74  &  3049 & 75  &  3061 & 75  \\
 3067 & 74  &  3079 & 75  &  3083 & 74  &  3089 & 75  &  3109 & 75  &  3119 & 75  \\
 3121 & 76  &  3125 & 76  &  3137 & 75  &  3163 & 76  &  3167 & 75  &  3169 & 75  \\
 3181 & 75  &  3187 & 75  &  3191 & 75  &  3203 & 76  &  3209 & 76  &  3217 & 77  \\
 3221 & 75  &  3229 & 76  &  3251 & 76  &  3253 & 77  &  3257 & 76  &  3259 & 77  \\
 3271 & 76  &  3299 & 77  &  3301 & 76  &  3307 & 77  &  3313 & 78  &  3319 & 77  \\
 3323 & 77  &  3329 & 77  &  3331 & 78  &  3343 & 77  &  3347 & 78  &  3359 & 78  \\
 3361 & 77  &  3371 & 78  &  3373 & 77  &  3389 & 78  &  3391 & 77  &  3407 & 77  \\
 3413 & 78  &  3433 & 78  &  3449 & 77  &  3457 & 78  &  3461 & 79  &  3463 & 78  \\
 3467 & 78  &  3469 & 78  &  3481 & 79  &  3491 & 78  &  3499 & 78  &  3511 & 78  \\
 3517 & 79  &  3527 & 79  &  3529 & 78  &  3533 & 79  &  3539 & 80  &  3541 & 79  \\
 3547 & 79  &  3557 & 79  &  3559 & 80  &  3571 & 79  &  3581 & 79  &  3583 & 79  \\
 3593 & 80  &  3607 & 79  &  3613 & 79  &  3617 & 80  &  3623 & 80  &  3631 & 80  \\
 3637 & 79  &  3643 & 80  &  3659 & 80  &  3671 & 80  &  3673 & 80  &  3677 & 81  \\
 3691 & 80  &  3697 & 80  &  3701 & 81  &  3709 & 80  &  3719 & 80  &  3721 & 81  \\
 3727 & 81  &  3733 & 80  &  3739 & 80  &  3761 & 81  &  3767 & 80  &  3769 & 80  \\
 3779 & 81  &  3793 & 81  &  3797 & 81  &  3803 & 81  &  3821 & 82  &  3823 & 80  \\
 3833 & 82  &  3847 & 81  &  3851 & 82  &  3853 & 82  &  3863 & 82  &  3877 & 82  \\
 3881 & 82  &  3889 & 81  &  3907 & 83  &  3911 & 83  &  3917 & 82  &  3919 & 83  \\
 3923 & 82  &  3929 & 83  &  3931 & 83  &  3943 & 82  &  3947 & 82  &  3967 & 83  \\
 3989 & 83  &  4001 & 83  &  4003 & 83  &  4007 & 83  &  4013 & 83  &  4019 & 83  \\
 4021 & 83  &  4027 & 83  &  4049 & 83  &  4051 & 83  &  4057 & 83  &  4073 & 83  \\
 4079 & 84  &  4091 & 83  &  4093 & 83  &  4096 & 68$^{+16}$  &  4099 & 84  &  4111 & 84  \\
 4127 & 83  &  4129 & 84  &  4133 & 83  &  4139 & 84  &  4153 & 84  &  4157 & 84  \\
 4159 & 84  &  4177 & 84  &  4201 & 84  &  4211 & 85  &  4217 & 85  &  4219 & 85  \\
 4229 & 84  &  4231 & 85  &  4241 & 85  &  4243 & 85  &  4253 & 85  &  4259 & 85  \\
 4261 & 85  &  4271 & 85  &  4273 & 85  &  4283 & 86  &  4289 & 86  &  4297 & 85  \\
 4327 & 86  &  4337 & 85  &  4339 & 85  &  4349 & 85  &  4357 & 86  &  4363 & 86  \\
 4373 & 86  &  4391 & 87  &  4397 & 87  &  4409 & 87  &  4421 & 85  &  4423 & 87  \\
 4441 & 86  &  4447 & 86  &  4451 & 86  &  4457 & 87  &  4463 & 88  &  4481 & 87  \\
 4483 & 88  &  4489 & 87  &  4493 & 88  &  4507 & 88  &  4513 & 88  &  4517 & 88  \\
 \hline
 \end{tabular}
 \newpage

\textbf{Table 4.} Continue 3 \medskip

 \renewcommand{\arraystretch}{0.9}
 \begin{tabular}
{@{}r@{\,}c@{\,}|r@{\,}c@{\,}|r@{\,}c@{\,}|r@{\,}c@{\,}|r@{\,}c@{\,}|r@{\,}c@{}}
 \hline
$q~~~$ & $\overline{n}_q(4,3)$ &
$q~~~$ & $\overline{n}_q(4,3)$ &
$q~~~$ & $\overline{n}_q(4,3)$ &
$q~~~$ & $\overline{n}_q(4,3)$ &
$q~~~$ & $\overline{n}_q(4,3)$ &
$q~~~$ & $\overline{n}_q(4,3)^{ \vphantom{H^{L}}}$ \\
 \hline
 4519 & 89  &  4523 & 89  &  4547 & 88  &  4549 & 89  &  4561 & 89  &  4567 & 89  \\
 4583 & 89  &  4591 & 88  &  4597 & 89  &  4603 & 88  &  4621 & 88  &  4637 & 88  \\
 4639 & 89  &  4643 & 89  &  4649 & 89  &  4651 & 89  &  4657 & 90  &  4663 & 89  \\
 4673 & 90  &  4679 & 88  &  4691 & 89  &  4703 & 89  &  4721 & 90  &  4723 & 90  \\
 4729 & 89  &  4733 & 89  &  4751 & 90  &  4759 & 90  &  4783 & 90  &  4787 & 89  \\
 4789 & 89  &  4793 & 89  &  4799 & 91  &  4801 & 89  &  4813 & 91  &  4817 & 89  \\
 4831 & 90  &  4861 & 89  &  4871 & 90  &  4877 & 90  &  4889 & 90  &  4903 & 91  \\
 4909 & 90  &  4913 & 72$^{+17}$  &  4919 & 90  &  4931 & 90  &  4933 & 91  &  4937 & 90  \\
 4943 & 91  &  4951 & 91  &  4957 & 90  &  4967 & 91  &  4969 & 91  &  4973 & 90  \\
 4987 & 90  &  4993 & 92  &  4999 & 92  &  5003 & 91  &  5009 & 92  &  5011 & 92  \\
 5021 & 91  &  5023 & 91  &  5039 & 91  &  5041 & 91  &  5051 & 91  &  5059 & 92  \\
 5077 & 91  &  5081 & 92  &  5087 & 92  &  5099 & 93  &  5101 & 92  &  5107 & 93  \\
 5113 & 93  &  5119 & 91  &  5147 & 92  &  5153 & 93  &  5167 & 92  &  5171 & 93  \\
 5179 & 93  &  5189 & 93  &  5197 & 93  &  5209 & 92  &  5227 & 92  &  5231 & 93  \\
 5233 & 93  &  5237 & 93  &  5261 & 93  &  5273 & 93  &  5279 & 94  &  5281 & 93  \\
 5297 & 93  &  5303 & 94  &  5309 & 93  &  5323 & 93  &  5329 & 94  &  5333 & 93  \\
 5347 & 94  &  5351 & 94  &  5381 & 94  &  5387 & 93  &  5393 & 95  &  5399 & 95  \\
 5407 & 94  &  5413 & 94  &  5417 & 94  &  5419 & 95  &  5431 & 94  &  5437 & 93  \\
 5441 & 94  &  5443 & 93  &  5449 & 93  &  5471 & 94  &  5477 & 94  &  5479 & 95  \\
 5483 & 95  &  5501 & 95  &  5503 & 95  &  5507 & 94  &  5519 & 96  &  5521 & 95  \\
 5527 & 95  &  5531 & 95  &  5557 & 94  &  5563 & 95  &  5569 & 95  &  5573 & 95  \\
 5581 & 94  &  5591 & 96  &  5623 & 95  &  5639 & 95  &  5641 & 96  &  5647 & 96  \\
 5651 & 95  &  5653 & 96  &  5657 & 95  &  5659 & 96  &  5669 & 95  &  5683 & 96  \\
 5689 & 96  &  5693 & 97  &  5701 & 96  &  5711 & 96  &  5717 & 97  &  5737 & 96  \\
 5741 & 95  &  5743 & 95  &  5749 & 97  &  5779 & 96  &  5783 & 96  &  5791 & 96  \\
 5801 & 94  &  5807 & 96  &  5813 & 96  &  5821 & 97  &  5827 & 97  &  5839 & 96  \\
 5843 & 97  &  5849 & 96  &  5851 & 97  &  5857 & 97  &  5861 & 97  &  5867 & 97  \\
 5869 & 98  &  5879 & 97  &  5881 & 97  &  5897 & 97  &  5903 & 97  &  5923 & 97  \\
 5927 & 97  &  5939 & 97  &  5953 & 98  &  5981 & 97  &  5987 & 98  &  6007 & 98  \\
 6011 & 98  &  6029 & 97  &  6037 & 98  &  6043 & 99  &  6047 & 98  &  6053 & 98  \\
 6067 & 99  &  6073 & 99  &  6079 & 98  &  6089 & 99  &  6091 & 97  &  6101 & 98  \\
 6113 & 99  &  6121 & 99  &  6131 & 98  &  6133 & 97  &  6143 & 98  &  6151 & 98  \\
 6163 & 99  &  6173 & 99  &  6197 & 100  &  6199 & 100  &  6203 & 98  &  6211 & 100  \\
 6217 & 99  &  6221 & 100  &  6229 & 99  &  6241 & 100  &  6247 & 99  &  6257 & 100  \\
 6263 & 100  &  6269 & 100  &  6271 & 100  &  6277 & 98  &  6287 & 100  &  6299 & 100  \\
 6301 & 99  &  6311 & 99  &  6317 & 100  &  6323 & 100  &  6329 & 100  &  6337 & 101  \\
 6343 & 100  &  6353 & 100  &  6359 & 100  &  6361 & 99  &  6367 & 101  &  6373 & 100  \\
 6379 & 100  &  6389 & 101  &  6397 & 101  &  6421 & 101  &  6427 & 101  &  6449 & 101  \\
 6451 & 101  &  6469 & 100  &  6473 & 101  &  6481 & 101  &  6491 & 101  &  6521 & 101  \\
 6529 & 100  &  6547 & 102  &  6551 & 102  &  6553 & 101  &  6561 & 102  &  6563 & 100  \\
 6569 & 101  &  6571 & 102  &  6577 & 101  &  6581 & 101  &  6599 & 101  &  6607 & 101 \\
 \hline
 \end{tabular}
 \newpage

\textbf{Table 4.} Continue 4 \medskip

 \renewcommand{\arraystretch}{0.9}
 \begin{tabular}
{@{}r@{\,}c@{\,}|r@{\,}c@{\,}|r@{\,}c@{\,}|r@{\,}c@{\,}|r@{\,}c@{\,}|r@{\,}c@{}}
 \hline
$q~~~$ & $\overline{n}_q(4,3)$ &
$q~~~$ & $\overline{n}_q(4,3)$ &
$q~~~$ & $\overline{n}_q(4,3)$ &
$q~~~$ & $\overline{n}_q(4,3)$ &
$q~~~$ & $\overline{n}_q(4,3)$ &
$q~~~$ & $\overline{n}_q(4,3)^{ \vphantom{H^{L}}}$ \\
 \hline
 6619 & 101  &  6637 & 102  &  6653 & 102  &  6659 & 102  &  6661 & 102  &  6673 & 102  \\
 6679 & 103  &  6689 & 102  &  6691 & 102  &  6701 & 102  &  6703 & 100  &  6709 & 102  \\
 6719 & 102  &  6733 & 102  &  6737 & 102  &  6761 & 101  &  6763 & 102  &  6779 & 103  \\
 6781 & 103  &  6791 & 102  &  6793 & 103  &  6803 & 102  &  6823 & 103  &  6827 & 103  \\
 6829 & 102  &  6833 & 103  &  6841 & 101  &  6857 & 103  &  6859 & 102  &  6863 & 102  \\
 6869 & 103  &  6871 & 103  &  6883 & 104  &  6889 & 102  &  6899 & 103  &  6907 & 103  \\
 6911 & 104  &  6917 & 103  &  6947 & 102  &  6949 & 103  &  6959 & 103  &  6961 & 103  \\
 6967 & 104  &  6971 & 102  &  6977 & 103  &  6983 & 103  &  6991 & 104  &  6997 & 103  \\
 7001 & 104  &  7013 & 104  &  7019 & 104  &  7027 & 105  &  7039 & 104  &  7043 & 103  \\
 7057 & 105  &  &  &  &  &  &  &  &  &  &  \\
 \hline
 \end{tabular}
 \newpage

 \textbf{Table 5.}  Lengths $\overline{n}_q(5,3)$ of the shortest \textbf{\emph{known}} $[\overline{n}_q(5,3),\overline{n}_q(5,3)-5]_q3$ codes  obtained by the leximatrix and Rand-Greedy algorithms;
 $\overline{n}_q(5,3)=\min\{n^\text{L}_q(5,3),n^\text{G}_q(5,3)\}$ for $3\leq q\leq 401$; $\overline{n}_q(5,3)=\overline{n}_q(5,3,5)=n^\text{L}_q(5,3)$ for $401<q\le839$. The improvements of code length in comparison with \cite[Tab.\,2]{DavOst-DESI2010} are noted in bold italic font. The cases $\ell_q(5,3)=\overline{n}_q(5,3)$ are noted by the subscript~``$\bullet$'' \medskip

\renewcommand{\arraystretch}{0.93}
\noindent \begin{tabular}
{@{}r@{\,}c@{\,}|r@{\,}c@{\,}|r@{\,}c@{\,}|r@{\,}c@{\,}|r@{\,}c@{\,}|r@{\,}c@{\,}|r@{\,}c@{}}
 \hline
$q~~~$ & $\overline{n}_q(5,3)$ &
$q~~~$ & $\overline{n}_q(5,3)$ &
$q~~~$ & $\overline{n}_q(5,3)$ &
$q~~~$ & $\overline{n}_q(5,3)$ &
$q~~~$ & $\overline{n}_q(5,3)$ &
$q~~~$ & $\overline{n}_q(5,3)$ &
$q~~~$ & $\overline{n}_q(5,3)^{ \vphantom{H^{L}}}$ \\
 \hline
 3 & 8$_\bullet$  &  4 & 9$_\bullet$  &  5 & 10$_\bullet$  &  7 & 13  &  8 & 14  &  9 & 16  &  11 & 18  \\
 13 & 21  &  16 & 24  &  17 & 25  &  19 & 27  &  23 & 32  &  25 & 34  &  27 & 36  \\
 29 & 38  &  31 & 40  &  32 & 41  &  \textbf{\emph{37}} & \textbf{\emph{47}}  &  \textbf{\emph{41}} & \textbf{\emph{51}}  &  \textbf{\emph{43}} & \textbf{\emph{52}}  &  47 & 56  \\
 49 & 58  &  53 & 61  &  59 & 67  &  61 & 68  &  64 & 71  &  67 & 73  &  71 & 76  \\
 73 & 78  &  79 & 83  &  81 & 85  &  83 & 85  &  89 & 90  &  97 & 96  &  101 & 99  \\
 103 & 100  &  107 & 103  &  109 & 104  &  113 & 108  &  121 & 113  &  125 & 116  &  127 & 117  \\
 128 & 119  &  131 & 119  &  137 & 123  &  139 & 125  &  149 & 131  &  151 & 132  &  157 & 136  \\
 163 & 140  &  167 & 142  &  169 & 145  &  173 & 146  &  179 & 150  &  181 & 151  &  191 & 157  \\
 193 & 158  &  197 & 161  &  199 & 162  &  211 & 169  &  223 & 177  &  227 & 179  &  229 & 180  \\
 233 & 183  &  239 & 185  &  241 & 188  &  243 & 188  &  251 & 193  &  256 & 195  &  257 & 197  \\
 263 & 200  &  269 & 203  &  271 & 204  &  277 & 208  &  281 & 209  &  283 & 211  &  289 & 213  \\
 293 & 216  &  307 & 224  &  311 & 226  &  313 & 227  &  317 & 230  &  331 & 237  &  337 & 240  \\
 343 & 243  &  347 & 245  &  349 & 245  &  353 & 248  &  359 & 251  &  361 & 254  &  367 & 256  \\
 373 & 261  &  379 & 262  &  383 & 264  &  389 & 267  &  397 & 273  &  401 & 274  &  409 & 284 \\
 419 & 292  &  421 & 290  &  431 & 297  &  433 & 299  &  439 & 301  &  443 & 304  &  449 & 309  \\
 457 & 311  &  461 & 311  &  463 & 309  &  467 & 314  &  479 & 320  &  487 & 324  &  491 & 324  \\
 499 & 328  &  503 & 330  &  509 & 334  &  512 & 334  &  521 & 339  &  523 & 341  &  529 & 344  \\
 541 & 348  &  547 & 349  &  557 & 353  &  563 & 360  &  569 & 364  &  571 & 362  &  577 & 365  \\
 587 & 371  &  593 & 374  &  599 & 375  &  601 & 376  &  607 & 376  &  613 & 380  &  617 & 384  \\
 619 & 382  &  625 & 385  &  631 & 387  &  641 & 393  &  643 & 398  &  647 & 396  &  653 & 399  \\
 659 & 402  &  661 & 401  &  673 & 407  &  677 & 407  &  683 & 411  &  691 & 416  &  701 & 417  \\
 709 & 424  &  719 & 427  &  727 & 430  &  729 & 430  &  733 & 429  &  739 & 431  &  743 & 436  \\
 751 & 439  &  757 & 440  &  761 & 443  &  769 & 447  &  773 & 450  &  787 & 453  &  797 & 458  \\
 809 & 464  &  811 & 464  &  821 & 467  &  823 & 468  &  827 & 471  &  829 & 473  &  839 & 475  \\
  \hline
 \end{tabular}
 \end{document}